\documentclass[a4paper,12pt]{amsart}
\usepackage[left=1.2in, right=1.2in, bottom=1.2in, top=1.2in]{geometry}

\usepackage[utf8]{inputenc}

\usepackage{amsmath, amssymb}
\usepackage{mathtools}
\usepackage{bm}

\usepackage{amsaddr}

\usepackage[all]{xy}

\usepackage{csquotes}
\usepackage{enumerate}

\usepackage{graphicx}

\usepackage{caption}
\usepackage{subcaption} 

\usepackage{color}
\usepackage{mwe}       
\usepackage{appendix}
\usepackage{arydshln}

\usepackage{multirow}
\usepackage{rotating}

\usepackage[linesnumbered,ruled]{algorithm2e}
\usepackage{comment}
\usepackage{cite}

\usepackage{url}

\usepackage{thmtools, thm-restate}

\usepackage{soul}
\usepackage{xcolor}
\newcommand\independent{\protect\mathpalette{\protect\independenT}{\perp}}
\def\independenT#1#2{\mathrel{\rlap{$#1#2$}\mkern2mu{#1#2}}}

\usepackage{tikz}
\usetikzlibrary{arrows.meta,shapes}
\usetikzlibrary{arrows,shapes.arrows,shapes.geometric,shapes.multipart,
decorations.pathmorphing,positioning,shapes.swigs,}
\usepackage{setspace}

\newtheorem{lemma}{Lemma}
\newtheorem{remark}{Remark}
\newtheorem{theorem}{Theorem}

\newtheorem{definition}{Definition}

\newtheorem{assumption}{Assumption}

\MakeOuterQuote{"}\EnableQuotes
\DeclareUnicodeCharacter{00A0}{ }

\title[]{On the choice of using raw or demographically-corrected scores}

\author{Ignacio Gonz\'alez-P\'erez, Mats Julius Stensrud, Marco Piccininni} \address{ \small  Institute of Mathematics, \'Ecole Polytechnique Fédérale de Lausanne, Switzerland
}
\allowdisplaybreaks
\DeclareMathOperator{\expit}{expit}
\DeclareMathOperator{\logit}{logit}

\DeclareMathOperator{\diag}{diag}
\begin{document}
\newcommand{\cor}{\overline{c}=0,\overline{r}=1}
\newcommand{\defeq}{:=}

\begin{abstract}
Demographic corrections are routinely performed in many disciplines, including psychology. Yet, there are ongoing debates about whether these corrections are appropriate and improve classification accuracy. Here, we focus on cognitive screening tests, and show that common demographic corrections, like the z-score correction, can be detrimental for classification in some settings. Formally, we present  sufficient conditions ensuring that  raw scores outperform the demographically-corrected ones, and give a substantive interpretation of this result. We also investigate the claim that using demographically-corrected scores results in more fair decisions compared to using raw scores. We apply our results to the Mini-Mental State Examination in the OASIS-3 dataset. 
\end{abstract}

\maketitle


\section{Introduction}\label{sec:Intro}

Demographic corrections are widely used in the health sciences, particularly in psychology. For example, tests assessing age-related cognitive impairment are usually corrected for age and education \cite{Petersen2003-na, Quaranta2016-cn, Mitrushina2005-gb}. Such corrections are also used in the clinical definitions of cognitive impairment \cite{Bradfield2020-ae}, such as the DSM-5 definition of Mild Neurocognitive Disorder \cite{American_Psychiatric_Association2013-ax} and the older Mild Cognitive Impairment criteria \cite{Petersen1997-od}.

The consequences of demographic corrections are particularly relevant to  screening for cognitive impairment. Cognitive screening tests are used to decide whether an individual should be referred to further clinical assessments \cite{Cullen2007-qt, Belle1996-ws, o2004correcting}. The validity of these tests is usually assessed by their classification accuracy, i.e., sensitivity, specificity, and area under the ROC curve, in categorizing individuals against the gold-standard clinical diagnoses of mild cognitive impairment and/or dementia \cite{Cullen2007-qt, Larner2016-si, Mitrushina2005-gb, kraemer1998adjusting, o2004correcting, Larouche2016-fd, Gregory2004-ic}.

Cognitive screening tests are designed to give a numerical result that summarizes the performance of the individual who undertakes the evaluation. This result, e.g., representing the number of correct answers or a rating, is commonly referred to as the “raw” test score \cite{Gregory2004-ic, Mitrushina2005-gb}.
Sometimes, a demographically-corrected score is preferred over the raw score \cite{magni1996mini, Larner2016-si, Larouche2016-fd, Mitrushina2005-gb, Gregory2004-ic}. Specifically, instead of interpreting raw scores directly, practitioners often draw conclusions by comparing the raw score to reference norms of individuals without cognitive impairment and with similar demographic characteristics \cite{Crum1993-hk, Mitrushina2005-gb, sliwinski1997effect, Larouche2016-fd, Larner2016-si}. In practice, this comparison is made by numerically transforming the raw score into a “corrected” version; this process is known as adjustment \cite{berkman1986association, Larner2016-si, magni1996mini} or correction \cite{o2010age}. The most widespread type of demographic correction is the “z-score” correction \cite{Mitrushina2005-gb, Larouche2016-fd}.

While demographic corrections are ubiquitous in neuropsychology, empirical results suggest that demographically-corrected scores have lower classification accuracy compared to raw scores. For example, a large 1996 study of older persons found that demographic corrections did not improve classification accuracy for dementia compared to using raw scores \cite{Belle1996-ws}. Other empirical studies reported no benefit of demographic corrections on the classification accuracy of the Mini-Mental State Examination (MMSE) \cite{kraemer1998adjusting} and the modified 3MS \cite{o2004correcting} screening tests. In these studies, accuracy was measured with respect to clinical diagnoses. Similar results were found when assessing the prognostic accuracy of the test scores in longitudinal studies \cite{hessler2014age, Quaranta2016-cn}, or when assessing the association between unadjusted or adjusted criteria and dementia biomarkers \cite{Hemminghyth2025}. 

Some researchers have advocated the claim that the suboptimal accuracy of demographically-corrected scores can be explained by the fact that demographic characteristics influence not only the raw scores but also the probability of the condition of interest (e.g., cognitive impairment). Already in 1986, \cite{berkman1986association} used this argument to question the use of correction for age and education in cognitive screening. A decade later,  \cite{sliwinski1997effect} claimed that since age and education were both risk factors for dementia, eliminating their effect on the raw test scores resulted in lower classification accuracy of the test, which was later supported by simulations and real data analyses \cite{o2010age}. More recently, it has been shown that, under a simplified causal model and strong parametric assumptions, corrected scores have a lower area under the ROC curve compared to raw scores, if age and education have sufficiently strong effects on both the test performance and the risk of having cognitive impairment \cite{piccininni2023should}. Simulations have suggested that using raw scores leads to higher marginal sensitivity (or specificity)  for the same specificity (or sensitivity) level, while using demographically-corrected scores leads to more homogeneous performance across demographic groups \cite{wischmann2024consequences}.

These empirical and theoretical results have, however, been met with scepticism by practitioners. Some have questioned whether it is true that raw scores have better classification accuracy compared to corrected scores, because of the small differences found in practice and the inevitable uncertainty of finite sample inference \cite{ilardi2025no, Hemminghyth2025}. Also, the existing simulation results have been questioned due to their strong assumptions, and the fact that there is no guarantee that the simulated data are realistic \cite{ilardi2025no}. Finally, some researchers have informally argued that corrected scores should be preferred over raw scores even if their classification accuracy is lower, because in some sense they are more "fair" and avoid overdiagnosis in vulnerable demographic groups \cite{ilardi2025no, Aiello2025}.

The aim of this work is to bring clarity to this debate. We will show, without relying on restrictive data generating models, parametric assumptions, or empirical results, that it is possible for raw scores to have better classification accuracy compared to corrected scores. Specifically, we will present sufficient conditions for the superiority in terms of classification accuracy of the raw scores over the corrected ones, and give an interpretation of this result.  In addition, we will engage with the claims that corrected scores are more fair. In this context, we formalize the notion of a score being insensitive to demographics, and assess whether demographically-corrected scores fulfil this property. Finally, we apply our results to the OASIS-3 dataset \cite{LaMontagne2019OASIS3}. In this dataset, we find evidence for the superiority of raw MMSE scores in terms of classification accuracy and no evidence of demographic sensitivity for age-corrected scores.

\section{Comparison of classification accuracy}\label{sec:ComparingScores}

\subsection{Notation and setup}\label{sec:Setup}

Consider a screening setting designed to assess whether an individual has a condition of interest ($D=1$) or not ($D=0$). As an example, let $D$ indicate the presence of cognitive impairment, possibly according to a clinical examination or some biomarkers. Let $X$ indicate the raw score resulting from the administration of a cognitive screening test, for example the MMSE \cite{Folstein1975}, with support on a subset of the real line, $\mathcal{X} \subseteq \mathbb{R}$. The variable $V$ represents  the demographic characteristics of the individual.  To simplify the exposition, assume that $V$ is a univariate covariate supported on (a subset of) the real line, $\mathcal{V}$. For example, $V$ could represent age, education, sex, or age-education stratum. We  present our results in full generality in Appendix \ref{sec:AppComparingGeneral}. For simplicity, we will assume that  the support of $(X,V)\mid D=d$ and of $X|V=v,D=d$ do not depend on $v,d$, avoiding positivity issues. In practice, this assumption can be relaxed, as we discuss in Section \ref{sec:OASIS3_Superiority}.

Let $Z$ denote the \textit{corrected score}. More formally, we define the demographic correction as follows.

\begin{definition}[Correction of the raw score]\label{def:Correction}
    Let $\mathcal{F}\subseteq\{F\colon \mathbb{R}\to[0,1] \text{ a c.d.f.}\}$, and let $T\colon\mathcal{F}\times\mathcal{V}\times\mathbb{R}\to\mathbb{R}$ be a map such that $T(F,v,\cdot)$ is strictly increasing and differentiable for all $F\in\mathcal{F}$ and all $v\in\mathcal{V}$. Then the corrected score is constructed as $Z=g_V(X)$ almost surely, with $g_v(x)=T(F_{X,0}(\cdot\mid v),v,x)$ for all $v$ in $\mathcal{V}$.
\end{definition}

Here, $F_{X,d}(\cdot\mid v)$ is  the c.d.f.\ of $X\mid D=d,V=v$ (following \cite{krzanowski2009roc}). While Definition \ref{def:Correction} is fairly general, most corrections found in practice either  make use only of $v$ or of $F_{X,0}(\cdot\mid v)$. Thus, as they will show slightly different behaviours, we formally distinguish them. 
\begin{definition}[Correction types]\label{def:CorrectionTypes}
    Consider a correction of the raw score defined by a map $T$ as per Definition \ref{def:Correction}. We say that this correction is:
    \begin{itemize}
        \item of Type I if there exists $\Tilde{T}\colon\mathcal{F}\times \mathbb{R}\mapsto\mathbb{R}$, such that $T(F,v,x)=\Tilde{T}(F,x)$ for all $(F,v,x)$ in $\mathcal{F}\times \mathcal{V}\times\mathbb{R}$.
        \item of Type II if there exists ${T'}\colon\mathcal{V}\times \mathbb{R}\mapsto\mathbb{R}$, such that $T(F,v,x)={T'}(v,x)$ for all $(F,v,x)$ in $\mathcal{F}\times \mathcal{V}\times\mathbb{R}$.
        \item of Type III if it is neither of Type I nor II.
    \end{itemize}
\end{definition}
Unless explicitly stated otherwise, we are agnostic to the type of correction we consider. To fix ideas, let $X$ be the MMSE score and let $V$ represent age-education stratum. A common age-education correction consists in calculating the ``z-score'' \cite{Mitrushina2005-gb, Larouche2016-fd} 
\begin{equation}\label{eq:AgeEduCorrection}
    Z=\frac{X-\mathbb{E}[X\mid V,D=0]}{\sqrt{Var[X\mid V,D=0]}},
\end{equation}
that is, we subtract from the observed raw score $X$ the expected score from a normative sample (without impairment, $D=0$) with the same demographic characteristics $V$, and divide by the standard deviation of the normative sample \cite{Mitrushina2005-gb, Petersen2003-na}.
We see how the ``z-score'' correction \eqref{eq:AgeEduCorrection} is a Type I correction with $T(F,v,x)=(x-\mu(F))/\sqrt{Var(F)}$ and $\mathcal{F}$ would be the set of univariate c.d.f.'s with positive finite variance. The percentile transformation, the "t-score", and other standardized scores \cite{Gregory2004-ic, Mitrushina2005-gb} are also Type I corrections. The normalized standard score \cite{Gregory2004-ic} is also a Type I correction with  $T(F,v,x)=\Phi^{-1}(F(x))$, where $\Phi$ is the  standard normal c.d.f.\ and $F$ has a continuous density w.r.t.\ the Lebesgue measure, and its strictly increasing. Nevertheless, not all corrections used in practice are of Type I. For example, the MoCA test is sometimes corrected for education status by adding 1 point if the individual has 12 years of education or less \cite{Nasreddine2005}. This is a  Type II  correction defined of the form $T(F,v,x)=x+I(v\leq 12)$, where $I(\cdot)$ denotes the indicator function, and $v$ years of formal education.

We adopt the convention that the raw score $X$ is constructed in such a way that smaller scores are associated with a greater risk of impairment. This convention aligns with several widely adopted cognitive screening tests such as the MMSE \cite{Folstein1975} and the MoCA \cite{Nasreddine2005}.
We therefore consider classification rules of the type
\begin{equation}\label{eq:BasicDecisionRule}
    I(X\leq t)
\end{equation}
and we adopt the same directionality for the corrected score as well. The threshold $t$ is often chosen to achieve a certain misclassification rate amongst the non-impaired, commonly referred to as False-Positive Rate $$FPR(t)\defeq\mathbb{P}(X\leq t\mid D=0)=F_{X,0}(t),$$ where  we denote by $F_{X,d}$ the cumulative distribution function of $X\mid D=d$. The FPR is also referred to as the false-alarm rate or one minus the specificity. The counterpart of the FPR is the probability of correctly classifying an individual as impaired: the True-Positive Rate $$TPR(t)\defeq\mathbb{P}(X\leq t\mid D=1)=F_{X,1}(t),$$ also referred to as the hit rate or sensitivity. 
If we reframe the classification problem as testing the hypothesis $\mathcal{H}_0\colon D=0$ against $\mathcal{H}_1\colon D=1$ \cite[Ch. 2.3.2]{krzanowski2009roc}, the false- and true-positive rates correspond to the  type-I error and  the power (one minus type-II error), respectively.

To understand the performance of the diagnostic test, it is common to plot the pairs of false- and true-positive rates for all possible thresholds, that is, the \textit{Receiver Operating Characteristic} (ROC) curve \cite{green1966signal} 
\begin{equation}\label{eq:ROCset}
    \{(FPR(t), TPR(t))\colon t\in\mathbb{R}\}\subset [0,1]^2.
\end{equation}
This curve is always contained in $[0,1]^2$, tracing between (0,0) and (1,1), which are associated with the threshold behaviour $t\to -\infty$ and $t\to \infty$  respectively.

Rather than parametrizing the ROC curve w.r.t. the classification threshold, we can consider the TPR corresponding to a  FPR $\alpha\in (0,1)$. If $X|D$ has an absolutely continuous distribution w.r.t.\ the Lebesgue measure, the ROC is continuous on $[0,1]$ as we have   $ROC_X(\alpha)=F_{X,1}(F_{X,0}^{-1}(\alpha))$ for all $\alpha$ in $(0,1)$, where we  define the generalized inverse of a c.d.f.\ $F$ as $F^{-1}(\alpha)\defeq\inf\{x\in\mathbb{R}\colon F(x)\geq\alpha\}$. However, if $X|D$ has a discrete distribution, as is the case for example of the MMSE, the ROC set \eqref{eq:ROCset} represents a ``disconnected'' set of points. In this case the ROC set does not represent a \textit{curve}, in the sense that for some point one can find an $\epsilon$-ball centred on it not containing any other points of the ROC set. It is not possible in this case to attain exactly every FPR $\alpha \in (0,1)$ with a decision rule as the one in Equation \eqref{eq:BasicDecisionRule}.

To overcome this issue, and propose a general framework which allows us to remain agnostic about the distribution of the score, for $\alpha\in(0,1)$ we consider the randomized decision rule 
\begin{equation}\label{eq:RandomizedDecision}
    \psi_\alpha(X)\defeq I(X<F_{X,0}^{-1}(\alpha)) + \underbrace{\frac{\alpha-\mathbb{P}(X<F_{X,0}^{-1}(\alpha)|D=0)}{\mathbb{P}(X=F_{X,0}^{-1}(\alpha)|D=0)}}_{=:\rho_\alpha(X)}    I(X=F_{X,0}^{-1}(\alpha)).
\end{equation}
Equation \eqref{eq:RandomizedDecision} can be interpreted as a randomized decision rule: classifying a unit as positive with probability $\psi_\alpha(X)$ corresponds to follow the rule ``when $X<F_{X,0}^{-1}(\alpha)$ classify as positive, when $X=F_{X,0}^{-1}(\alpha)$ classify as positive with probability $\rho_\alpha(X)$''. When $X$ has an absolutely continuous distribution, we adopt the convention that $\rho_\alpha(X)=1$ and therefore $\psi_\alpha(X)=I(X\leq F_{X,0}^{-1}(\alpha))$. In addition, even in the discrete case, if $\alpha$ is chosen such that $F_{X,0}(F_{X,0}^{-1}(\alpha))=\alpha$, then $\rho_\alpha(X)=1$ and again $\psi_\alpha(X)=I(X\leq F_{X,0}^{-1}(\alpha))$ like in the absolutely continuous case. The choice of such an $\alpha$ is equivalent to choosing a threshold $t\in\text{supp}(X)$, like in the non-randomized rule \eqref{eq:BasicDecisionRule}, and then considering $\alpha= F_{X,0}(t)$.

Because by construction $\mathbb{E}[\psi_\alpha(X)|D=0]=\alpha$, we say that  $\psi_\alpha(X)$ has FPR $\alpha$. We define the ROC curve at FPR $\alpha\in(0,1)$ as the TPR (power) of this decision rule 
\begin{align*}
    ROC_X(\alpha)&\defeq\mathbb{E}[\psi_\alpha(X)|D=1]\\&=\mathbb{P}(X<F_{X,0}^{-1}(\alpha)|D=1)+\rho_\alpha(X)\mathbb{P}(X=F_{X,0}^{-1}(\alpha)|D=1)\\&=\mathbb{P}(X<F_{X,0}^{-1}(\alpha)|D=1)\\&\quad+\frac{\mathbb{P}(X=F_{X,0}^{-1}(\alpha)|D=1)}{\mathbb{P}(X=F_{X,0}^{-1}(\alpha)|D=0)}  (\alpha-\mathbb{P}(X<F_{X,0}^{-1}(\alpha)|D=0))\\\text{(in the continuous case) }&=F_{X,1}(F_{X,0}^{-1}(\alpha)).
\end{align*}
The ROC curve when $X$ is discrete corresponds therefore to a linear interpolation of the ROC set points \eqref{eq:ROCset}, and the slope of the interpolated ROC at $\alpha$ (when not at one point of the ROC set) is the marginal likelihood ratio at $F_{X,0}^{-1}(\alpha)$ \cite{gneiting2022receiver}. Constructing the ROC curve as a linear interpolation of points \cite{fawcett2006introduction, gneiting2022receiver} generalizes results in the continuous case, such as the probabilistic interpretation of the area under the ROC curve or the connection between slopes of the ROC curve and likelihood ratios \cite{muschelli2020roc,fawcett2006introduction} (see Lemma \ref{lemma:SuperGradient} in Appendix \ref{sec:AppProofs}).

The rule $\psi_\alpha(\cdot)$ can be seen as a randomized test for $\mathcal{H}_0\colon D=0$ \cite[Ch. 2.3.2]{krzanowski2009roc}. In practice, most cognitive screening test, like the MMSE, are used in a non-randomized (hard thresholding) way, like Equation \eqref{eq:BasicDecisionRule}. We introduce the randomized rule in Equation \eqref{eq:RandomizedDecision} as an auxiliary rule which allows us to study the ROC curves regardless of the support of the score they are based on.

Global performance measures of a diagnostic score, which do not depend on the choice of a threshold or FPR, can be obtained by summarizing the ROC curve.  The most common of these measures is the area under the ROC curve (AUC) \cite[Ch. 2.4.1]{krzanowski2009roc}, defined as $$AUC(X)=\int_0^1ROC_X(\alpha)d\alpha.$$ The AUC is upper-bound by 1, with this value representing a perfect classifier. Hence, higher AUC values indicate better classification accuracy. 

When demographic covariates are available, we can also define covariate-specific ROC curves \cite{janes2009adjusting} in a stratum of $V$,  $$ROC_{X\mid V}(\alpha\mid v)=\mathbb{E}[\psi_\alpha(X|v)|D=1,V=v],$$ where 
\begin{align*}
    \psi_\alpha(X|v)\defeq& I(X<F_{X,0}^{-1}(\alpha|v)) \\&+ {\frac{\alpha-\mathbb{P}(X<F_{X,0}^{-1}(\alpha|v)|D=0,V=v)}{\mathbb{P}(X=F_{X,0}^{-1}(\alpha|v)|D=0,V=v)}}    I(X=F_{X,0}^{-1}(\alpha|v)).
\end{align*}
We use the same formalism for the construction of ROC curves of the corrected score.
\subsection{Main Results}

The main goal of this section is to compare the accuracy of classification rules based on raw and corrected scores. 
To  make a fair comparison, we compare marginal ROC curves of raw and corrected scores at an arbitrary value $\alpha$ of the FPR.
Differently from the existing literature handling covariates and ROC curves \cite{pepe1997regression,pepe1998three,cai2002semiparametric,GonzalezManteiga2011ROC,huang2013logistic,sewak2022transformation,inacio2021robust}, we do not focus on how to ``optimally'' adjust for covariates, but rather we study how the common ways in which covariates are accounted for in demographic corrections affect classification accuracy.

Under  Definition \ref{def:Correction}, classification based on raw or corrected scores conditional on a certain $V$-stratum is equivalent, in the sense that covariate-specific ROC curves  are identical, as we can see in the following lemma.
 
\begin{restatable}{lemma}{lemmaEqualCovariateSpecificRocs}\label{lemma:EqualCovariateSpecificRocs}
    Under Definition \ref{def:Correction}, covariate-specific ROC curves for raw and corrected scores are identical, in the sense that $ROC_{Z\mid V}(\alpha\mid v)=ROC_{X\mid V}(\alpha\mid v)$ for all $\alpha\in (0,1)$ and $v$ in $\mathcal{V}$.
\end{restatable}

For a proof see Appendix \ref{sec:AppProofs_ComparingScores}. Following Lemma \ref{lemma:EqualCovariateSpecificRocs}, to simplify notation, we denote $ROC_{X\mid V}(\alpha\mid v)$ by $ROC_v(\alpha)$ for all $v$ in the support of $V$ and all $\alpha$ in (0,1). 
Another property implied by Definition \ref{def:Correction}, which we also prove in Appendix \ref{sec:AppProofs_ComparingScores}, relates to a special case where raw scores are conditionally independent from the covariates.

\begin{restatable}{lemma}{lemmaIndependenceImplication}\label{lemma:IndependenceImplication}
    If $X\independent V\mid D$ and the corrected score is obtained under a Type I correction as per Definitions \ref{def:Correction} and \ref{def:CorrectionTypes}, then $Z\independent V\mid D$.
\end{restatable}
This lemma states that if the demographic covariates are not associated with the raw score among impaired and non-impaired, then performing a Type I correction  does not introduce a dependence with the demographic covariates. The reason why  Lemma \ref{lemma:IndependenceImplication} requires a type I correction is to ensure that the corrected score is only affected by $V$ through the covariate-specific distribution of the raw score amongst the non-impaired. This is the case for all the corrections introduced before, except for the simple (additive) education-adjustment of the MoCA test. 
\begin{remark}\label{remark:ReciprocateNotTrue}
    Under  Type I correction, $Z\independent V\mid D$ does not imply $X\independent V\mid D$; i.e., the reciprocate of Lemma \ref{lemma:IndependenceImplication} is not true. Take for example the following additive model  $X\stackrel{d}{=}h(V)+\beta D+\varepsilon$, with $\varepsilon\sim\mathcal{N}(0,\sigma^2)$ independent of $V$ and $D$, $\sigma>0$,  $\beta\in\mathbb{R}$ and $h\colon\mathcal{V}\to\mathbb{R}$ measurable. Then for the transformation in Equation \eqref{eq:AgeEduCorrection} we have $Z\stackrel{d}{=}(\beta D+\varepsilon)/\sigma$, meaning $Z\independent V\mid D$ (as $\varepsilon\independent V\mid D$), but  $X\independent V\mid D$ does not hold.
\end{remark} 

\begin{restatable}{corollary}{coroEquality}\label{coro:Equality}
    Under the conditions of Lemma \ref{lemma:IndependenceImplication}, if $X\independent V\mid D$ then for $\alpha\in(0,1)$ $$ROC_X(\alpha)=ROC_Z(\alpha)$$
\end{restatable}
This Corollary shows that when the covariates $V$ are independent  of  the raw score in the impaired and non-impaired groups, performing a Type I correction  is irrelevant for classification accuracy.  For example, if we correct the MMSE raw score for age using the z-score correction \eqref{eq:AgeEduCorrection} but age is not associated with the raw MMSE score in the groups with and without impairment, then raw and corrected scores will have the same marginal classification accuracy. 

To compare adjusted and unadjusted ROC curves in less trivial scenarios, we will make three assumptions. The first two are likely to be accepted in most real-life settings.

\begin{assumption}[The test is well-designed]\label{ass:TestWellDesigned}
    The conditional likelihood ratio of the raw score given the demographic covariates, $$LR_{X\mid V}(\cdot\mid v)=\dfrac{f_{X\mid V,D}(\cdot\mid v,1)}{f_{X\mid V,D}(\cdot\mid v,0)},$$ is  non-increasing for all $v$ in $\mathcal{V}$.
\end{assumption}
 By $f_{X\mid V,D}$ we denote the density (w.r.t. a certain dominating measure, usually the Lebesgue or the counting measure) of $X| V,D$. Assumption \ref{ass:TestWellDesigned} is equivalent to the covariate-specific ROC curves $ROC_v$ being concave for all $v$ in $\mathcal{V}$ \cite[Thm. 2 \& 3]{gneiting2022receiver}. Assumption \ref{ass:TestWellDesigned} is expected to hold when the cognitive screening test from which the raw score $X$ is obtained is designed to discriminate impaired individuals using a ``smaller than'' threshold. This is the case for the MMSE, the MoCA and several other cognitive screening tests. If Assumption \ref{ass:TestWellDesigned} does not hold, there would be a $V$-stratum and a certain FPR for which the optimal decision rule (the likelihood ratio test) classifies as ``positive'' those with a score in a set not described by a  ``smaller than'' threshold, meaning the clinical test is flawed in its construction. This assumption aligns with comments made in \cite{pesce2010convexity} that non-concave ROC curves ``must be considered irrational''. 

\begin{assumption}[The correction is well-designed]\label{ass:WellDesignedCorrection}
    When the corrected score is constructed under Definition \ref{def:Correction}, then $$Z\independent V\mid D=0.$$
\end{assumption}

This assumption can be expected to hold in practice, as achieving the conditional independence in Assumption \ref{ass:WellDesignedCorrection} is arguably the purpose of any correction under Definition \ref{def:Correction}. For example, when describing age-correction, \cite{Petersen2003-na} write that the correction removes ``the effect of age so that the average test score across persons of different ages will be invariant''. The claim suggests that the procedure makes the corrected score $Z$ independent of the demographic variable $V$ in the $D=0$ group.
The most common demographic correction, the z-score correction \eqref{eq:AgeEduCorrection}, is often justified by the assumption that the raw score is normally distributed conditional on the demographic variables, among individuals in the normative sample \cite{Mitrushina2005-gb}. If $X\mid V,D=0$ is normally distributed, or more generally belongs to a location-scale family, the z-score correction is guaranteed to satisfy Assumption \ref{ass:WellDesignedCorrection}. In Appendix \ref{sec:AppComparingGeneral_GeneralCase} we show that Assumption \ref{ass:WellDesignedCorrection} can be omitted by strengthening the other assumptions (see Assumption \ref{ass:SufficientConditionsGeneral} in Appendix \ref{sec:AppComparingGeneral_GeneralCase}).

Assumptions \ref{ass:TestWellDesigned} and \ref{ass:WellDesignedCorrection} are plausible in most settings. However the third assumption will not hold just based on basic principles. Yet this assumption helps us describe when performing a correction leads to worse classification accuracy. 

\begin{assumption}[Probabilistic co-monotonicity]\label{ass:Comonotonicity}
    For a given $\alpha\in(0,1)$ it holds that $\alpha_v^X\defeq\mathbb{E}[\psi_\alpha(X)|D=0,V=v]$ and $\mathbb{P}(D=1\mid {X=F_{X, 0}^{-1}(\alpha)},V=v)$ are co-monotone\footnote{By co-monotone we mean they are either both increasing or decreasing. For a more general definition see \cite{armstrong1993chebyshev}.} in~$v$.
\end{assumption}

Informally, this assumption establishes that the association between $V$ and $X$ amongst non-impaired individuals  ($D=0$) and the association between $V$ and $D$ after conditioning on a certain raw score must have the same ``direction'' or ``sign''. We further examine this Assumption in Section \ref{sec:CausalInterpretation}.

Assumptions \ref{ass:TestWellDesigned}-\ref{ass:Comonotonicity} can be checked empirically if the $D$-status of individuals is known. This is the case, for example, in many studies aimed to determine classification accuracy of cognitive screening tests \cite{Cullen2007-qt}, where  cognitive impairment is also assessed. If parametric models are posited, the conditions in Assumptions \ref{ass:TestWellDesigned}-\ref{ass:Comonotonicity} could be translated to conditions on the coefficients of the models, which then could be checked empirically using the fitted parameters.  We give an example in Appendix \ref{sec:AppParametricModel_X}, related to Remark \ref{remark:ReciprocateNotTrue}.

Under  Assumptions \ref{ass:TestWellDesigned} and \ref{ass:WellDesignedCorrection}, it can be shown that when the co-monotonicity in Assumption \ref{ass:Comonotonicity} holds,  raw scores have equal or better marginal classification accuracy than corrected scores.

\begin{restatable}{theorem}{thmSuperiorityExtended}\label{thm:Superiority}
    Fix $\alpha\in (0,1)$. Under Definition \ref{def:Correction} and Assumptions \ref{ass:TestWellDesigned} and \ref{ass:WellDesignedCorrection} it holds that
    \begin{equation}\label{eq:LowerBound}
        ROC_X(\alpha)-ROC_Z(\alpha)\geq \mathbb{E}\left[LR_{(X,V)}(F_{X, 0}^{-1}(\alpha),V)(\alpha_V^X-\alpha)\mid D=0\right].
    \end{equation}
    If furthermore Assumption \ref{ass:Comonotonicity}  holds, then $$ROC_X(\alpha)\geq ROC_Z(\alpha).$$
\end{restatable}
We give a proof in Appendix \ref{sec:AppProofs_ComparingScores}. When both scores are ``given an equal starting point'' by selecting thresholds leading to the same FPR, Theorem \ref{thm:Superiority} gives concrete conditions guaranteeing that the raw score has higher TPR. Furthermore,  Theorem \ref{thm:Superiority} motivates an implementable test for superiority of the raw score: let $\tau(\alpha)$ be the expectation on the right-hand-side of Equation \eqref{eq:LowerBound}. Under Assumptions \ref{ass:TestWellDesigned} and \ref{ass:WellDesignedCorrection}, Theorem \ref{thm:Superiority} establishes that rejecting the null hypothesis $\mathcal{H}_0\colon\tau(\alpha)< 0$ provides evidence for the superiority of the raw score over the demographically-corrected score. 

Theorem \ref{thm:Superiority} has practical consequences. In particular, it can be used to explain observed empirical results from previous studies \cite{Belle1996-ws, kraemer1998adjusting, o2004correcting, hessler2014age, Quaranta2016-cn, Hemminghyth2025} showing better classification accuracy of raw scores compared to corrected scores. 
Theorem \ref{thm:Superiority} can also justify the intuition that corrections are undesirable when the demographic variables are strongly associated with cognitive impairment \cite{berkman1986association, sliwinski1997effect}. We elaborate more on this point in Section \ref{sec:CausalInterpretation}. Previous simulations and analytical results have suggested that corrections are detrimental when demographic variables are strongly associated with cognitive impairment under strict parametric models \cite{o2010age, piccininni2023should, wischmann2024consequences}. Here, we improve on these results, providing sufficient conditions for the superiority of the raw score, without relying on parametric models, or vague notions of ``strong'' associations.

Some practitioners believe that corrected scores outperform raw ones as they leverage more information, and therefore they have higher classification accuracy \cite{piccininni2023should}. We give an information-theoretic argument for why this belief is misleading in Appendix~\ref{sec:AppInformationTheory}.

The assumptions in Theorem \ref{thm:Superiority} can be modified to determine sufficient conditions for superiority of the corrected score. Under Definition \ref{def:Correction}, Assumption \ref{ass:TestWellDesigned} holds for the raw score, in the sense that $LR_{X|V}(\cdot\mid v)$ is decreasing for all $v$ in $\mathcal{V}$ if and only if $LR_{Z|V}(\cdot\mid v)$ is decreasing. This follows from Lemma \ref{lemma:EqualCovariateSpecificRocs}. To achieve an analogous result to Theorem \ref{thm:Superiority}, we also  need to modify Assumption \ref{ass:Comonotonicity}.

\begin{assumption}[Probabilistic counter-monotonicity]\label{ass:Countermonotonicity}
    For a given $\alpha\in(0,1)$ it holds that $\alpha^X_v$ and $\mathbb{P}(D=1\mid {Z=F_{Z, 0}^{-1}(\alpha)},V=v)$ are counter-monotone in~$v$.
\end{assumption}

Informally, Assumption \ref{ass:Countermonotonicity} states that the association between $V$ and $X$ among individuals with $D=0$ goes in the opposite direction compared to the association between $V$ and $D$ for a certain value of $Z$. When this is the case, corrected scores will have higher classification accuracy.  

\begin{restatable}{theorem}{thmSuperiorityZ}\label{thm:Superiority_Z}
    Fix $\alpha\in (0,1)$. Under Definition \ref{def:Correction} and Assumptions \ref{ass:TestWellDesigned}, \ref{ass:WellDesignedCorrection} and \ref{ass:Countermonotonicity}, it holds that $$ROC_Z(\alpha)\geq ROC_X(\alpha).$$ 
\end{restatable}
This gives a result similar to Corollary \ref{coro:Equality}, though not restricted to Type I corrections. 
\begin{restatable}{corollary}{coroEqualityweak}\label{coro:Equality_weak}
    Fix $\alpha\in (0,1)$. Under Definition \ref{def:Correction} and Assumptions \ref{ass:TestWellDesigned}-\ref{ass:Countermonotonicity}, it holds that $$ROC_Z(\alpha)= ROC_X(\alpha).$$
\end{restatable}

We elaborate on these results with a parametric example in Appendix \ref{sec:AppParametricModel_Z}, related to Remark \ref{remark:ReciprocateNotTrue}. 
The fact that the comparison provided by Theorems \ref{thm:Superiority} and \ref{thm:Superiority_Z} is point-wise can be seen as a strength or a weakness. If the goal is to compare the classifiers at a certain FPR which has been deemed optimal in the past, then our formalism is sufficient. However, Theorems \ref{thm:Superiority} and \ref{thm:Superiority_Z} do not provide a comparison of a global accuracy measure, i.e, one that does not depend on the FPR $\alpha$. Nevertheless, there are some settings where  Assumption \ref{ass:Comonotonicity} holds for either all $\alpha$ or none. See for example the  parametric model studied in Appendix~\ref{sec:AppParametricModel}.   A straightforward application of these theorems is to compare overall classification accuracy in terms of AUCs. 

\begin{restatable}{corollary}{coroAUC}\label{coro:AUC}
    Under Definition \ref{def:Correction}, assume Assumptions \ref{ass:TestWellDesigned} and \ref{ass:WellDesignedCorrection} hold. Furthermore, if
    \begin{itemize}
        \item Assumption \ref{ass:Comonotonicity}  holds for all $\alpha\in(0,1)$, then $AUC(X)\geq AUC(Z)$.
        \item Assumption \ref{ass:Countermonotonicity}  holds for all $\alpha\in(0,1)$, then $AUC(Z)\geq AUC(X)$.
    \end{itemize}
\end{restatable}

This Corollary is useful as it yields sufficient conditions for superiority in terms of the AUC, a global classification performance measure. Still, it relies on some strong assumptions: it does not allow raw and corrected ROC curves to cross. For example, it is possible that ROC curves cross and still $AUC(X)\geq AUC(Z)$. This is considered a disadvantage of the AUC as a comparative tool \cite{pepe2003statistical}. However, our previous results like Theorems \ref{thm:Superiority} and \ref{thm:Superiority_Z} are not affected by this caveat, as they perform local (at a specific FPR) comparisons.

\subsection{A causal interpretation of the superiority conditions}\label{sec:CausalInterpretation}
So far we have made no causal assumptions about our problem of interest. Definition \ref{def:Correction} and Assumptions \ref{ass:TestWellDesigned}-\ref{ass:Comonotonicity} (as well as all assumptions in Appendix \ref{sec:AppComparingGeneral}) refer only to observational quantities. In this section, we use formal causal reasoning to clarify the meaning of the assumptions and results, as also suggested by \cite{piccininni2023should}. 

We represent variables within a Finest Fully Randomized Causally Interpretable Structured Tree Graphs (FFRCISTG) causal model \cite{robins1986new,richardson2013single}, which strictly generalizes Pearl's popular non-parametric structural equation model with independent errors (NPSEM-IE) \cite{robins2010alternative,richardson2013single,pearl}. We use directed acyclic graphs (DAGs) to represent a statistical model induced by an FFRCISTG \cite{robins1986new,richardson2013single}. Recall that $V$ represents a set of demographic characteristics. As cognitive screening tests are designed to identify individuals who have cognitive impairment, a reasonable topological ordering for the variables of our system would be $(V,D,X)$. This ordering is also compatible with the fact that demographic characteristics like age and education are often considered ``confounders'' \cite{hessler2014age}, or ``having a confounding effect'' \cite{magni1996mini}, or being ``confounding factors'' \cite{sliwinski1997effect} of the test-impairment relationship. Suppose, for simplicity of exposition, that no additional common causes of these variables exist. This causal model is reflected in the causal DAG in Figure~\ref{fig:BasicDAG}.

\begin{figure}
    \centering
    {\begin{tikzpicture}
\begin{scope}[every node/.style={thick,draw=none}]
    \node (V) at (1,0) {$V$};
    \node (X) at (3,1) {$X$};
    \node (D) at (3,-1) {$D$};
\end{scope}

\begin{scope}[>={Stealth[black]},
              every node/.style={fill=white,circle},
              every edge/.style={draw=black,very thick}]
    \path [->] (V) edge (X);
    \path [->] (V) edge (D);
    \path [->] (D) edge (X);

\end{scope}
\end{tikzpicture}
}

\caption{\small DAG representing the causal structure of the variables in our system. }
\label{fig:BasicDAG}
\end{figure}
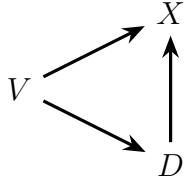

This causal DAG allows us to give a clear interpretation to the key Assumption~\ref{ass:Comonotonicity}. First, how $\alpha_v^X=\mathbb{E}[\psi_\alpha(X)|D=0,V=v]$ varies with $v$ can be seen as a statement about (direct) effects of $V$ on $X$. Indeed, for a fixed $\alpha$, $\psi_\alpha(X)$ is a deterministic function of $X$. Conditioning on $D=0$ blocks the causal path mediated by $D$, $V\to D\to X$, so that the association between $V$ and $X$ is only induced by the direct path $V\to X$. The notion that this conditional FPR is monotone in $v$ is closely related to the notion of weak monotonic effects \cite{vanderweele2009minimal,vanderweele2010signed}. Second, in Assumption \ref{ass:Comonotonicity} we consider how the conditional probability $\mathbb{P}(D=1\mid {X=F_{X, 0}^{-1}(\alpha)},V=v)$ varies with $v$. This means that we are making a statement about the association between $V$ and $D$ conditional on $X$. When conditioning on $X$, the association between $V$ and $D$ is explained by two paths: the causal path $V\to D$, and the non-causal collider path $V\to X \leftarrow D$.
Conceptually, Assumption \ref{ass:Comonotonicity} requires the combination of these two associations, causal and non-causal, to be in the same direction as the direct effect of $V$ on $X$ among the non-impaired.

Consider the example of age-correction for the MMSE, so that $X$ represents the MMSE score, $V$ represents age, and $D$ represents dementia status. There is clear empirical evidence indicating that non-impaired individuals have lower MMSE scores as they age, see for example \cite{magni1996mini, grigoletto1999norms}. Therefore we know that $\alpha^X_v$ is increasing in $v$. We also know that the prevalence of dementia increases with age \cite{nichols2022estimation}. Therefore, the conditional probability $\mathbb{P}(D=1\mid V=v)$ increases in $v$. 
Then, Assumption \ref{ass:Comonotonicity} would require that the collider bias induced on the association between $V$ and $D$ by opening the path $V\to X \leftarrow D$ when conditioning on $X=F_{X, 0}^{-1}(\alpha)$ is not strong enough to reverse the association between $V$ and $D$ owed to the causal effect $V \to D$. 
Conceptually, when the collider bias is not strong enough to ``flip'' the association between age and dementia, the raw score has better classification accuracy than the age-corrected score, according to Theorem \ref{thm:Superiority}. This result aligns with the idea that age correction can be detrimental if age has a strong effect on cognitive impairment \cite{piccininni2023should, wischmann2024consequences, sliwinski1997effect, o2010age}. This is reflected in the parametric example we present in Appendix \ref{sec:AppParametricModel}.

\section{Other reasons for using corrections}\label{sec:FDP}
\subsection{Equal classification accuracy across demographic groups}

Marginal classification accuracy is not the only criterion to compare raw and demographically-corrected scores. Already in 1998, \cite{kraemer1998adjusting} recognized that different researchers were using different criteria to evaluate screening strategies for dementia. Specifically, \cite{kraemer1998adjusting} illustrated how proponents of demographic corrections were not concerned with maximizing classification accuracy, but with ensuring equal true- and false-positive rates across demographic groups.

Demographic corrections of cognitive screening tests are often proposed as a remedy to an undesirable association between demographic variables and test scores. For example, corrections of the MMSE are motivated by the observation that age and education are associated with MMSE scores \cite{kraemer1998adjusting, magni1996mini}. Similarly, age-education corrections of the MoCA test have been motivated by the fact that age and education explain up to 49\% of the MoCA scores variance in normative samples \cite{Larouche2016-fd}. 

Demographic corrections have been proposed as a way to avoid overdiagnosis in vulnerable groups. Proponents of education corrections were concerned about a "detection bias caused by under-education": individuals without impairment who have low education tend to get low scores in cognitive screening tests and are therefore erroneously classified as impaired \cite{Belle1996-ws}. For example, \cite{grigoletto1999norms} argued that without an education correction "many persons may be unfairly labelled as demented (or at risk for dementia) when the only problem was a lack of schooling". Similarly, \cite{mungas1996age} considered the  age-and-education-corrected MMSE a preferable screening tool for low-education and minority individuals. More recently, some authors have argued that demographic corrections can be  justified as a way to fulfil certain fairness criteria in cognitive screening \cite{piccininni2023should,wischmann2024consequences}. This point has been echoed by supporters of demographically-corrected scores, who find desirable to have similar classification accuracy  in different demographic strata \cite{Aiello2025}, avoiding overdiagnosis in vulnerable groups \cite{ilardi2025no}.

As we have discussed in Section \ref{sec:ComparingScores}, a well-designed demographic correction makes the corrected score independent of the demographic covariates amongst individuals without impairment; that is, we expect $Z\independent V | D=0$ (Assumption \ref{ass:WellDesignedCorrection}). In other words, we expect the corrected score to have the same FPR (and the same specificity) in all demographic groups \cite{wischmann2024consequences}. This would solve,  for example, the above-mentioned issue of a higher FPR for older individuals with lower education, which can lead to overdiagnosis and social discrimination for vulnerable groups. More formally, demographic corrections can be seen as a way to achieve the "False-positive error rate balance" criterion \cite{verma2018fairness}, which is one of a large number of different algorithmic fairness criteria. Furthermore, in the specific setting where $X$ follows an additive model, the z-score correction ensures that both FPR and TPR are constant across demographic strata, see Remark \ref{remark:ReciprocateNotTrue}, achieving the so-called "equalized odds" criterion \cite{verma2018fairness}.

\subsection{Removal of demographic effects}\label{sec:FDP_Main} 
Demographic corrections have also been described as a way to ``remove the effect'' of the demographic variables from the test score \cite{sliwinski1997effect, Petersen2003-na, piccininni2023should, scheffels2023influence}. \cite{scheffels2023influence} explicitly argued that, in an ideal world, cognitive screening tests should be  ``independent of culture, language and education'' and ``free of socio-demographic effects''; because it is practically impossible to design cognitive screening tests that are free of these effects, demographic corrections are used to remove them. This suggests that demographic corrections are intended as a tool to remove effects of variables that are determined by demographics.

To formalize this plain-English argument, which is nevertheless based on causal reasoning, we consider an extension of the causal model introduced in Section~\ref{sec:CausalInterpretation}, where again we represent variables within  a FFRCISTG model and use DAGs to represent a statistical model induced by this FFRCISTG \cite{robins1986new,richardson2013single}. The assumptions we make are not necessarily plausible in practice; rather, we present results that clarify sufficient conditions under which certain procedures have the properties often attributed to them in discussions.

Consider the example of age-and-education correction of the MMSE. Like other screening tests, the MMSE is designed to rely on communication, language, and arithmetic to assess the cognitive ability of the screenees. The concern is that the test scores depend on age and education as these demographic variables affect sensory abilities and levels of background knowledge. For a given screenee, let $V$ denote the demographic characteristics of interest, like age and education, and consider some score $S$. For now, we are agnostic as to whether this is the raw MMSE or any corrected (MMSE) score. As before $D$ indicates the impairment status. We also consider a set of baseline covariates $L$, which could possibly affect other variables. 

To fix ideas, consider a screenee with low education, who is not familiar with basic arithmetic and has a narrow vocabulary. This screenee is expected to have a low score. Similarly, an older screenee, who has impaired vision and hearing, is expected to score low. We introduce a new variable $V_S$, potentially unmeasured, which can be interpreted as the level of sensory ability/background knowledge to which the screening test score is \textit{sensitive}\footnote{Our use of the term \textit{sensitivity} here and in this Section is not related to the notion of true positive rate in Section \ref{sec:ComparingScores}.}, due to a fundamental design limitation. We reason that $V_S$ is exactly the variable that practitioners wish would not affect the test score. 

We will assume that, in the observed data, $V$ and $V_S$ are deterministically related, i.e., $V=V_S$ almost surely. In other words, every level of age-education corresponds to one distinct level of sensory ability/background knowledge.  For example, a person who earned a PhD has a very specific PhD-level knowledge relevant to the test. Clearly, this assumption can be implausible. For example, there can be individuals with the same years of education who have different relevant background knowledge levels. The determinism is often implicit in demographic corrections as, for example, individuals with the same age and number of years of education will have their scores corrected in the same way, regardless of their actual amount of background knowledge or sensory abilities.

There might also exist other variables, beyond impairment status, affected by $V$, which a desirable test could be sensitive to (see Appendix \ref{sec:App_FDP_Identification} for  further elaboration). We denote these additional covariates as $W$. In the context of the age-and-education correction of the MMSE, $W$ could indicate the presence of another neurological disease not included in $D$. Therefore, we consider the topological ordering  $(L,V,W,D,S)$ in our system. In Appendix \ref{sec:App_FDP} we will assume $L$ and $W$ to have discrete supports $\mathcal{L}$ and $\mathcal{W}$ respectively, but our arguments  extend straightforwardly to the continuous case w.r.t.\ the Lebesgue measure.

Interventions on age and education $V$ are hard to conceive, but we can hypothetically intervene on $V_S$ by having the screenee take a different version of the test, changing what knowledge or sensory skills are advantageous and relevant. For example, we could replace original words with simpler synonyms to reduce the ``sensitivity'' to education, or administering the test with a hearing aid to alter the ``sensitivity'' to age. These interventions would change the sensory skills or background knowledge that are relevant to the test, $V_S$, rather than altering the true demographics $V$. The variable $V_S$ is what \cite{wen2024causal}  label an \textit{intervening} variable  (see \cite[Ex.\ 2]{wen2024causal} for an example of race discrimination closely related to our discussion). As we have said, here we interpret demographic correction as an attempt to correct for $V_S$. This specification makes no difference in the observed world, where $V\equiv V_S$. Yet, in the hypothetical experiment where $V_S$ is intervened, we envision that demographic corrections of the raw score would be based on $V_S$ rather than on $V$, since $V_S$ are the consequences of the demographics where the undesired effects are rooted.

We will use superscripts to denote interventions on $V_S$. To be precise, $S^{v_S}$ denotes the counterfactual score if, possibly contrary to fact, $V_S$ had taken the value $v_S$. Analogous notation will apply to $W,D$. Even though $V$ and $V_S$ are deterministically related in the observed data, we can identify the causal effect of $V_S$ using results on the separable effects  \cite{stensrud2021generalized,stensrud2022separable,stensrud2023conditional,robins2010alternative, robins2020interventionist}. The key assumption is that $V_S$ does not affect cognitive impairment $D$ neither directly nor indirectly. We formalize this assumption in Appendix \ref{sec:App_FDP_Identification}. To represent causal assumptions in the observed data involving $V_S$, we will use an {extended} causal DAG, which includes both $V$ and $V_S$, together with a bold arrow $V\boldsymbol{\rightarrow} V_S$ to represent the observed data determinism. An example of such a DAG can be seen in Figure \ref{fig:FDP_ExtendedDAG}.

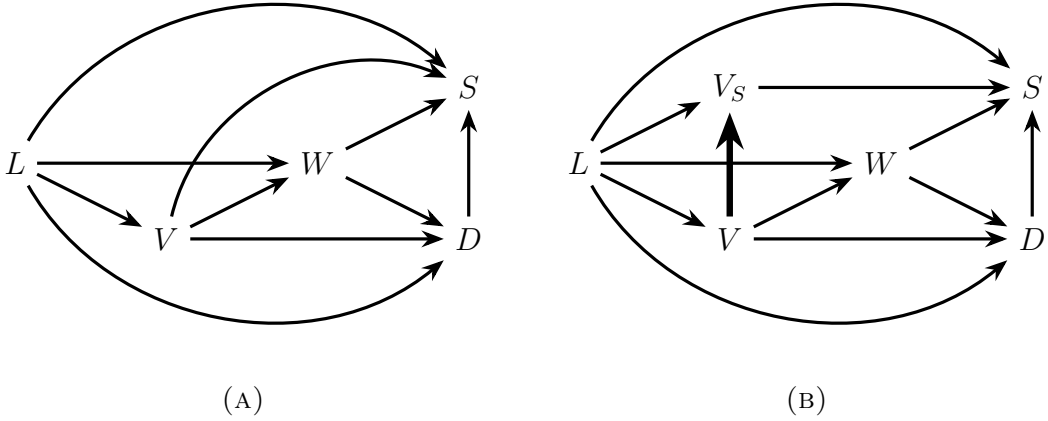
\begin{figure}
    \centering    
\begin{subfigure}{0.5\textwidth}
\centering
\begin{tikzpicture}
\begin{scope}[every node/.style={thick,draw=none}]
    \node (L) at (-1,0) {$L$};
    \node (V) at (1,-1) {$V$};
    \node (S) at (5,1) {$S$};
    \node (W) at (3,0) {$W$};
    \node (D) at (5,-1) {$D$};
\end{scope}

\begin{scope}[>={Stealth[black]},
              every node/.style={fill=white,circle},
              every edge/.style={draw=black,very thick}]
	\path [->] (L) edge (V);
    \path [->] (L) edge[bend left=50] (S);
    \path [->] (L) edge[bend right=50] (D);
    \path [->] (V) edge[bend left=50] (S);
    \path [->] (V) edge (D);
    \path [->] (W) edge (D);
    \path [->] (V) edge (W);
    \path [->] (W) edge (S);
    \path [->] (L) edge (W);
    \path [->] (D) edge (S);

\end{scope}
\end{tikzpicture}
\caption{}
\end{subfigure}%
\hfill
\begin{subfigure}{0.5\textwidth}
\centering
\begin{tikzpicture}
\begin{scope}[every node/.style={thick,draw=none}]
    \node (L) at (-1,0) {$L$};
    \node (V) at (1,-1) {$V$};
    \node (VS) at (1,1) {$V_S$};
    \node (S) at (5,1) {$S$};
    \node (D) at (5,-1) {$D$};
    \node (W) at (3,0) {$W$};
\end{scope}

\begin{scope}[>={Stealth[black]},
              every node/.style={fill=white,circle},
              every edge/.style={draw=black,very thick}]
	\path [->] (L) edge (V);
    \path [->] (L) edge (VS);
    \path [->] (L) edge (W);
    \path [->] (L) edge[bend left=50] (S);
    \path [->] (L) edge[bend right=50] (D);
    \path [->] (V) edge[line width=0.85mm] (VS);
    \path [->] (VS) edge (S);
    \path [->] (V) edge (D);
    \path [->] (D) edge (S);
    \path [->] (W) edge (S);
    \path [->] (W) edge (D);
    \path [->] (V) edge (W);

\end{scope}
\end{tikzpicture}
\caption{}
\end{subfigure}

\caption{\small Example of a causal DAG for the scenario introduced in Section \ref{sec:FDP}  (A) and its associated extended causal DAG (B). }
\label{fig:FDP_ExtendedDAG}
\end{figure}

We can now formalize the notion of a desirable test expressed in \cite{scheffels2023influence}. We consider a score $S$ to be insensitive to demographics $V$, when is not affected by interventions on the sensitive  attributes $V_S$. Informally, if we were to perform an intervention by changing the screening test to alter the knowledge and sensory ability relevant to its undertaking, a desirable test would result in the same score distribution.  To be more precise, we will say that a score $S$ satisfies demographic insensitivity when all true and false covariate specific positive rates are not affected by an intervention on the sensitive attributes.

\begin{definition}[Strong demographic insensitivity]\label{def:FDP}
    We say that a score $S$ is insensitive to the demographics $V$ in the strong sense   if the distribution of $S^{v_S}\mid \{V=v,D=d\}$ does not depend on $v_S$ for all $(v,d)$ in~$\mathcal{V}\times\{0,1\}$.
\end{definition}
Yet, in our screening problem we are not necessarily interested in the full distribution of the score, but rather in the decision that is taken based on it. Thus, we define a weaker notion of demographic insensitivity.
\begin{definition}[Weak demographic insensitivity]\label{def:FDP_weak}
    Let $S$ be a score. We say that an $S$-based decision rule $\varphi\colon{\rm supp}(S)\to\{0,1\}$ is insensitive to the demographics $V$ if  $\mathbb{P}(\varphi(S^{v_S})=1\mid V=v,D=d)$ does not depend on $v_S$ for all $(v,d)$ in~$\mathcal{V}\times\{0,1\}$.
\end{definition}
We refer to this notion  as weak, in contrast to the strong version (Definition \ref{def:FDP}), as it is implied by the latter as long any randomness in $\varphi$ is not affected by the intervention $V_S=v_S$\footnote{This is the case for example of the practical implementation of randomized decision rules $\varphi(S)=I(\mathcal{U}\leq\psi_\alpha(S))$ where $\psi_\alpha(S)$ is defined as in Equation \eqref{eq:RandomizedDecision} and $\mathcal{U}\sim$Unif(0,1).  }. This is the case for deterministic decision rules, and specifically for thresholdings $\varphi(S)=I(S\leq t)$ with a pre-specified $t$. As these are the kind of decision rules we are interested in, we will implicitly consider those in Definition \ref{def:FDP_weak}, but the formalism we  deploy in Appendix \ref{sec:App_FDP} applies in general as long as the decision rule is pre-specified.

In practice, assessing whether a score $S$ is insensitive to demographics $V$ requires identification and estimation of the counterfactual probabilities $\{\mathbb{P}(S^{v_S}\leq s\mid V=v,D=d)\colon (v,v_S)\in\mathcal{V}^2,\ d\in\{0,1\}\}$ for a fixed $s$ or for (almost-surely) all $s$ in supp$(S)$. Furthermore, as long as $V$ is not almost surely constant, testing for demographic insensitivity will be a multiple testing problem, which has to be adequately dealt with to ensure proper error control. We discuss all these issues in Appendix \ref{sec:App_FDP}.

As a simple example, consider the  scenario  in Section \ref{sec:CausalInterpretation}, where $(L,W)=\emptyset$. Fix $t\in$ supp$(S)$ and take the decision rule $\varphi(\cdot)=I(\cdot\leq t)$. Under the conditions discussed in Appendix \ref{sec:App_FDP_Identification}, the screening strategy $\varphi(S)$ is insensitive to $V$ in the weak sense (Definition \ref{def:FDP_weak}) if and only if $\mathbb{P}(S\leq t\mid V=v_S,D=d)$ does not depend on $v_s$ for each $d$ in $\{0,1\}$. In other words, if and only if the TPR and the FPR calculated at the threshold $t$ do not change with $V$. As we have discussed, while the demographic correction is expected to ensure equal FPR across demographic strata, it does not guarantee equal TPR. Therefore, except in special cases (for example, see Appendix \ref{sec:AppParametricModel_FDP}), demographic corrections do not ensure demographic insensitivity.

\section{Data application: the OASIS-3 study}\label{sec:OASIS3}
\subsection{The OASIS-3 data}
We analysed data from OASIS-3 \cite{LaMontagne2019OASIS3}, which is part of the OASIS (Open Access Series of Imaging Studies) project. Its goal is to make neuroimaging data sets of the brain freely available to the scientific community. In particular, the OASIS-3 dataset ``is a retrospective compilation of data for 1378 participants that were collected across several ongoing projects through the WUSTL Knight ADRC over the course of 30 years''.

At study entry, age, sex, and socio-economic status were recorded. Moreover, participants were administered a MMSE. For each participant, neurological status was assessed at baseline using the \textit{Clinical Dementia Rating} (CDR) scale. We only included in our analyses participants who had no missing values in any of these variables, resulting in a data set including 1317 individuals. We considered participants to have cognitive impairment ($D=1$) when their CDR was above zero.
Approximately 28\% of the individuals were therefore classified as having cognitive impairment. The median MMSE score was  29 in the $D=0$ group, and 27 in participants with cognitive impairment $D=1$. We denote the MMSE raw score with the letter $X$ for consistency. We also define the demographic covariate $V$ to be a categorical variable representing age groups: below 65, 65-70, 70-75 and above 75 years of age at baseline. This partitioning corresponds roughly to the age quartiles of the OASIS-3 cohort.

We have two goals. First, we want to study  the  classification superiority of the MMSE raw score over the $V$-corrected scores. Specifically,  we will focus on  Assumption \ref{ass:Comonotonicity}, crucial to determine the classification superiority. Second, we will study whether the raw and age-corrected MMSE scores are insensitive to the age-quartiles group indicator $V$. 

\subsection{Classification accuracy  comparison} \label{sec:OASIS3_Superiority}
To study whether raw MMSE scores have a better classification accuracy than age-corrected scores, we apply the formalism developed in Section \ref{sec:ComparingScores}. Under Assumptions \ref{ass:TestWellDesigned} and \ref{ass:WellDesignedCorrection}, Theorem \ref{thm:Superiority} tells us that a sufficient condition for superiority of the raw MMSE score at FPR $\alpha$ is that $$\tau(\alpha)\defeq \mathbb{E}\left[LR_{(X,V)}(F_{X, 0}^{-1}(\alpha),V)(\alpha_V^X-\alpha)\mid D=0\right]>0.$$ Rather than fixing a false-positive rate $\alpha$, as $X$ has a discrete support, we fix $t\in\text{supp}(X)$ and consider its associated marginal FPR $\alpha(t)=F_{X,0}(t)$. We then reframe the former quantity in terms of the  threshold $t$, such that $$\tau(t)= \mathbb{E}\left[LR_{(X,V)}(t,V)(\alpha_V^X(t)-\alpha(t))\mid D=0\right].$$
Fixing a threshold $t$ rather than a FPR $\alpha$ not only enhances interpretability, but also serves a theoretical purpose. As $t\in\text{supp}(X)$ and $\alpha(t)=F_{X,0}(t)$, it holds that   $F_{X,0}(F_{X, 0}^{-1}(\alpha(t)))=\alpha(t)$, that implies $\rho_{\alpha(t)}(X)=1$,  $\psi_{\alpha(t)}(X)=I(X\leq t)$, and most importantly $\alpha^X_v(t)=\mathbb{P}(X\leq t|D=0,V=v)=F_{X,0}(t|v)$. Furthermore, this allows us to relax the positivity assumptions in Section \ref{sec:ComparingScores}, as we only need $t$ to belong to the support of $X|V=v,D=d$ for all $v,d$, rather than the full support being invariant.

As $X$ and $V$ have a discrete support and $D$ is binary, all the quantities involved in $\tau(t)$ can be consistently estimated by sample averages, as for example $\widehat{\alpha}(t)=\widehat{\mathbb{E}}_n[I(X\leq t)|D=0]$. Thus, for a range of thresholds $t\in\{25,26,27,28,29\}$ we estimated $\tau(t)$ and computed a 95\% right-sided asymptotic confidence interval. These results, together with the estimates of the associated marginal FPRs at each thresholds can be found in Table \ref{tab:OASIS_Superiority_NP}. Further details on the estimation method and the central limit theorem used to compute the confidence intervals can be found in Appendix \ref{sec:AppOASIS_SuperiorityCLT}.

\begin{table}
\centering
\begin{tabular}{cc|cc}
$t$  & $\widehat{\alpha}(t)$ & $\widehat{\tau}(t)$ & $LCL$    \\ \hline
25 & 0.015      & 0.062    & -0.018 \\
26 & 0.052      & 0.030    & -0.001 \\
27 & 0.100      & -0.012   & -0.091 \\
28 & 0.209      & 0.043    & 0.009 \\
29 & 0.512      & 0.024    & 0.011
\end{tabular}
\caption{ Point estimates and the lower confidence limit $LCL$ of 95\% right-sided asymptotic confidence intervals for the superiority lower bound $\tau(t)$ for different thresholds $t$ of the raw MMSE score, together with estimates of the associated FPRs.  }
\label{tab:OASIS_Superiority_NP}
\end{table}

Recall that ${\tau}(t)>0$ implies $ROC_X(\alpha(t))>ROC_Z(\alpha(t))$. We see from the estimates in Table \ref{tab:OASIS_Superiority_NP} that for most thresholds $t$ we get point estimates $\widehat{\tau}(t)>0$, pointing towards superiority of the raw score in comparison to any age-corrected score derived from it under Definition \ref{def:Correction} satisfying Assumption \ref{ass:WellDesignedCorrection}. We reject the null hypothesis $\mathcal{H}_0\colon \tau(t)<0$ at a 5\% level for the threshold $t\in\{28,29\}$, as 95\% right-sided confidence intervals do not contain 0. The confidence intervals intersect the null region for $t \in \{25,26,27\}$, so no conclusion can be made with regard of these thresholds. Overall, this means that the OASIS-3 data points towards superiority of the raw score  with respect to any correction using age-quartiles under Definition \ref{def:Correction} and Assumption \ref{ass:WellDesignedCorrection}. However, using non-parametric methods, the precision is too low to obtain definitive conclusion for all the thresholds we considered.

In an alternative analysis, we assumed the following hierarchical model
\begin{equation}\label{eq:OASIS_SEM}
    \begin{cases}
        D\mid V^*\sim Ber(\expit(\beta_0+\beta_1V^*))\\
        X\mid D,V^*\sim\mathcal{N}(\mu_0+\mu_1V^*+\mu_2D,\sigma^2),
    \end{cases}
\end{equation}
where $V^*$ is the non-categorized age at time of assessment, $Ber(p)$ is a Bernoulli random variable with success probability $p$ and $\mathcal{N}(\mu,\sigma^2)$ is a normal distribution with mean $\mu$ and variance $\sigma^2$. This model is  the one  introduced  and studied in Appendix \ref{sec:AppParametricModel}. While this model is  parametric, the parametric assumptions are not too far-fetched. The demographic correction is often implemented assuming a linear model for the raw score, at least among individuals without cognitive impairment \cite{Kornak2019, Larouche2016-fd}. The reasoning  in Appendix \ref{sec:AppParametricModel_X} allows to translate the assumptions of Theorem \ref{thm:Superiority} into conditions on the coefficients of the model. A sufficient condition for the superiority of the raw score is that the following null hypothesis is false
\begin{equation}\label{eq:Oasis_Superiority_Null}
    \mathcal{H}_0: \mu_1\left(\beta_1-\frac{\mu_1\mu_2}{\sigma^2}\right)>0 \text{ or } \mu_2>0.
\end{equation}
We tested this non-superiority condition  using 10 000 bootstrap repetitions \cite[Chs. 4 \& 5]{Davison_Hinkley_1997}. This parametric analysis led to rejection of  the null hypothesis in Equation \eqref{eq:Oasis_Superiority_Null}, with p-values considerably below the 5\% level. Consequently, we have statistical evidence in favour of the superiority of the raw score regarding classification accuracy at any FPR, under the  posited parametric model.

An alternative way to study whether Assumption \ref{ass:Comonotonicity} is satisfied is to  specify the following  conditional (non-causal) model 
\begin{equation*}
    \begin{cases}
        D\mid X,V^*\sim Ber(\expit(\delta_0+\delta_1V^*+\delta_2 X))\\
        X\mid D=0,V^*\sim\mathcal{N}(\eta_0+\eta_1V^*,\sigma^2),
    \end{cases}
\end{equation*}
which is compatible with the hierarchical model in Equation \eqref{eq:OASIS_SEM}, but does not require specification of $X\mid D=1,V^*$. In this case, a sufficient condition for Assumption \ref{ass:Comonotonicity} to hold is that the null hypothesis $\mathcal{H}_0\colon \eta_1\delta_1>0$ is false. A 95\% bootstrapped percentile confidence interval for $\eta_1\delta_1$, based on 10 000 bootstrap repetitions, is $(-0.0025, -0.0012)$. Thus, again, we have evidence in favour of Assumption \ref{ass:Comonotonicity} being true, under the specified conditional models. Overall, these analyses strengthen our confidence in Assumption \ref{ass:Comonotonicity} being true and in the superiority of the raw score to the age-corrected score in terms of classification accuracy.

\subsection{Demographic insensitivity}\label{sec:OASIS3_FDP}
We also investigated whether the raw and age-corrected MMSE scores are  insensitive to the age-quartiles group indicator $V$. In this analysis, we will consider the individual's sex as baseline covariate $L$, and as $W$ the indicator of whether an individual belonged to the highest socio-economic status group at the time of entry in the OASIS-3 cohort. Socio-economic status might be influenced by the individual's age at the time of study entry, as older individuals may be in different financial circumstances than younger ones. Our choice of including the socio-economic status indicator as $W$ and not as another component of $L$ or $V$ is due to the fact that in this application we are only concerned with the demographic effect of age, and do not consider "problematic" the effect of socio-economic status on the test score.

First, we will study whether the raw MMSE score is insensitive to the age-group demographics $V$ in the weak sense at a threshold of $t=26$. Therefore, our goal is to jointly estimate the counterfactual probabilities $\{\mathbb{P}(X^{w}\leq 26\mid V=v,D=d)\colon (v,w)\in\mathcal{V}^2,d=0,1\}$. This is done, under the identification assumptions in Appendix \ref{sec:App_FDP_Identification}, using the simple estimator described in Appendix \ref{sec:FDP_Estimation}. Testing whether the raw score is demographic insensitive follows the procedure discussed in Appendix \ref{sec:FDP_Testing} with p-values derived using a Central Limit Theorem (Theorem \ref{thm:CLT_FPD} in Appendix \ref{sec:App_OASIS_FDP_CLT}). As a result of this analysis we reject the hypothesis that the raw MMSE score is age insensitive. This means that we have significant evidence that an individual could have had a different probability of being classified as having cognitive impairment, given their age group and  cognitive impairment status, had their sensitive attributes  been different at the time of their MMSE assessment. This is in line with our knowledge that age affects MMSE scores regardless of impairment status \cite{magni1996mini, grigoletto1999norms}. For further details of the analysis see Appendix \ref{sec:App_Oasis_FDP_raw}.

We now turn to age-corrected scores.  We consider the z-score correction \eqref{eq:AgeEduCorrection}, and target whether the age-corrected MMSE  is insensitive to the age-group $V$ in the weak sense at a threshold of $t=-1.5$, commonly used in practice \cite[Tab. 1]{Bradfield2020-ae}. In this analysis, we did not find evidence against the null hypothesis of age insensitivity. Informally, this means that we cannot rule out that the age-correction achieved insensitivity to age. This result is consistent with our discussion in Section \ref{sec:FDP}. We give further details of this analysis in Appendix \ref{sec:App_Oasis_FDP_corrected}.

\section{Conclusion}\label{sec:Conclusion}
Our results contribute to the ongoing methodological debate about the use of demographic corrections. We present a general sufficient condition for the classification accuracy superiority of raw scores over demographically-corrected scores. This result is useful as it allows reasoning about the consequences of demographic correction on classification accuracy using subject-matter knowledge. For example, our result explain why age-correction can be detrimental in certain cognitive screening settings, justifying formally some criticisms that have been moved against demographic correction in cognitive screening \cite{berkman1986association, sliwinski1997effect, o2010age, piccininni2023should, wischmann2024consequences}. 

We have also discussed claims about demographically-corrected scores being more "fair". Demographic corrections are expected to result in equal false-positive rates across demographic groups, a property that some practitioners find desirable. Yet, we have shown that demographic corrections are, in general, not sufficient to ensure  demographic insensitivity.

While in our work we have focused on cognitive screening, we highlight that our methods and results apply more broadly. Demographic corrections are ubiquitous in health sciences. For example,  race corrections are debated in neuropsychological assessments and neurological diagnoses, illustrated by the legal controversy surrounding the National Football League case \cite{possin2021nfl}. Demographic corrections are also often used when  assessing malnutrition in children \cite{who_malnutrition_children}, identifying \textit{small for gestational age} newborns \cite{hokken2023international}, assessing aortic root diameter \cite{dascenzi2022aortic}, measuring respiratory functions \cite{stanojevic2022ers}, or diagnosing osteoporosis \cite{sheu2016bone}.

\section*{Acknowledgements}
Data were provided  by OASIS-3: Longitudinal Multimodal Neuroimaging: Principal Investigators: T. Benzinger, D. Marcus, J. Morris; NIH P30 AG066444, P50 AG00561, P30 NS09857781, P01 AG026276, P01 AG003991, R01 AG043434, UL1 TR000448, R01 EB009352. AV-45 doses were provided by Avid Radiopharmaceuticals, a wholly owned subsidiary of Eli Lilly.

The authors were supported by the Swiss National Science Foundation (SNSF Starting Grants, Grant number: 211550). This project was funded by the Dementia Research Switzerland -- Synapsis Foundation (2026-BF01).

\clearpage
\bibliography{references_accuracy,references_FDP}
\bibliographystyle{unsrt}

\clearpage
\appendix

\section{An information-theoretic justification of the plausibility of demographic corrections worsening classification accuracy}\label{sec:AppInformationTheory}

It might seem ``intuitive'' that corrected scores cannot perform worse than raw scores, as they funnel more information into the scoring rule \cite{piccininni2023should}. This intuition is erroneous, as shown in Section \ref{sec:ComparingScores}, and it can be clarified by an information-theoretic argument. Consider the hypothetical case where we have access to the likelihood ratios of both $X$ and $Z$, but also to the joint likelihood ratio of $(X,V)$. The likelihood-ratio test provides the classification rule with the highest TPR for any FPR (as it is the uniformly most powerful test thanks to the Neyman-Pearson lemma, see \cite[Thm. 10.30]{wasserman2013all} or \cite[Ch. 2.3.2]{krzanowski2009roc} for the connection to ROC analysis). As both $X$ and $Z$ are garblings (post-processings) of $(X,V)$, the Blackwell-Sherman-Stein theorem \cite{blackwell1954theory} implies that $(X,V)$ has higher Blackwell order than both $X$ and $Z$, meaning the joint likelihood ratio test for $(X,V)$ cannot have lower TPR than that of $X$ or $Z$ at the same FPR. This is expected due to the optimal use of a larger amount of information. Nevertheless, even under Definition \ref{def:Correction}, $X$ and $Z$ are not comparable in Blackwell's order without additional assumptions, meaning it is not trivial to ensure which one performs better. For an introduction to the Blackwell-Sherman-Stein theorem see \cite{le1996comparison}, or \cite[Ch. 3]{torgersen1991comparison} for further detail.

\section{Proofs of the results in the main text}\label{sec:AppProofs}\label{sec:AppProofs_ComparingScores}

\lemmaEqualCovariateSpecificRocs*
\begin{proof}
     Under Definition \ref{def:Correction},  $X\mid V=v,D$ and $Z\mid V=v,D$ are almost surely related by a strictly increasing differentiable transformation $g_v$. Therefore, covariate-specific ROC curves of $X$ and $Z$ are identical across all $V$ levels.
\end{proof}
\lemmaIndependenceImplication*
\begin{proof}
    This trivially holds for the standardization in \eqref{eq:AgeEduCorrection}, but it also holds in general, as in this case $F_{X,0}(\cdot\mid v)\equiv F_{X,0}(\cdot)$ and therefore $g_v(x)=T(F_{X,0}(\cdot\mid v),v,x)=\Tilde{T}(F_{X,0}(\cdot),x)=\Tilde{g}(x)$ for all $x$ and $v$, meaning $Z=\Tilde{g}(X)$ almost surely, and therefore $Z\independent V|D$.
\end{proof}
\coroEquality*
\begin{proof}
    Under the conditions of Lemma \ref{lemma:IndependenceImplication}, if $X\independent V\mid D$, then  $Z=\Tilde{g}(X)$ almost surely with $\Tilde{g}$ strictly increasing. The result then follows from \cite[Ch. 2.2.3, Property 2]{krzanowski2009roc}.
\end{proof}
Before giving the proof of Theorem \ref{thm:Superiority}, we state and prove three auxiliary Lemmas, which will also be useful in the general study provided in Appendix \ref{sec:AppComparingGeneral}.
\begin{lemma}\label{lemma:CovariateRocConnection}
    Let $\alpha\in (0,1)$. For  $S\in\{X,Z\}$ define the marginal randomized rule $\psi_\alpha(S)$ as per Equation \eqref{eq:RandomizedDecision}. Then for all $v\in\mathcal{V}$  $$\mathbb{E}[\psi_\alpha(S)|V=v,D=1]=ROC_v(\alpha^S_v),$$ where $\alpha^S_v=\mathbb{E}[\psi_\alpha(S)|D=0,V=v]$.
\end{lemma}
\begin{proof}
    Recall that $\psi_\alpha(S)=I(S<t_\alpha)+\rho_\alpha I(S=t_\alpha)$, where $t_\alpha=F_{S,0}^{-1}(\alpha)$ and 
    \begin{equation*}
        \rho_\alpha=\frac{\alpha-\mathbb{P}(S<t_\alpha|D=0)}{\mathbb{P}(S=t_\alpha|D=0)}.
    \end{equation*}
    Furthermore we have $$\alpha_v^S=\mathbb{P}(S<t_\alpha|D=0,V=v)+\rho_\alpha \mathbb{P}(S=t_\alpha|D=0,V=v).$$ On the one hand  consider the conditional power of the marginal test $$\mathbb{E}[\psi_\alpha(S)|V=v,D=1]=\mathbb{P}(S<t_\alpha|D=1,V=v)+\rho_\alpha \mathbb{P}(S=t_\alpha|D=1,V=v).$$ On the other hand, for $\alpha'\in(0,1)$, the $V$-specific randomized decision rule with covariate specific FPR $\alpha'$ is $\psi_{\alpha'}(S|v)=I(S<t_{\alpha'}(v))+\rho_{\alpha'}(v) I(S=t_{\alpha'}(v))$ where $t_{\alpha'}(v)=F_{S,0}^{-1}(\alpha'|v)$ and $$\rho_{\alpha'}(v)=\frac{\alpha'-\mathbb{P}(S<t_{\alpha'}(v)|V=v,D=0)}{\mathbb{P}(S=t_{\alpha'}(v)|V=v,D=0)},$$ which leads to $$ROC_v(\alpha')=\mathbb{P}(S<t_{\alpha'}(v)|D=1,V=v)+\rho_{\alpha'}(v) \mathbb{P}(S=t_{\alpha'}(v)|D=1,V=v).$$ Therefore, it suffices to show that 1) $t_{\alpha_v^S}(v)=t_\alpha$, and 2) $\rho_{\alpha^S_v}(v)=\rho_\alpha$.  Consider first the case of a discrete score. For 1) we first note that as $\rho_\alpha\in(0,1]$ we have 
    \begin{equation}\label{eq:ProbGap}
        \mathbb{P}(S< t_\alpha|D=0,V=v)<\alpha^S_v\leq F_{S,0}(t_\alpha|v)
    \end{equation}
    and therefore $t_\alpha=F_{S,0}^{-1}(\alpha^S_v|v)=:t_{\alpha^S_v}(v)$. The proof of 2) follows from 1) as 
    \begin{align*}
        \rho_{\alpha^S_v}(v)&=\frac{\alpha^S_v-\mathbb{P}(S<t_{\alpha^S_v}(v)|V=v,D=0)}{\mathbb{P}(S=t_{\alpha^S_v}(v)|V=v,D=0)}\\&=\frac{\mathbb{P}(S<t_\alpha|D=0,V=v)+\rho_\alpha \mathbb{P}(S=t_\alpha|D=0,V=v)}{\mathbb{P}(S=t_{\alpha^S_v}(v)|V=v,D=0)}\\&\quad -\frac{\mathbb{P}(S<t_{\alpha^S_v}(v)|V=v,D=0)}{\mathbb{P}(S=t_{\alpha^S_v}(v)|V=v,D=0)}\\&=\frac{\mathbb{P}(S<t_\alpha|D=0,V=v)+\rho_\alpha \mathbb{P}(S=t_\alpha|D=0,V=v)}{\mathbb{P}(S=t_\alpha|V=v,D=0)}\\&\quad -\frac{\mathbb{P}(S<t_\alpha|V=v,D=0)}{\mathbb{P}(S=t_\alpha|V=v,D=0)}\\&=\rho_\alpha.
    \end{align*}
    In the absolutely continuous case,  Equation \eqref{eq:ProbGap} does not hold with inequalities but rather with two equalities. Yet, as in this case we consider $\rho_\alpha=1$ we have $\alpha^S_v= F_{S,0}(t_\alpha|v)$ what leads to 1), and 2) holds trivially as by convention $\rho_\alpha=1=\rho_{\alpha^S_v}(v)$.
\end{proof}
\begin{lemma}\label{lemma:SuperGradient}
    Consider a score $S$ such that its associated ROC curve is concave. Fix a FPR $\alpha\in(0,1)$. Then $LR_S(F_{S,0}^{-1}(\alpha))$ is a valid super-gradient at $\alpha$, meaning $$LR_S(F_{S,0}^{-1}(\alpha))\in\partial ROC_S(\alpha).$$
\end{lemma}
\begin{proof}
    If $S|D$ has an absolutely continuous distribution and $ROC_S$ is differentiable at $\alpha$, then we know that the super-differential is precisely the derivative $\partial ROC_S(\alpha)=\{ROC_S'(\alpha)\}$ \cite[Part I, Ch. VI, Cor. 2.1.4]{hiriart2013convex},  but the derivative of a differentiable ROC curve  at FPR $\alpha$ is the likelihood ratio evaluated at the threshold leading to said FPR \cite[Ch. 2.2.3]{krzanowski2009roc}: $ROC_S'(\alpha)=LR_S(F_{S,0}^{-1}(\alpha))$.

     If $S|D$ has an absolutely continuous distribution but $ROC_S$ is not differentiable at $\alpha$, the result still holds because, given $\beta>\alpha$ we have
     \begin{align*}
         ROC_S(\beta)-ROC_S(\alpha)&=F_{S,1}(F_{S,0}^{-1}(\beta))-F_{S,1}(F_{S,0}^{-1}(\alpha))\\ &=\int_{F_{S,0}^{-1}(\alpha)}^{F_{S,0}^{-1}(\beta)}f_{S|D}(s|1)ds\\&=\int_{F_{S,0}^{-1}(\alpha)}^{F_{S,0}^{-1}(\beta)}LR_S(s)f_{S|D}(s|0)ds\\&\leq LR_S(F_{S,0}^{-1}(\alpha))\int_{F_{S,0}^{-1}(\alpha)}^{F_{S,0}^{-1}(\beta)}f_{S|D}(s|0)ds\\&=LR_S(F_{S,0}^{-1}(\alpha))(\beta-\alpha),
     \end{align*}
     where the inequality follows from the fact that $ROC_S$ being concave is equivalent to $LR_S$ being non-increasing \cite{gneiting2022receiver}. This implies that $$ROC_S(\beta)\leq ROC_S(\alpha)+LR_S(F_{S,0}^{-1}(\alpha))(\beta-\alpha)$$ for $\beta>\alpha$. The proof for $\beta<\alpha$ follows analogously. This implies   $LR_S(F_{S,0}^{-1}(\alpha))\in\partial ROC_S(\alpha)$, as we wanted to show.

     If $S|D$ has a discrete support, we considered the ROC curves as originating from linearly interpolating the points in the ROC set. In this case, the slope of $ROC_S$ at $\alpha$ from the left (which is the largest super-differential thanks to the concavity assumption \cite{gneiting2022receiver} and the continuous piece-wise linear curve) is precisely $LR_S(F_{S,0}^{-1}(\alpha))$, which thanks to the concavity assumption is a valid super-differential (the largest one).   This leads to the same conclusion as in the continuous case.
\end{proof}
\begin{restatable}{lemma}{lemmaSuperiority}\label{lemma:Superiority}
    Let $\alpha\in (0,1)$. For  $S\in\{X,Z\}$ define $t_S(\alpha)=F_{S,0}^{-1}(\alpha)$ the $\alpha$-quantile of $S\mid D=0$, and for each $v\in\mathcal{V}$ define $\alpha^S_v=\mathbb{E}[\psi_\alpha(S)|D=0,V=v]$   for score $S$ as in Assumption \ref{ass:Comonotonicity}, where $\psi_\alpha(\cdot)$ is defined as in Equation \eqref{eq:RandomizedDecision}. Assume that 
    \begin{enumerate}
        \item $ROC_v$ is concave for all $v\in\mathcal{V}$,
        \item $\mathbb{E}[LR_{(X,V)}(t_X(\alpha),V)(\alpha^X_V-\alpha^Z_V)\mid D=0]\geq 0$, 
    \end{enumerate}
    where $LR_{(X,V)}(x,v)=f_{(X,V)\mid D}(x,v\mid 1)/f_{(X,V)\mid D}(x,v\mid 0)$ is the joint likelihood ratio.  Then $X$ performs better marginally at level $\alpha$, meaning $$ROC_X(\alpha)\geq ROC_Z(\alpha).$$
\end{restatable}
\begin{proof}
    The concavity assumption and the super-differential inequality \cite[Part I, Ch. VI]{hiriart2013convex} (namely the sub-differential inequality applied to $-ROC_v$) tell us that for all $v\in\mathcal{V}$
    $$ROC_v(\alpha^Z_v) \leq ROC_v(\alpha^X_v) + s_v(\alpha^Z_v-\alpha^X_v), $$
    or equivalently that
    $$ROC_v(\alpha^X_v)-ROC_v(\alpha^Z_v)\geq s_v(\alpha^X_v-\alpha^Z_v).$$ 
    Therefore, for the marginal ROC curves we have
    \begin{align*}
        ROC_X(\alpha)-ROC_Z(\alpha)&=\mathbb{E}[\psi_\alpha(X)|D=1]-\mathbb{E}[\psi_\alpha(Z)|D=1]\\&=\mathbb{E}[\mathbb{E}[\psi_\alpha(X)|V,D=1]-\mathbb{E}[\psi_\alpha(Z)|V,D=1]|D=1]\\&=\mathbb{E}[ROC_V(\alpha^X_V)-ROC_V(\alpha^Z_V)|D=1]\\&\geq \mathbb{E}[s_V(\alpha^X_V-\alpha^Z_V)|D=1]\\&= \mathbb{E}[LR_V(V)s_V(\alpha^X_V-\alpha^Z_V)|D=0],
    \end{align*}
    where the third equality follows from Lemma \ref{lemma:CovariateRocConnection}. By Lemma \ref{lemma:SuperGradient} we know that $s_v=LR_{X|V}(F_{X,0}^{-1}(\alpha^X_v|v)|v)$ is a valid super-gradient. Furthermore, in the proof of Lemma \ref{lemma:CovariateRocConnection} we also showed that $F_{X,0}^{-1}(\alpha^S_v|v)=F_{X,0}^{-1}(\alpha)$. This concludes the proof.
\end{proof}
Using  Lemma \ref{lemma:Superiority} we will be able to prove Theorem \ref{thm:Superiority} as follows. 
\thmSuperiorityExtended*
\begin{proof}
    Under Definition \ref{def:Correction} and Assumption \ref{ass:TestWellDesigned}, Lemma \ref{lemma:Superiority} provides us with the desired lower bound, as  for all $v$ under Assumption \ref{ass:WellDesignedCorrection} $$\alpha^Z_v=\mathbb{E}[\psi_\alpha(Z)|D=0,V=v]=\mathbb{E}[\psi_\alpha(Z)|D=0]=\alpha.$$ Regarding the superiority, again under Assumption \ref{ass:TestWellDesigned}, Lemma \ref{lemma:Superiority} tells us that a sufficient condition for $ROC_X(\alpha)\geq ROC_Z(\alpha)$ is that $$\Delta_\alpha=\mathbb{E}[LR_{V\mid X}(V\mid t_X(\alpha))(\alpha_V^X-\alpha)|D=0]\geq 0,$$ as $LR_X(t_X(\alpha))\geq 0$ does not depend on $v$. Under Assumption \ref{ass:WellDesignedCorrection} we know that $\alpha_v^Z=\alpha$ for all $v\in\mathcal{V}$.  Bayes' rule  implies that $LR_{V\mid X}(v\mid t_X(\alpha))$ is co-monotone in $v$ with the conditional probability $\mathbb{P}(D=1\mid {X=F_{X, 0}^{-1}(\alpha)},V=v)$.    Assumption \ref{ass:Comonotonicity} precisely states that $\mathbb{P}(D=1\mid {X=F_{X, 0}^{-1}(\alpha)},V=v)$ and   $\alpha_v^X$  are co-monotone   in $v$, and so they are with $\alpha_v^X-\alpha$.      Chebyshev's weighted inequality (see \cite[Thm. 43 and Thm. 167]{hardy1952inequalities} for the discrete case or \cite{armstrong1993chebyshev} for a more general statement) implies that 
    \begin{align*}
        \Delta_\alpha&=\mathbb{E}[LR_{V\mid X}(V\mid t_X(\alpha))(\alpha_V^X-\alpha)|D=0]\\&\geq\mathbb{E}[LR_{V\mid X}(V\mid t_X(\alpha))|D=0]\mathbb{E}[\alpha_V^X-\alpha|D=0]\\&= 0,
    \end{align*}
    where the last equality follows from the fact that $\mathbb{E}[\alpha_V^X-\alpha|D=0]=0$ because $$\mathbb{E}[\alpha_V^X|D=0]=\mathbb{E}[\mathbb{E}[\psi_\alpha(X)|V,D=0]|D=0]=\mathbb{E}[\psi_\alpha(X)|D=0]=\alpha.$$
\end{proof}

\begin{remark}
    In the case when the score has a discrete support and $\alpha$ is chosen such that $F_{X,0}(F_{X,0}^{-1}(\alpha))=\alpha$ (i.e. we choose a threshold rather than a FPR) then the bound given in Equation \eqref{eq:LowerBound} can be slightly improved. Indeed, if we re-examine the proof of Lemma \ref{lemma:SuperGradient} in the discrete case for such an $\alpha$, we see that $$s_v\in \partial ROC_v(\alpha^X_v)=[LR_{X|V}(F_{X,0}^{-1}(\alpha+\varepsilon)|v),LR_{X|V}(F_{X,0}^{-1}(\alpha)|v)],$$ for  a certain $\varepsilon>0$ small enough. Essentially, $LR_{X|V}(F_{X,0}^{-1}(\alpha+\varepsilon)|v)$ is the slope of the linearly interpolated covariate-specific ROC curve at $\alpha$ from the right, as $F_{X,0}^{-1}(\alpha+\varepsilon)$ is the smallest value larger than $F_{X,0}^{-1}(\alpha)$ that $X$ can take. As we aim for a lower bound, we could try to maximize $s_v(\alpha^X_v-\alpha)$ for each $v$, in the sense that if $\alpha^X_v>\alpha$ then we would choose $s_v=LR_{X|V}(F_{X,0}^{-1}(\alpha)|v)$ (as we have done so far), but when $\alpha^X_v<\alpha$ then we would take $s_v=LR_{X|V}(F_{X,0}^{-1}(\alpha+\varepsilon)|v)$ the smallest valid super-gradient. In Section \ref{sec:OASIS3_Superiority} we have chosen not to implement this improved bound due to the non-differentiability of the super-gradient decision, which complicates the asymptotic analysis of Appendix \ref{sec:AppOASIS_SuperiorityCLT}.
\end{remark}

\thmSuperiorityZ*
\begin{proof}
    The proof follows analogously as that of Theorem \ref{thm:Superiority} but interchanging the roles of $X$ and $Z$. To then apply the Weighted Chebyshev inequality one requires the counter-monotonicity granted by Assumption \ref{ass:Countermonotonicity} due to change in sign of the lower bound.
\end{proof}

\coroEqualityweak*
\begin{proof}
    Follows trivially from Theorems \ref{thm:Superiority} and \ref{thm:Superiority_Z}.
\end{proof}

\coroAUC*
\begin{proof}
    Follows trivially from Theorems \ref{thm:Superiority} and \ref{thm:Superiority_Z}, the definition of the AUC and monotonicity of the integral.
\end{proof}

\section{Comparing ROC curves: The case with arbitrary covariates.}\label{sec:AppComparingGeneral}
\subsection{The general case}\label{sec:AppComparingGeneral_GeneralCase}
In Section \ref{sec:ComparingScores},  we have assumed $V$ to be univariate and supported on a subset of the real line. The main goal of this Assumption was to simplify the exposition of Assumption \ref{ass:Comonotonicity}. This Assumption requires further examination in the case when $V$ is multivariate and/or its support lacks a total order. 

A special case is that when all components of $V$ have a discrete support. In this case, taking the Cartesian product of these variables leads to a univariate discrete random variable $R$ equivalent to $V$. In this case, one can substitute Assumption \ref{ass:Comonotonicity} by ``there is an ordering of the labels of $R$ such that $\mathbb{E}[\psi_\alpha(X)|D=0,R=r]$ and $\mathbb{P}(D=1\mid {X=F_{X, 0}^{-1}(\alpha)},R=r)$ are co-monotone in $r$.'' One can easily check that the proofs in Appendix \ref{sec:AppProofs_ComparingScores} still stand in this case.

However, when $V$ is multivariate and has one or more continuous components, the former argument does not hold, as co-monotonicity is not straightforwardly defined. We now present a new set of conditions which apply in the general case where we remain agnostic to the nature of the distribution of $V$.

Consider then the case where $V=(V_1,\ldots,V_p)$ supported on $\mathcal{V}$. We will use underlines and overlines to denote posterior and previous components, in the sense that $\underline{V}_k=(V_k,\ldots,V_p)$ and $\overline{V}_k=(V_1,\dots,V_k)$, supported on $\underline{\mathcal{V}}_k$ and $\overline{\mathcal{V}}_k$ respectively.

\begin{assumption}\label{ass:SufficientConditionsGeneral}
    For a given $\alpha\in(0,1)$
    \begin{enumerate}
        \item[(a)] $LR_{X\mid V}(\cdot\mid v)$ is  decreasing for all $v$ in $\mathcal{V}$,
    \end{enumerate}
    and for all $k$ in $\{1,\ldots,p\}$ 
    \begin{enumerate}  
        \item[(b.1)] If $V_k$ is a continuous random variable, then  $\mathbb{E}_0[\alpha^X_V-\alpha^Z_V\mid V_k=v_k,\underline{V}_{k+1}]$ and  $\mathbb{E}_0[LR_{V\mid X}(V\mid t_X(\alpha))\mid V_k=v_k,\underline{V}_{k+1} ]$  are co-monotone in $v_k$ almost surely w.r.t.\ $\underline{V}_{k+1}$,
        \item [(b.2)] If $V_k$ has a discrete support, then  there is an ordering of the labels of $V_k$ such that  $\mathbb{E}_0[\alpha^X_V-\alpha^Z_V\mid V_k=v_k,\underline{V}_{k+1}]$ and  $\mathbb{E}_0[LR_{V\mid X}(V\mid t_X(\alpha))\mid V_k=v_k,\underline{V}_{k+1} ]$  are co-monotone in $v_k$ almost surely w.r.t.\ $\underline{V}_{k+1}$,
    \end{enumerate}
    where by $\mathbb{E}_0$ we refer to the expectation conditional on $D=0$.
\end{assumption}
\begin{remark}
    We see how Assumptions \ref{ass:TestWellDesigned} and  \ref{ass:SufficientConditionsGeneral} (a) are identical, and that in the case $p=1$ Assumptions \ref{ass:Comonotonicity} and  \ref{ass:SufficientConditionsGeneral} (b) are equivalent under Assumption \ref{ass:WellDesignedCorrection}. While in the main text we could have dropped Assumption \ref{ass:WellDesignedCorrection} by strengthening Assumption \ref{ass:Comonotonicity}, we have chosen to present them as so because the conditional independence of the corrected score is usually expected by practitioners. 
\end{remark}
\begin{remark}
    Even though we presented Lemmas \ref{lemma:CovariateRocConnection}-\ref{lemma:Superiority} as auxiliary to the proof of Theorem \ref{thm:Superiority} where we assume $V$ to have a  univariate and absolutely continuous distribution, we can see from their proofs that they apply to the general context of a multivariate $V$ as we treat in this section. The only requirement is that there are no positivity issues so that all likelihood ratios are well defined. 
\end{remark}
\begin{theorem}\label{thm:SuperiorityGeneral}
    Fix $\alpha\in (0,1)$. Under Definition \ref{def:Correction} and Assumption \ref{ass:SufficientConditionsGeneral} it holds that $$ROC_X(\alpha)\geq ROC_Z(\alpha).$$ 
\end{theorem}
\begin{proof}
    By Lemma \ref{lemma:Superiority} is suffices to show that $\mathbb{E}_0[LR_{V\mid X}(V\mid t_X(\alpha))(\alpha^X_V-\alpha^Z_V)]\geq 0.$ We have that
    \begin{align*}\allowdisplaybreaks
        &\mathbb{E}_0[LR_{V\mid X}(V\mid t_X(\alpha))(\alpha^X_V-\alpha^Z_V)]\\=& \mathbb{E}_0[\mathbb{E}_0[LR_{V\mid X}(V\mid t_X(\alpha))\mid \underline{V}_{1}]\mathbb{E}_0[(\alpha^X_V-\alpha^Z_V)\mid \underline{V}_{1} ]]\\=&\mathbb{E}_0[\mathbb{E}_0[ \mathbb{E}_0[LR_{V\mid X}(V\mid t_X(\alpha))\mid \underline{V}_{1}]\mathbb{E}_0[(\alpha^X_V-\alpha^Z_V)\mid \underline{V}_{1} ] \mid \underline{V}_{2} ]]\\ \geq & \mathbb{E}_0[\mathbb{E}_0[ \mathbb{E}_0[LR_{V\mid X}(V\mid t_X(\alpha))\mid \underline{V}_{1}]\mid \underline{V}_{2} ]\mathbb{E}_0[\mathbb{E}_0[(\alpha^X_V-\alpha^Z_V)\mid \underline{V}_{1} ] \mid \underline{V}_{2} ]]\\=&\mathbb{E}_0[\mathbb{E}_0[ LR_{V\mid X}(V\mid t_X(\alpha))\mid \underline{V}_{2} ]\mathbb{E}_0[(\alpha^X_V-\alpha^Z_V) \mid \underline{V}_{2} ]]\\=&\mathbb{E}_0[\mathbb{E}_0[\mathbb{E}_0[ LR_{V\mid X}(V\mid t_X(\alpha))\mid \underline{V}_{2} ]\mathbb{E}_0[(\alpha^X_V-\alpha^Z_V) \mid \underline{V}_{2} ]\mid \underline{V}_{3}]]\\\geq& \mathbb{E}_0[\mathbb{E}_0[\mathbb{E}_0[ LR_{V\mid X}(V\mid t_X(\alpha))\mid \underline{V}_{2} ]\mid \underline{V}_{3}]\mathbb{E}_0[\mathbb{E}_0[(\alpha^X_V-\alpha^Z_V) \mid \underline{V}_{2} ]\mid \underline{V}_{3}]]\\=& \mathbb{E}_0[\mathbb{E}_0[ LR_{V\mid X}(V\mid t_X(\alpha))\mid \underline{V}_{3}]\mathbb{E}_0[(\alpha^X_V-\alpha^Z_V) \mid \underline{V}_{3}]]\\ &[\text{iterating this procedure until we get to } p+1]\\ \geq & \mathbb{E}_0[ LR_{V\mid X}(V\mid t_X(\alpha))]\mathbb{E}_0[(\alpha^X_V-\alpha^Z_V) ]=0,
    \end{align*}
    where we repeatedly applied the law of iterated expectations plus Chebyshev's weighted inequality  \cite{armstrong1993chebyshev} under Assumption \ref{ass:SufficientConditionsGeneral} \textit{(b)}. Together with Definition \ref{def:Correction} and  Assumption \ref{ass:SufficientConditionsGeneral} \textit{(a)}, the conditions for  Lemma \ref{lemma:Superiority} are satisfied. This concludes the proof.
\end{proof}
\begin{remark}
    It is worth noting  that 
    \begin{align*}
        \mathbb{E}_0[\alpha_V^X \mid V_k=v_k,\underline{V}_{k+1}]&=\mathbb{E}_0[\mathbb{E}_0[\psi_\alpha(X)\mid V]\mid V_k=v_k,\underline{V}_{k+1}]\\&=\mathbb{E}_0[\psi_\alpha(X)\mid V_k=v_k,\underline{V}_{k+1}]
    \end{align*}
    meaning $\mathbb{E}_0[\alpha_V^X \mid \underline{V}_{k}=\underline{v}_k]$ represents a covariate-specific FPR marginalized among $V_1,\ldots,V_{k-1}$.
\end{remark}

\subsection{The jointly independent case}\label{sec:AppComparingGeneral_JointlyIndependent}
Assumption \ref{ass:SufficientConditionsGeneral} \textit{(b)} accounts for the dependence structure between the components of $V$ by taking sequential conditional expectations. 
In the special case where the components of $V$ are jointly independent among the non-impaired, Harris' inequality \cite{harris1960lower} provides a more compact condition, as we can see in the following conditions.
\begin{assumption}\label{ass:SufficientConditionsIndependence}
    For a given $\alpha\in(0,1)$ assume
    \begin{enumerate}
        \item[(a)] $LR_{X\mid V}(\cdot\mid v)$ is  decreasing for all $v$ in $\mathcal{V}$,
        \item[(b)] The random vector $V\mid D=0$ has jointly independent components,
    \end{enumerate}
    and for all $k$ in $\{1,\ldots,p\}$ 
    \begin{enumerate}
        \item[(c.1)] If $V_k$ is a continuous random variable then  $(\alpha^X_v-\alpha^Z_v)$ and  $LR_{V\mid X}(v\mid t_X(\alpha))$  are  co-monotone in $v_k$ for all $v_1,\ldots,v_{k-1},v_{k+1},\ldots,v_p$\footnote{We understand this co-monotonicity as to be ``in the same direction in $v_k$'' for all $v_1,\ldots,v_{k-1},v_{k+1},\ldots,v_p$; meaning always both increasing or decreasing. We do not allow the co-monotonicity direction to change with $v_1,\ldots,v_{k-1},v_{k+1},\ldots,v_p$.  }.
        \item[(c.2)] If $V_k$ has a discrete support then  there is an ordering of the labels of $V_k$ such that $(\alpha^X_v-\alpha^Z_v)$ and  $LR_{V\mid X}(v\mid t_X(\alpha))$  are  co-monotone in $v_k$ for all $v_1,\ldots,v_{k-1},v_{k+1},\ldots,v_p$.
    \end{enumerate}
\end{assumption}
\begin{theorem}
    Fix $\alpha\in (0,1)$. Under Definition \ref{def:Correction} and Assumption \ref{ass:SufficientConditionsIndependence} it holds that $$ROC_X(\alpha)\geq ROC_Z(\alpha).$$ 
\end{theorem}
\begin{proof}
    Analogously to the proof of Theorem \ref{thm:SuperiorityGeneral}, by Lemma \ref{lemma:Superiority} is suffices to show that $\mathbb{E}_0[LR_{V\mid X}(V\mid t_X(\alpha))(\alpha^X_V-\alpha^Z_V)]\geq 0.$ This follows from Assumptions \ref{ass:SufficientConditionsIndependence} \textit{(b)} and \textit{(c)}, as then thanks to Harris' inequality \cite{harris1960lower,kleitman1966families}
    $$\mathbb{E}_0[LR_{V\mid X}(V\mid t_X(\alpha))(\alpha^X_V-\alpha^Z_V)]\geq \mathbb{E}_0[LR_{V\mid X}(V\mid t_X(\alpha))]\mathbb{E}_0[(\alpha^X_V-\alpha^Z_V)]=0.$$
    This concludes the proof.
\end{proof}
\begin{remark}
    We see again how in the  univariate case $p=1$ treated in the main text Assumption \ref{ass:SufficientConditionsIndependence} (b) is trivially satisfied, and (c) becomes equivalent to Assumption \ref{ass:Comonotonicity} under Assumption \ref{ass:WellDesignedCorrection}. 
\end{remark}

\section{Demographic insensitivity}\label{sec:App_FDP}
\subsection{Partial isolation}\label{sec:App_FDP_PartialIsolation}

To study the counterfactual distribution of interest (Definitions \ref{def:FDP} and \ref{def:FDP_weak}) we will \textit{not} assume a specific causal DAG to be correct. However, we will still discuss how $V$ and $V_S$ affect both $S$ and $D$ in the extended causal DAG.  In general, we will assume the sensitive attributes $V_S$  to have no causal effect on $D$. For example, if the demographics $V$  represent the years of formal education, the sensitive attribute $V_S$ can represent knowledge of certain words. This sensitive attribute is expected to have no effect on developing cognitive impairment, while the actual education likely does. Knowing certain words per se has no effect on impairment, but having studied for several years is a protective factor against cognitive decline. We can formalize these ideas using the notions of partial isolation \cite{stensrud2021generalized}.

\begin{definition}[$V_S$ partial isolation]\label{def:XPI}
    There are no causal paths from $V_S$ to $D$ in the extended causal DAG.
\end{definition}
\begin{definition}[$V$ partial isolation]\label{def:DPI}
    All causal paths from $V$ to $S$ in the extended causal DAG intersect\footnote{ We say that a path in a DAG with node set $N$ intersects  $U\subset N$ when it contains one of the nodes in $U$, not being the terminal nodes of the path.} $D$.
\end{definition}

If both $V_S$ and $V$ partial isolations hold, we say that \textit{full isolation} holds. See that $V_S$ partial isolation implies that all causal paths from $V_S$ to $S$ in the extended DAG  represent effects not through the impairment status which we aim to predict with $S$. Analogously, if $V$ partial isolation holds, the demographics $V$ have all its effect on the score mediated by $D$. We will see that $V$ partial isolation will not be required for identifying the counterfactual distribution $S^{v_S}\mid \{V=v,D=d\}$, yet we will make assumptions equivalent to $V_S$ partial isolation. Such assumption is not required in most of the existing separable effects literature \cite{stensrud2021generalized,stensrud2022separable,stensrud2023conditional,wen2024causal}. This owes to the fact that  Definitions \ref{def:FDP} and \ref{def:FDP_weak} target the distribution of a counterfactual outcome conditional on the factual event $\{V=v,D=d\}$, which we will need to be observable in the counterfactual world under an intervention $V_S=v_S$. As we show in Appendix \ref{sec:App_FDP_Identification}, this boils down to there being no direct effect from $V_S$ on $D$,  which is the notion of $V_S$ partial isolation. We provide further discussion in Propositions \ref{prop:XPI} and \ref{prop:FullIsolation}.

\subsection{Identification}\label{sec:App_FDP_Identification}
To study whether a certain score $S$ is  demographic insensitive, we need to identify the distribution of $S^{v_S}\mid \{V=v,D=d\}$. In the case of a decision rule based on thresholding at $s$, it suffices to identify $\mathbb{P}(S^{v_S}\leq s\mid V=v,D=d)$. Thus, we see that regardless of whether we target our strong or weak notion of demographic insensitivity, our target for identification and inference will be $\{\mathbb{P}(S^{v_S}\leq s\mid V=v,D=d)\colon s\in {\rm supp}(S)\}$. If we had access to an experiment where $V$ and $V_S$ were separately observed, $V_S$ was randomized, and $V_S$ partial isolation held, then we would be able to identify the  target distribution as $$\mathbb{P}(S\leq s\mid V_S=v_S,V=v,D=d),$$ for all $s$ in the support of $S$, as the required conditions (consistency, exchangeability and positivity) would hold by design \cite{hernan2010causal}. Thus, this is a single-world estimand which, in principle, can be identified from a randomized experiment if $V_S$ partial isolation can be accepted. 

When only data from experiments where there is no direct access to $V_S$ (meaning $V_S\equiv V$) is available, identifying the distribution of $S^{v_S}\mid \{V=v,D=d\}$ might not be guaranteed under the previous assumptions when $v_S\neq v$.  We will now give further assumptions that are sufficient for identification in this setting.
\begin{assumption}[Positivity]\label{ass:Positivity}
    For all $v\in\mathcal{V}$ and $d\in\{0,1\}$ 
    \begin{subequations}\label{eq:Positivity}
        \begin{equation}\label{eq:PositivityBase}
            f_{V,D}(v,d)>0,
        \end{equation}
        and for all $l\in\mathcal{L}$ and $w\in\mathcal{W}$
        \begin{equation}\label{eq:PositivityV}
            \mathbb{P}(L=l,W=w,D=d)>0\Rightarrow f_{V|L,W,D}(v\mid l,w,d)>0.
        \end{equation}
    \end{subequations}
\end{assumption}
Condition \eqref{eq:PositivityBase} is required to be able to condition on the event $\{V=v,D=d\}$, while \eqref{eq:PositivityV} states that for any covariate and impairment status we find individuals in every stratum of the demographic variables. These two conditions refer only to the data we have access to, and can be tested.

\begin{assumption}[Consistency]\label{ass:Consistency}
    If $V_S=v_S$ then $S^{v_S}=S$ and $W^{v_S}=W$ for all 
    $v_S\in\mathcal{V}$.
\end{assumption}
This is the standard consistency assumption made in the causal inference literature for interventions on $V_S$. It requires that interventions on $V_S$ are well-defined, and makes an implicit no-interference assumption. Both are expected to hold given the intervening variable construction of $V_S$ \cite{wen2024causal} and the i.i.d. assumption.

Next, we will consider a set of conditional exchangeability assumptions. First consider that the  covariates $W$ can be expressed as having two components $W=(W_S,W_D)$. Furthermore, following \cite{stensrud2021generalized, wen2024causal} denote by ``$(G)$'' the experiment where $V_S$ is randomly assigned, and consider the following dismissible component conditions.
\begin{assumption}[Dismissible component conditions]\label{ass:DCC}
    \begin{subequations}\label{eq:DCC}
        \begin{align}
            S(G)&\independent V(G)\mid V_S(G),D(G),W(G),L(G),\label{eq:DCC_X}\\
            W_D(G)&\independent V_S(G)\mid V(G),L(G),D(G),W_S(G),\label{eq:DCC_WD}\\
            W_S(G)&\independent V(G)\mid V_S(G),D(G),L(G),\label{eq:DCC_WX}\\
            D(G)&\independent V_S(G)\mid L(G),V(G).\label{eq:DCC_D}
        \end{align}
    \end{subequations}
\end{assumption}
While conditions \eqref{eq:DCC} do not hold by design, they are observational and could be tested on real data coming from the experiment $(G)$, which is realizable in principle as $V_S$ could be randomized.

We will show that the dismissible component conditions are closely related to the notion of partial isolation, as is the case in the separable effects literature. Yet, when studying this relation, \cite[App. C]{stensrud2021generalized} implicitly makes a faithfulness assumption. Faithfulness is the  counterpart of the global Markov property, and will allow us to translate conditional independences into graphical statements related to partial isolation. To that extent, denote by $\mathcal{G}_{(G)}$ the (extended) causal DAG representing the variables of the experiment $(G)$. An example of these DAGs can be seen in Figure \ref{fig:FDP_ExtendedDAG_G} as a continuation of the example introduced in Figure \ref{fig:FDP_ExtendedDAG}.

\begin{figure}
\centering

\begin{subfigure}{0.45\textwidth}
\centering
\begin{tikzpicture}
\begin{scope}[every node/.style={thick,draw=none}]
    \node (L) at (-1,0) {$L$};
    \node (V) at (1,-1) {$V$};
    \node (VS) at (1,1) {$V_S$};
    \node (S) at (5,1) {$S$};
    \node (D) at (5,-1) {$D$};
    \node (W) at (3,0) {$W$};
\end{scope}

\begin{scope}[>={Stealth[black]},
              every node/.style={fill=white,circle},
              every edge/.style={draw=black,very thick}]
	\path [->] (L) edge (V);
    \path [->] (L) edge (VS);
    \path [->] (L) edge (W);
    \path [->] (L) edge[bend left=50] (S);
    \path [->] (L) edge[bend right=50] (D);
    \path [->] (V) edge[line width=0.85mm] (VS);
    \path [->] (VS) edge (S);
    \path [->] (V) edge (D);
    \path [->] (D) edge (S);
    \path [->] (W) edge (S);
    \path [->] (W) edge (D);
    \path [->] (V) edge (W);

\end{scope}
\end{tikzpicture}
\caption{Extended causal DAG $\mathcal{G}$}
\end{subfigure}
\hfill
\begin{subfigure}{0.45\textwidth}
\centering
\begin{tikzpicture}
\begin{scope}[every node/.style={thick,draw=none}]
    \node (L) at (-1,0) {$L(G)$};
    \node (V) at (1,-1) {$V(G)$};
    \node (VS) at (1,1) {$V_S(G)$};
    \node (S) at (5,1) {$S(G)$};
    \node (D) at (5,-1) {$D(G)$};
    \node (W) at (3,0) {$W(G)$};
\end{scope}

\begin{scope}[>={Stealth[black]},
              every node/.style={fill=white,circle},
              every edge/.style={draw=black,very thick}]
	\path [->] (L) edge (V);
    \path [->] (L) edge (W);
    \path [->] (L) edge[bend left=50] (S);
    \path [->] (L) edge[bend right=50] (D);
    \path [->] (VS) edge (S);
    \path [->] (V) edge (D);
    \path [->] (D) edge (S);
    \path [->] (W) edge (S);
    \path [->] (W) edge (D);
    \path [->] (V) edge (W);

\end{scope}
\end{tikzpicture}
\caption{Extended causal DAG $\mathcal{G}_{(G)}$}
\end{subfigure}

\caption{\small Extended causal DAG introduced in Figure \ref{fig:FDP_ExtendedDAG} (A) together with its associated extended DAG in the related $(G)$ experiment (B).}
\label{fig:FDP_ExtendedDAG_G}
\end{figure}
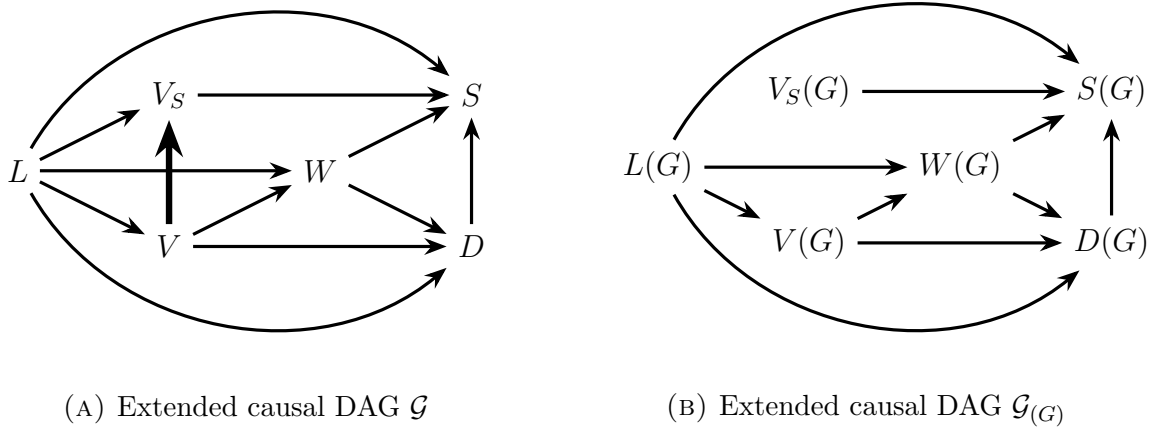

\begin{assumption}[Faithfulness \cite{pearl}]\label{ass:Faithfulness}
    The distribution $\mathbb{P}_{(G)}$ of the variables in the \textit{experiment} $(G)$ is faithful  to the DAG $\mathcal{G}_{(G)}$, meaning conditional independence in $\mathbb{P}_{(G)}$ implies $d$-separation in $\mathcal{G}_{(G)}$.
\end{assumption}

Despite Assumption \ref{ass:Faithfulness}, we do not require for our formalism the specification of a DAG $\mathcal{G}$ or $\mathcal{G}_{(G)}$, but rather the existence of one to which our assumptions can be related. This way, we provide practitioners with the freedom to specify a causal DAG in their specific application.

While faithfulness is an untestable assumption, it has been deemed realistic by results stating the existence of faithful laws \cite[Sec. 3.6.5]{richardson2013single}, and the fact that in some some scenarios almost all laws are faithful \cite[Thm. 3.2]{spirtes2000causation}. For an extensive discussion on whether faithfulness typically holds, see \cite{boeken2024bayesian}. Assumption~\ref{ass:Faithfulness} together with condition \eqref{eq:DCC_D} could be dropped if we were willing to directly assume $V_S$ partial isolation, or the assumption of no effect $D^{v_S}=D$ almost surely. Yet, we chose to present the faithfulness assumption together with the dismissible component condition \eqref{eq:DCC_D} as at least one of these is, in principle, testable on real data. 

We will show that the former assumptions are sufficient to identify the counterfactual distribution of interest. However, before stating the identification theorem we will study the relationship between our identification conditions and partial isolation.

\begin{restatable}{lemma}{lemmaXPI}\label{lemma:XPI}
    Under Assumption \ref{ass:Faithfulness}, condition \eqref{eq:DCC_D} is equivalent to  the absence of directed paths from $V_S(G)$ to $D(G)$
\end{restatable}
\begin{proof}
    ``$\Leftarrow$'': We will show the $d$-separation $D(G)\perp V_S(G)\mid \{L(G),V(G)\}$ in $\mathcal{G}_{(G)}$, and conclude independence from the global Markov property under a FFRCISTG model \cite[Sec.\ 3.5.2]{richardson2013single}. Denote  $\mathcal{B}\defeq\{L(G),V(G)\}$. Consider a path in $\mathcal{G}_{(G)}$ from $V_S(G)$ to $D(G)$. If this path contains $L(G)$ then it is blocked by $\mathcal{B}$ as $L(G)$ is a source node in $\mathcal{G}_{(G)}$ and therefore it is never a collider. If the path contains $V(G)$ it is also blocked by $\mathcal{B}$, as $V(G)$ can only have $L(G)$ as parent, due to the topological order in the DAG, meaning it cannot be a collider.  The path is also blocked if it contains $S(G)$, as this is always a collider without descendants in $\mathcal{B}$ (it is a sink node). Assume then that this path only contains nodes in $\{V_S(G),W_S(G),W_D(G),D(G)\}$. Per the definition of $(G)$, $V_S(G)$ has no parents in the DAG, therefore this path starts with an edge $V_S(G)\to$. As there are no directed paths from $V_S(G)$ to $D(G)$ in $\mathcal{G}_{(G)}$, this path necessarily contains an intermediate node ($W_S$ or $W_D$) which is a collider without descendants in $\mathcal{B}$ due to the topological ordering of the DAG. This path is therefore closed. 

    ``$\Rightarrow$'': Under Assumption \ref{ass:Faithfulness}, condition \eqref{eq:DCC_D} implies the $d$-separation $D(G)\perp V_S(G)\mid \mathcal{B}$ in $\mathcal{G}_{(G)}$. If there is a directed path from $V_S(G)$ to $D(G)$, this path cannot contain any node in $\mathcal{B}$, which would violate $d$-separation, leading to a contradiction.
\end{proof}
\begin{restatable}{lemma}{lemmaDPI}\label{lemma:DPI}
    Under Assumption \ref{ass:Faithfulness}, conditions \eqref{eq:DCC_X} and \eqref{eq:DCC_WX} imply that all causal paths in $\mathcal{G}_{(G)}$ from $V(G)$ to $S(G)$ which do not intersect $D(G)$ are mediated by $W_D(G)$. 
\end{restatable}
\begin{proof}
    We want to study directed paths from $V(G)$ to $S(G)$ which do not go through $D(G)$. All these possible paths are:
    \begin{enumerate}
        \item[(a)] $V(G)\to S(G)$
        \item[(b)] $V(G)\to W_D(G)\to S(G)$
        \item[(c)] $V(G)\to W_D(G)\to W_S(G)\to S(G)$
        \item[(d)] $V(G)\to W_S(G)\to W_D(G)\to S(G)$
        \item[(e)] $V(G)\to  W_S(G)\to S(G)$
    \end{enumerate}
    Under Assumption \ref{ass:Faithfulness}, path (a) is forbidden by condition \eqref{eq:DCC_X}, and paths (c-e) are forbidden by condition \eqref{eq:DCC_WX}. However, path (b) violates none of the $d$-separation statements implied by the  dismissible component conditions under faithfulness.
\end{proof}

\begin{restatable}{proposition}{propXPI}\label{prop:XPI}
    Under a FFRCISTG model,  Assumptions \ref{ass:DCC} and \ref{ass:Faithfulness} imply $V_S$ partial isolation (Definition~\ref{def:XPI}).
\end{restatable}
\begin{proof}
    Follows directly from Lemma \ref{lemma:XPI} and Definition \ref{def:XPI}.
\end{proof}

Recall that under $V_S$ partial isolation all the causal paths from $V_S$ to $S$ in the extended causal DAG  are not mediated by $D$. It remains to study the causal paths stemming from $V$, and to determine whether they are all mediated by $D$. As we see in the next Proposition, this need not be the case in general. 
\begin{restatable}{proposition}{propFullIsolation}\label{prop:FullIsolation}
    Under a FFRCISTG model and Assumption \ref{ass:Faithfulness},  if Assumption \ref{ass:DCC} holds under a partition such that $W_D(G)=\emptyset$, then full isolation holds.
\end{restatable}
\begin{proof}
    Follows directly from Lemmas  \ref{lemma:XPI} and  \ref{lemma:DPI},  and the definition of full isolation.
\end{proof}

This means that if, for example, $W_D=\emptyset$, the identification conditions imply full isolation, and changes in the distribution of $S^{v_S}\mid \{V=v,D=d\}$ w.r.t.\ $v_S$ capture all effects of $V_S$, none of which are mediated by $D$. If $W_D\neq\emptyset$, the former assumptions do not forbid the existence of paths $V\to W_D\to S$ in the extended causal DAG (see Lemma \ref{lemma:DPI}). 

Now that we have stated the required assumption and discussed their relation to partial isolation, we can present the main identification result. 

\begin{restatable}{theorem}{thmIdentification}\label{thm:Identification}
    If Assumptions \ref{ass:Positivity}-\ref{ass:Faithfulness} hold under a FFRCISTG model,  the counterfactual distribution of interest is identified as 
    \begin{align}\label{eq:Identification}
        \mathbb{P}(S^{v_S}\leq s\mid V=v,D=d)=&\sum_{l,w}\mathbb{P}(S\leq s\mid V=v_S,D=d,L=l,W=w)\\&\times \mathbb{P}(W_D=w_D\mid V=v,D=d,W_S=w_S,L=l)\nonumber\\&\times \mathbb{P}(W_S=w_S\mid V=v_S,D=d,L=l)\nonumber\\&\times \mathbb{P}(L=l\mid V=v,D=d),\nonumber
    \end{align}
    for $s$ in the support of $S$.
\end{restatable}
\begin{proof}
    Identification goes as follows
    \begin{align*}
        &\mathbb{P}(S^{v_S}\leq s\mid V=v,D=d)\\\stackrel{\rm (I)}{=}&\mathbb{P}(S^{v_S}(G)\leq s\mid V(G)=v,D(G)=d)\\
        \stackrel{\rm (II)}{=}&\sum_{l,w}\mathbb{P}(S^{v_S}(G)\leq s\mid V(G)=v,D(G)=d,L(G)=l,W^{v_S}(G)=w)\\&\times \mathbb{P}(W_D^{v_S}(G)=w_D\mid V(G)=v,D(G)=d,W_S^{v_S}(G)=w_S,L(G)=l)\\&\times \mathbb{P}(W_S^{v_S}(G)=w_S\mid V(G)=v,D(G)=d,L(G)=l)\\&\times \mathbb{P}(L(G)=l\mid V(G)=v,D(G)=d)\\
        \stackrel{\rm (III)}{=}&\sum_{l,w}\mathbb{P}(S^{v_S}(G)\leq s\mid V_S(G)=v_S, V(G)=v,D(G)=d,L(G)=l,W^{v_S}(G)=w)\\&\times \mathbb{P}(W_D^{v_S}(G)=w_D\mid V_S(G)=v_S, V(G)=v,D(G)=d,W_S^{v_S}(G)=w_S,L(G)=l)\\&\times \mathbb{P}(W_S^{v_S}(G)=w_S\mid V_S(G)=v_S, V(G)=v,D(G)=d,L(G)=l)\\&\times \mathbb{P}(L(G)=l\mid V(G)=v,D(G)=d)\\
        \stackrel{\rm (IV)}{=}&\sum_{l,w}\mathbb{P}(S(G)\leq s\mid V_S(G)=v_S, V(G)=v,D(G)=d,L(G)=l,W(G)=w)\\&\times \mathbb{P}(W_D(G)=w_D\mid V_S(G)=v_S, V(G)=v,D(G)=d,W_S(G)=w_S,L(G)=l)\\&\times \mathbb{P}(W_S(G)=w_S\mid V_S(G)=v_S, V(G)=v,D(G)=d,L(G)=l)\\&\times \mathbb{P}(L(G)=l\mid V(G)=v,D(G)=d)\\
        \stackrel{\rm (V)}{=}&\sum_{l,w}\mathbb{P}(S(G)\leq s\mid V_S(G)=v_S, V(G)=v_S,D(G)=d,L(G)=l,W(G)=w)\\&\times \mathbb{P}(W_D(G)=w_D\mid V_S(G)=v, V(G)=v,D(G)=d,W_S(G)=w_S,L(G)=l)\\&\times \mathbb{P}(W_S(G)=w_S\mid V_S(G)=v_S, V(G)=v_S,D(G)=d,L(G)=l)\\&\times \mathbb{P}(L(G)=l\mid V(G)=v,D(G)=d)\\
        \stackrel{\rm (VI)}{=}&\sum_{l,w}\mathbb{P}(S\leq s\mid V=v_S,D=d,L=l,W=w)\\&\times \mathbb{P}(W_D=w_D\mid V=v,D=d,W_S=w_S,L=l)\\&\times \mathbb{P}(W_S=w_S\mid V=v_S,D=d,L=l)\\&\times \mathbb{P}(L=l\mid V=v,D=d),
    \end{align*}
    where each equality follows from the following
    \begin{itemize}
        \item[(I)] Follows from the definition of $(G)$ and from Lemmas \ref{lemma:XPI} and \ref{lemma:DPI}. As in $(G)$ it holds that $V(G)$ is not a descendant of $V_S(G)$ (by construction), nor is $D(G)$ (by Lemma \ref{lemma:XPI} under the identification assumptions), neither of these variables gets labelled in the SWIG derived from $\mathcal{G}_{(G)}$ under an intervention $V_S(G)=v_S$ \cite{richardson2013single}.
        \item[(II)] Law of total probability.
        \item[(III)] Follows again from the definition of $(G)$. As $V_S$ is randomized in $(G)$, $V_S(G)$ has no parents in $\mathcal{G}_{(G)}$. Thus, in the SWIG derived from $\mathcal{G}_{(G)}$ considered under an intervention $V_S(G)=v_S$, the node-splitting procedure leaves the random node $V_S(G)$ without parents or descendants, which then renders it independent of any other variable represented by a random node in this SWIG due to modularity \cite{richardson2013single}.
        \item[(IV)] Follows from consistency under interventions on $V_S$.
        \item[(V)] Follows from positivity and the first three dismissible component conditions (Assumption~\ref{ass:DCC}).
        \item[(VI)] Follows from the definition of $(G)$, positivity, and the determinisms that holds in the observed data relating to the two components $V$ and $V_S$.
    \end{itemize}
\end{proof}

We refer to the right-hand-side of Equation \eqref{eq:Identification} as the g-formula for $\mathbb{P}(S^{v_S}\leq s\mid V=v,D=d)$ \cite{robins1986new}, which is expressed in terms of factual quantities. Thus, Theorem \ref{thm:Identification} provides sufficient conditions to identify the targetted quantities for both our notions of demographic insensitivity (Definitions \ref{def:FDP}-\ref{def:FDP_weak}).

\subsection{Estimation when $V$ is discrete}\label{sec:FDP_Estimation}
As we have shown that under appropriate assumptions the observed data provides us with enough information to identify the counterfactual distribution of interest, we will now propose three estimators for it. For a fixed $s$ in the support of $S$ denote by $\nu(\mathbb{P},s)$ the right-hand-side of Equation \eqref{eq:Identification} in Theorem \ref{thm:Identification}. We begin by defining a simple, plug-in estimator. Let $\Tilde{\mathbb{P}}$ be posited (possibly parametric) models for the conditional distributions of $S,W_D,W_S,L$ as they appear in Theorem \ref{thm:Identification}, and $\Tilde{\mathbb{P}}_n$ be these models fitted on the available data. We can then construct the simple estimator as $\widehat{\nu}_{\rm simple}(s)=\nu(\Tilde{\mathbb{P}}_n,s)$. If the models $\Tilde{\mathbb{P}}$ are correctly specified and consistently estimated, then $\widehat{\nu}_{\rm simple}(s)$ is consistent for $\mathbb{P}(S^{v_S}\leq s\mid V=v,D=d)$ under the identification conditions of Theorem \ref{thm:Identification}.

The main burden when computing the simple estimator is the need to fit outcome models for all the variables in the system except those in $V$. This is somehow alleviated by weighted estimators. To that extent, consider the next theorem. We will assume for the remainder of this sub-section that $V$ has a discrete support. We provide alternative estimators which apply otherwise in Appendix \ref{sec:AppIPW_Continuous}. 

\begin{restatable}{theorem}{thmIPWS}\label{thm:IPW_S}
    Under the conditions of Theorem \ref{thm:Identification} an equivalent identification formula is
    \begin{align}\label{eq:IPW_S}
        \mathbb{P}(S^{v_S}\leq s\mid V=v,D=d)=\mathbb{E}[I(S\leq s)H_{S}H_{W_S}\mid V=v,D=d],
    \end{align}
    where we have defined the weights
    \begin{align*}
        H_S&=\dfrac{\mathbb{P}(S\leq s\mid V=v_S,D=d,L,W)}{\mathbb{P}(S\leq s\mid V=v,D=d,L,W)},\\
        H_{W_S}&=\dfrac{\mathbb{P}(W_S=W_S\mid V=v_S,D=d,L)}{\mathbb{P}(W_S=W_S\mid V=v,D=d,L)}.
    \end{align*}
\end{restatable}
\begin{proof}
    We see that
    \begin{align*}
        &\mathbb{E}[I(S\leq s)H_{S}H_{W_S}\mid V=v,D=d]\\=&\mathbb{E}[\mathbb{E}[I(S\leq s)\mid L,W,V=v,D=d]H_{S}H_{W_S}\mid V=v,D=d]\\
        =&\sum_{l,w}\mathbb{P}(S\leq s\mid L=l,W=w,V=v,D=d)H_S\\
        &\qquad \times\mathbb{P}(W_D=w_D\mid L=l,W_S=w_S,V=v,D=d)\\
        &\qquad \times\mathbb{P}(W_S=w_S\mid L=l,V=v,D=d)H_{W_S}\\
        &\qquad \times\mathbb{P}(L=l\mid V=v,D=d)\\
        =&\sum_{l,w}\mathbb{P}(S\leq s\mid L=l,W=w,V=v_S,D=d)\\
        &\qquad \times\mathbb{P}(W_D=w_D\mid L=l,W_S=w_S,V=v,D=d)\\
        &\qquad \times\mathbb{P}(W_S=w_S\mid L=l,V=v_S,D=d)\\
        &\qquad \times\mathbb{P}(L=l\mid V=v,D=d),
    \end{align*}
    which equals the right-hand-side of Equation \eqref{eq:Identification}, which under the conditions of Theorem \ref{thm:Identification} equals $\mathbb{P}(S^{v_S}\leq s\mid V=v,D=d)$. This concludes the proof.
\end{proof}

Equation \eqref{eq:IPW_S} motivates the weighted estimator based on posited models for the quantities in the weights $H_{S},H_{W_S}$. For fixed $s$ in the support of $S$ consider
$$\widehat{\nu}_{{\rm weighted},S}(s)=\widehat{\mathbb{E}}_n\left[\dfrac{I(S\leq s,V=v,D=d)}{\widehat{\mathbb{P}}_n(V=v,D=d)}\Tilde{H}_S\Tilde{H}_{W_S}\right],$$
where $\widehat{\mathbb{P}}_n(V=v,D=d)=\widehat{\mathbb{E}}_n[I(V=v,D=d)]$ and $\widehat{\mathbb{E}}_n$ denotes an empirical mean. Under the conditions of Theorem \ref{thm:Identification}, $\widehat{\nu}_{{\rm weighted},S}(s)$ is a consistent estimator for the c.d.f.\ of the counterfactual distribution of interest as long as the models for  the conditional distributions of $S$ and $W_S$ are correctly specified and consistently estimated.  Alternatively, we propose a weighted estimator which requires specification of outcome models for $L$ and $W_D$ instead. It is based on the next Theorem.

\begin{restatable}{theorem}{thmIPWD}\label{thm:IPW_D}
    Under the conditions of Theorem \ref{thm:Identification} an equivalent identification formula is
    \begin{align}\label{eq:IPW_D}
        \mathbb{P}(S^{v_S}\leq s\mid V=v,D=d)&=\mathbb{E}[I(S\leq s)H_LH_{W_D}\mid V=v_S,D=d],
    \end{align}
    where we have defined the weights
    \begin{align*}
        H_{W_D}&=\dfrac{\mathbb{P}(W_D=W_D\mid V=v,D=d,L,W_S)}{\mathbb{P}(W_D=W_D\mid V=v_S,D=d,L,W_S)},\\
        H_{L}&=\dfrac{\mathbb{P}(L=L\mid V=v,D=d)}{\mathbb{P}(L=L\mid V=v_S,D=d)}.
    \end{align*}
\end{restatable}
\begin{proof}
    Follows analogously to the proof of Theorem \ref{thm:IPW_S}.
\end{proof}
Analogously to the weighted estimator formerly presented, Equation \eqref{eq:IPW_D} motivates the weighted estimator based on posited models for the quantities in the weights $H_{L},H_{W_D}$. For fixed $s$ in the support of $S$ consider
$$\widehat{\nu}_{{\rm weighted},D}(s)=\widehat{\mathbb{E}}_n\left[\dfrac{I(S\leq s,V=v_S,D=d)}{\widehat{\mathbb{P}}_n(V=v_S,D=d)}\Tilde{H}_L\Tilde{H}_{W_D}\right].$$
Under the conditions of Theorem \ref{thm:Identification}, $\widehat{\nu}_{{\rm weighted},D}(s)$ is a consistent estimator for the c.d.f.\ of the counterfactual distribution of interest as long as the models for the  conditional distributions of $L$ and $W_D$ are correctly specified and consistently estimated.

Lastly, we introduce a one-step estimator, defined as $$\widehat{\nu}_{\rm OS}(s)=\widehat{\mathbb{E}}_n[\nu^1(\Tilde{\mathbb{P}}_n,s)]+\nu(\Tilde{\mathbb{P}}_n,s),$$ where $\nu^1(\cdot,s)$ is $\nu(\cdot,s)$'s influence function.

\begin{definition}[Influence function of an estimand, adapted from \cite{tsiatis2007semiparametric}] \label{def:IF}
    We define the influence function $\chi^1$ of a (pathwise differentiable) estimand $\chi$ as a random variable with mean zero and finite variance such that for every (regular) parametric submodel $\{\mathbb{P}_t\colon t\in[0,1)\}$ it satisfies $$\frac{d\chi(\mathbb{P}_t)}{dt}\bigg |_{t=0}=\mathbb{E}[\chi^1g],$$ where $g$ is the score of the true law $\mathbb{P}_0$.
\end{definition}

\begin{restatable}{theorem}{thmEIF}\label{thm:EIF}
    The influence function of $\nu(\mathbb{P},s)$ is 
    \begin{align}\label{eq:EIF}
        \nu^1(\mathbb{P},s) &=\dfrac{I(\varphi(v))}{\mathbb{E}[I(\varphi(v))]}\{\mathbb{E}[\mathbb{E}[\mathbb{E}[I(S\leq s)\mid L,W,\varphi(v_S)]\mid L,W_S,\varphi(v)]\mid L,\varphi(v_S)]-\nu\}\\
        &+\dfrac{I(\varphi(v_S))}{\mathbb{E}[I(\varphi(v_S))]}H_L\left\{\mathbb{E}[\mathbb{E}[I(S\leq s)\mid L,W,\varphi(v_S)]\mid L,W_S,\varphi(v)]\right.\nonumber\\
        &\qquad\qquad-\left.\mathbb{E}[\mathbb{E}[\mathbb{E}[I(S\leq s)\mid L,W,\varphi(v_S)]\mid L,W_S,\varphi(v)]\mid L,\varphi(v_S)]\right\}\nonumber\\
        &+\dfrac{I(\varphi(v))}{\mathbb{E}[I(\varphi(v))]}H_{W_S}\{\mathbb{E}[I(S\leq s)\mid L,W,\varphi(v_S)]\nonumber\\
        &\qquad\qquad-\mathbb{E}[\mathbb{E}[I(S\leq s)\mid L,W,\varphi(v_S)]\mid L,W_S,\varphi(v)]\}\nonumber\\
        &+\dfrac{I(\varphi(v_S))}{\mathbb{E}[I(\varphi(v_S))]}H_LH_{W_D}\{I(S\leq s)-\mathbb{E}[I(S\leq s)\mid L,W,\varphi(v_S)]\},\nonumber
    \end{align}
    where we have defined $\varphi(v)=\{V=v,D=d\}$ and $H_L,H_{W_S},H_{W_D}$ are taken as in Theorems \ref{thm:IPW_D} and \ref{thm:IPW_S}.
\end{restatable}
\begin{proof}
    For a fixed $s$ in the support of $S$ denote $Y=I(S\leq s)$, and we drop the second argument in $\nu(\mathbb{P},s)$ as it will remain fixed. We begin by noticing that we can write this quantity as an iterated conditional expectation (ICE)
    $$\nu(\mathbb{P})=\mathbb{E}[\mathbb{E}[\mathbb{E}[\mathbb{E}[Y \mid L,W, \varphi(v_S)]\mid L,W_S, \varphi(v)]\mid L, \varphi(v_S)]\mid \varphi(v)].$$
    Taking Gateaux derivatives within a parametric submodel and applying the chain rule, we get
    \begin{subequations}\label{eq:GateauxDerivative}
        \begin{align}
            \dfrac{d\nu(\mathbb{P}_t)}{dt}\bigg|_{t=0}&=\nonumber\\
            &=\mathbb{E}[\mathbb{E}[\mathbb{E}[\mathbb{E}[Y g_{Y\mid L,W, \varphi(v_S)} \mid L,W, \varphi(v_S)]\mid L,W_S, \varphi(v)]\mid L, \varphi(v_S)]\mid \varphi(v)]\label{eq:GateauxDerivative_1}\\
            &+\mathbb{E}[\mathbb{E}[\mathbb{E}[\mathbb{E}[Y \mid L,W, \varphi(v_S)]g_{W_D\mid L,W_S, \varphi(v)}\mid L,W_S, \varphi(v)]\mid L, \varphi(v_S)]\mid \varphi(v)]\label{eq:GateauxDerivative_2}\\
            &+\mathbb{E}[\mathbb{E}[\mathbb{E}[\mathbb{E}[Y \mid L,W, \varphi(v_S)]\mid L,W_S, \varphi(v)]g_{W_S\mid L, \varphi(v_S)}\mid L, \varphi(v_S)]\mid \varphi(v)]\label{eq:GateauxDerivative_3}\\
            &+\mathbb{E}[\mathbb{E}[\mathbb{E}[\mathbb{E}[Y \mid L,W, \varphi(v_S)]\mid L,W_S, \varphi(v)]\mid L, \varphi(v_S)]g_{L\mid  \varphi(v)}\mid \varphi(v)],\label{eq:GateauxDerivative_4}
    \end{align}
    \end{subequations}
    where by $g_A$ we denote the score function for $A$. We now apply the same procedure to each of the four addends in Equation \eqref{eq:GateauxDerivative}: centring, expressing conditioning as inverse probability weighting, and completing the score. 
    \begin{align*}
        \eqref{eq:GateauxDerivative_1}&=\mathbb{E}[\mathbb{E}[\mathbb{E}[\mathbb{E}[(Y-[\mathbb{E}[Y  \mid L,W, \varphi(v_S)]) g_{Y\mid L,W, \varphi(v_S)} \mid L,W, \varphi(v_S)]\mid L,W_S,\\&\qquad\qquad \varphi(v)]\mid L, \varphi(v_S)]\mid \varphi(v)]\\
        &=\mathbb{E}\left[\frac{I(\varphi(v))}{\mathbb{E}[I(\varphi(v))]}\mathbb{E}\left[\frac{I(\varphi(v_S))}{\mathbb{E}[I(\varphi(v_S))\mid L]}\mathbb{E}\left[\frac{I(\varphi(v))}{\mathbb{E}[I(\varphi(v))\mid L,W_S]}\times\right.\right.\right.\\
        &\qquad\qquad\mathbb{E}\left[\frac{I(\varphi(v_S))}{\mathbb{E}[I(\varphi(v_S))\mid L,W]}(Y-[\mathbb{E}[Y  \mid L,W, \varphi(v_S)])\times\right.\\
        &\qquad\qquad \left.\left.\left.\left.g_{Y\mid L,W,V,D} \mid L,W\right]\mid L,W_S\right]\mid L\right]\right]\\
        &=\mathbb{E}\left[\frac{\mathbb{E}[I(\varphi(v))\mid L]}{\mathbb{E}[I(\varphi(v))]}\frac{\mathbb{E}[I(\varphi(v_S))\mid L,W_S]}{\mathbb{E}[I(\varphi(v_S))\mid L]} \frac{\mathbb{E}[I(\varphi(v))\mid L,W]}{\mathbb{E}[I(\varphi(v))\mid L,W_S]}\times\right.\\
        &\left.\qquad\qquad\frac{I(\varphi(v_S))}{\mathbb{E}[I(\varphi(v_S))\mid L,W_S]}(Y-[\mathbb{E}[Y  \mid L,W, \varphi(v_S)])g_{Y\mid L,W,V,D}    \right]\\
        &=\mathbb{E}\left[\frac{I(\varphi(v_S))}{\mathbb{E}[I(\varphi(v_S))]}H_{W_D}H_L(Y-[\mathbb{E}[Y  \mid L,W, \varphi(v_S)])(g_{Y\mid L,W,V,D}+g_{L,W,D,V})\right]\\
        &=\mathbb{E}\left[\frac{I(\varphi(v_S))}{\mathbb{E}[I(\varphi(v_S))]}H_{W_D}H_L(Y-[\mathbb{E}[Y  \mid L,W, \varphi(v_S)])g\right],
    \end{align*}
    where the first equality follows from the fact that the score has mean zero what allows us to centre $Y$; the second equality follows from the discrete IPW identity $$\mathbb{E}[A\mid B,C=c]=\mathbb{E}\left[\frac{I(C=c)}{\mathbb{P}(C=c\mid B)}A\mid B\right]$$ applied to all iterated conditional expectations; the third equality follows from the identity $\mathbb{E}[A\mathbb{E}[B\mid C]]=\mathbb{E}[B\mathbb{E}[A\mid C]]$; the fourth equality follows by handling the inner IPW factor using Bayes' theorem to rewrite it as the weights in Theorems \ref{thm:IPW_D} and \ref{thm:IPW_S} together with adding the marginal score, as the factor multiplying the score in this expression has mean zero; and the last equality follows from additivity of the scores.

    We see  how the factor multiplying the score $g$ in the last equality is the fourth addend in the right-hand-side of Equation \eqref{eq:EIF}. Proceeding analogously with \eqref{eq:GateauxDerivative_2}-\eqref{eq:GateauxDerivative_4} we retrieve the remaining terms for $\nu$'s influence function. This concludes the proof.
\end{proof}

The one-step estimator corrects the simple one by the empirical mean of the estimand's influence function. By doing so, it achieves double-robustness as we see in the following theorem.

\begin{restatable}{theorem}{thmDR}\label{thm:DR}
    The one-step estimator $\widehat{\mathbb{E}}_n[\nu^1(\Tilde{\mathbb{P}}_n,s)]+\nu(\Tilde{\mathbb{P}}_n,s)$ is a doubly robust estimator of the  target quantity, meaning that under the identification conditions of Theorem \ref{thm:Identification}  it is consistent for the  probability of  interest if either the models for ${L,W_D}$ or  $S,W_S$ are correctly specified and consistently estimated, but not necessarily both.
\end{restatable}
\begin{proof}
    As the postulated models are consistently estimated, we can work directly with $\Tilde{\mathbb{P}}$. Then the one-step estimator is the solution to the estimating equations $\widehat{\mathbb{E}}_n[\nu^1(\Tilde{\mathbb{P}},s)+\nu(\Tilde{\mathbb{P}},s)-\widehat{\nu}_{\rm OS}(s)]=0$. Thus, as $\nu(\Tilde{\mathbb{P}},s)$ is also a deterministic quantity, it suffices to show that $\mathbb{E}[\nu^1(\Tilde{\mathbb{P}},s)]=\nu(s)-\nu(\Tilde{\mathbb{P}},s)$. We show this in the case when the postulated models for $\mathbb{P}(V=v,D=d)$ and $L,W_D$, but not necessarily those for $S,W_S$. 

    Let us denote again $Y=I(S\leq s)$. For the first term of the influence function in Equation \eqref{eq:EIF} we have
    \begin{align*}
        &\mathbb{E}\left[\dfrac{I(\varphi(v))}{\mathbb{E}[I(\varphi(v))]}\Tilde{\mathbb{E}}[\mathbb{E}[\Tilde{\mathbb{E}}[Y\mid L,W,\varphi(v_S)]\mid L,W_S,\varphi(v)]\mid L,\varphi(v_S)]\right]-\nu(\Tilde{\mathbb{P}},s)\\
        =&\mathbb{E}\left[\Tilde{\mathbb{E}}[\mathbb{E}[\Tilde{\mathbb{E}}[Y\mid L,W,\varphi(v_S)]\mid L,W_S,\varphi(v)]\mid L,\varphi(v_S)]\mid \varphi(v)\right]-\nu(\Tilde{\mathbb{P}},s)\\
        =&\nu(\Tilde{\mathbb{P}},s)-\nu(\Tilde{\mathbb{P}},s)\\=&0.
    \end{align*}
    Proceeding similarly for the third term of the influence function in Equation \eqref{eq:EIF}
    \begin{align*}
        &\mathbb{E}\left[\dfrac{I(\varphi(v))}{\mathbb{E}[I(\varphi(v))]}\Tilde{H}_{W_S}\left\{\Tilde{\mathbb{E}}[Y\mid L,W,\varphi(v_S)]-\mathbb{E}[\Tilde{\mathbb{E}}[Y\mid L,W,\varphi(v_S)]\mid L,W_S,\varphi(v)]\right\}\right]\\
        =&\mathbb{E}\left[\Tilde{H}_{W_S}\mathbb{E}\left[\Tilde{\mathbb{E}}[Y\mid L,W,\varphi(v_S)]\right.\right.\\&\left.\left.\qquad\qquad-\mathbb{E}[\Tilde{\mathbb{E}}[Y\mid L,W,\varphi(v_S)]\mid L,W_S,\varphi(v)]\mid L,W_S,\varphi(v)\right]\mid \varphi(v)\right]\\=&0,
    \end{align*}
    because the second most outward expectation is almost surely zero.
    The expectation of the other two terms will not be zero. For the second addend in Equation~\eqref{eq:EIF}
    \begin{align*}
        &\mathbb{E}\left[ 
        \dfrac{I(\varphi(v_S))}{\mathbb{E}[I(\varphi(v_S))]}H_L\mathbb{E}[\Tilde{\mathbb{E}}[Y\mid L,W,\varphi(v_S)]\mid L,W_S,\varphi(v)]\right]\\
        &-\mathbb{E}\left[ 
        \dfrac{I(\varphi(v_S))}{\mathbb{E}[I(\varphi(v_S))]}H_L\Tilde{\mathbb{E}}[\mathbb{E}[\Tilde{\mathbb{E}}[Y\mid L,W,\varphi(v_S)]\mid L,W_S,\varphi(v)]\mid L,\varphi(v_S)]\right]\\
        =&\mathbb{E}\left[H_L\mathbb{E}[\Tilde{\mathbb{E}}[Y\mid L,W,\varphi(v_S)]\mid L,W_S,\varphi(v)]\mid \varphi(v_S)\right]\\
        &-\mathbb{E}\left[ H_L\Tilde{\mathbb{E}}[\mathbb{E}[\Tilde{\mathbb{E}}[Y\mid L,W,\varphi(v_S)]\mid L,W_S,\varphi(v)]\mid L,\varphi(v_S)]\mid\varphi(v_S)\right]\\
        =&\mathbb{E}\left[H_L\mathbb{E}\left[\mathbb{E}[\Tilde{\mathbb{E}}[Y\mid L,W,\varphi(v_S)]\mid L,W_S,\varphi(v)]\mid L,\varphi(v_S)\right]\mid \varphi(v_S)\right]\\
        &-\mathbb{E}\left[\Tilde{\mathbb{E}}[\mathbb{E}[\Tilde{\mathbb{E}}[Y\mid L,W,\varphi(v_S)]\mid L,W_S,\varphi(v)]\mid L,\varphi(v_S)]\mid\varphi(v)\right]\\
        =&\mathbb{E}\left[\mathbb{E}\left[\mathbb{E}[\Tilde{\mathbb{E}}[Y\mid L,W,\varphi(v_S)]\mid L,W_S,\varphi(v)]\mid L,\varphi(v_S)\right]\mid \varphi(v)\right]-\nu(\Tilde{\mathbb{P}},s).
    \end{align*}
    For the remaining term of the influence function we will get the required complement to this last difference. 
    \begin{align*}
        &\mathbb{E}\left[\dfrac{I(\varphi(v_S))}{\mathbb{E}[I(\varphi(v_S))]}H_LH_{W_D}\{Y-\Tilde{\mathbb{E}}[Y\mid L,W,\varphi(v_S)]\}\right]\\
        =&\underbrace{\mathbb{E}\left[H_LH_{W_D}Y\mid\varphi(v_S)\right]}_{\rm (I)}-\underbrace{\mathbb{E}\left[H_LH_{W_D}\Tilde{\mathbb{E}}[Y\mid L,W,\varphi(v_S)]\mid\varphi(v_S)\right]}_{\rm (II)}.
    \end{align*}
    For the first term we have
    \begin{align*}
        {\rm (I)}&=\mathbb{E}\left[H_LH_{W_D}\mathbb{E}[Y\mid L,W,\varphi(v_S)]\mid\varphi(v_S)\right]\\
        &=\mathbb{E}\left[H_L\mathbb{E}[H_{W_D}\mathbb{E}[Y\mid L,W,\varphi(v_S)]\mid L,W_S,\varphi(v_S) ]\mid\varphi(v_S)\right]\\
        &=\mathbb{E}\left[H_L\mathbb{E}[\mathbb{E}[Y\mid L,W,\varphi(v_S)]\mid L,W_S,\varphi(v) ]\mid\varphi(v_S)\right]\\
        &=\mathbb{E}\left[H_L\mathbb{E}[\mathbb{E}[\mathbb{E}[Y\mid L,W,\varphi(v_S)]\mid L,W_S,\varphi(v) ]\mid L,\varphi(v_S)]\mid\varphi(v_S)\right]\\
        &=\mathbb{E}\left[\mathbb{E}[\mathbb{E}[\mathbb{E}[Y\mid L,W,\varphi(v_S)]\mid L,W_S,\varphi(v) ]\mid L,\varphi(v_S)]\mid\varphi(v)\right]\\
        &=\nu(s),
    \end{align*}
    and for the second one, proceeding analogously
    \begin{align*}
        {\rm (II)}&= \mathbb{E}\left[\mathbb{E}\left[\mathbb{E}[\Tilde{\mathbb{E}}[Y\mid L,W,\varphi(v_S)]\mid L,W_S,\varphi(v)]\mid L,\varphi(v_S)\right]\mid \varphi(v)\right].
    \end{align*}
    Therefore we see that when we sum up the expectations of the four addends in $\mathbb{E}[\nu^1(\Tilde{\mathbb{P}},s)]$ we get precisely $\nu(s)-\nu(\Tilde{\mathbb{P}},s)$, as we wanted to show. The proof when  the postulated models for $\mathbb{P}(V=v,D=d)$ and $S,W_S$ are correctly specified, but not necessarily those for $L,W_D$, follows analogously. This concludes the proof.
\end{proof}

\subsection{Weighted estimators for an arbitrary support of $V$}\label{sec:AppIPW_Continuous}
The weighted estimators we have introduced in Appendix \ref{sec:FDP_Estimation}   only apply if $V$ has a discrete support, as otherwise the IPW identity is not valid. Nevertheless, the identification Theorems \ref{thm:Identification}, \ref{thm:IPW_S} and \ref{thm:IPW_D} are still valid if some components of $V$ have continuous (marginal) support. We will show that under some assumptions we can rewrite the right-hand-side of Equations \eqref{eq:IPW_S} and \eqref{eq:IPW_D} in an equivalent way which will allow for implementable weighted estimators. 

When  $V$ is univariate and  has an absolutely continuous distribution w.r.t.\ the Lebesgue measure, a counterfactual quantity identified by an expression such as Equations \eqref{eq:IPW_S} and \eqref{eq:IPW_D} is usually referred to as a \textit{dose-response function} \cite{imai2004causal}, as dosage is usually understood as a continuous treatment for which we only get to observe a finite number of values: the doses which are compared. 

Some of the dose-response and continuous treatment effects literature proposes weighted estimators which rely on kernel density estimators \cite{colangelo2025double,huling2024independence}. Despite the flexibility of these methods and the well-understood inference on them, they inherently make a restriction on the class of estimators which can be implemented. To remain as general as possible and provide practitioners with freedom-of-choice regarding their estimation algorithms, we follow the approach introduced in \cite{galvao2015uniformly}. To do so, we assume that $V$ can be expressed as having two (potentially multivariate) components $V=(V_c,V_d)$ such that $V_d$ has a discrete support and $V_c$ has an absolutely continuous distribution w.r.t.\ the Lebesgue measure. We then consider the following regularity conditions.
\begin{assumption}\label{ass:RegularityIPW}
    Let $Q$ denote a random vector and $$g\colon {\rm supp}(Q)\mapsto\mathbb{R}$$ a measurable function. Then assume that
    \begin{enumerate}
        \item $V_c\mid Q=q,V_d=v_d,D=d$ is supported on a connected open subset of $\mathbb{R}^{|V_c|}$ which does not depend on $(v_d,q,d)$; and its density w.r.t.\ the Lebesgue measure is continuous.
        \item There exists a measurable function  $t\colon {\rm supp}(Q)\mapsto\mathbb{R}^+$ such that $\int_{{\rm supp}(Q)} t(q)dq<\infty$ and $$|g(q)f_{(Q,V_c)| D,V_d}(q,v_c+\Delta v_c|d,v_d)|\leq t(q)$$ for all $v\in\mathcal{V}$, all $\Delta v_c\in (\mathbb{R}^+)^{|V_c|}$ with $||\Delta v_c||_2$ small enough, and all $d\in\{0,1\}$.
        \item It holds that 
        \begin{align*}
            &\mathbb{E}[g(Q)\mid V=v,D=d]\\=&\lim_{\substack{||\Delta v_c||_2\to 0^+\\\Delta v_c\in (\mathbb{R}^+)^{|V_c|}}}\mathbb{E}[g(Q)\mid V_d=v_d,D=d,V_c\in[v_c,v_c+\Delta v_c]],
        \end{align*}
        where we denote $[x,x+\delta]=\{z\colon z_i\in[x_i,x_i+\delta_i] \text{ for all }i\}.$
    \end{enumerate}
\end{assumption}
This assumption is adapted from  \cite[Ass.\ I.III]{galvao2015uniformly} and essentially allows to apply the dominated convergence theorem to exchange limits and expectations to obtain an IPW factor in the continuous case. Under this assumption, we will be able to rewrite the right-hand-side of Equations \eqref{eq:IPW_S} and \eqref{eq:IPW_D} as expectations conditional only on discrete random variables.

\begin{lemma}\label{lemma:IPW_Continuous}
    Under Assumption \ref{ass:RegularityIPW} it holds that $$\mathbb{E}[g(Q)\mid V=v,D=d]=\mathbb{E}[g(Q)\pi(Q,v,d)\mid V_d=v_d,D=d],$$ where 
    \begin{equation}\label{eq:DensityRatio}
        \pi(Q,v,d)=\dfrac{f_{V_c|Q,V_d,D}(v_c|Q,v_d,d)}{f_{V_c|V_d,D}(v_c|v_d,d)}
    \end{equation}
    is expressed as a density ratio.
\end{lemma}
\begin{proof}
    This proof follows closely the proof of \cite[Thm.\ 1]{galvao2015uniformly}, and it can be seen as a generalization to a multivariate case. We have that
    \begin{align*}
        &\mathbb{E}[g(Q)\mid V=v,D=d]\\
        \stackrel{(I)}{=}&\lim_{\substack{||\Delta v_c||_2\to 0^+\\\Delta v_c\in (\mathbb{R}^+)^{|V_c|}}}\mathbb{E}[g(Q)\mid V_d=v_d,D=d,V_c\in[v_c,v_c+\Delta v_c]]\\
        \stackrel{(II)}{=}&\lim_{\substack{||\Delta v_c||_2\to 0^+\\\Delta v_c\in (\mathbb{R}^+)^{|V_c|}}}\dfrac{\mathbb{E}[g(Q)I(V_c\in[v_c,v_c+\Delta v_c])\mid V_d=v_d,D=d]}{\mathbb{P}(V_c\in[v_c,v_c+\Delta v_c]\mid V_d=v_d,D=d )}\\
        \stackrel{(III)}{=}&\lim_{\substack{||\Delta v_c||_2\to 0^+\\\Delta v_c\in (\mathbb{R}^+)^{|V_c|}}}\dfrac{\int_{\text{supp}(Q)}\int_{[v_c,v_c+\Delta v_c]}g(q)f_{Q,V_c|V_d,D}(q,w|v_d,d)dwdq}{\mathbb{P}(V_c\in[v_c,v_c+\Delta v_c]\mid V_d=v_d,D=d )}\\
        \stackrel{(IV)}{=}&\lim_{\substack{||\Delta v_c||_2\to 0^+\\\Delta v_c\in (\mathbb{R}^+)^{|V_c|}}}\dfrac{\int_{\text{supp}(Q)}g(q)f_{Q,V_c|V_d,D}(q,v_c+\Tilde{\Delta} v_c|v_d,d)dq}{\mu([v_c,v_c+\Delta v_c])^{-1}\mathbb{P}(V_c\in[v_c,v_c+\Delta v_c]\mid V_d=v_d,D=d )}\\
        \stackrel{(V)}{=}&\dfrac{\int_{\text{supp}(Q)}g(q)f_{Q,V_c|V_d,D}(q,v_c|v_d,d)dq}{f_{V_c|V_d,D}(v_c|v_d,d)}\\
        \stackrel{}{=}&\mathbb{E}[g(Q)\pi(Q,v,d)\mid V_d=v_d,D=d],
    \end{align*}
    where $\Tilde{\Delta} v_c\in[0,{\Delta} v_c]$.  These equalities follow from the following arguments
    \begin{enumerate}
        \item[$(I)$] Follows from Assumption \ref{ass:RegularityIPW} \textit{(3)}.
        \item[$(II)$] The standard IPW expression for discrete events.
        \item[$(III)$] Explicitly writing out the expectation in the numerator.
        \item[$(IV)$] Multivariate integral mean-value theorem.
        \item[$(V)$] By Assumption \ref{ass:RegularityIPW} \textit{(2)} and the dominated convergence the integral in the numerator and the limit can be permuted, after which the results follows from continuity granted by Assumption \ref{ass:RegularityIPW} \textit{(1)}.
    \end{enumerate}
    This concludes the proof.
\end{proof}
We now apply this Lemma to the targetted conditional expectations in Theorems \ref{thm:IPW_S} and \ref{thm:IPW_D} to re-write them as conditional expectations with only discrete variables on the conditioning set.
\begin{theorem}[IPW for continuous demographic covariates]\label{thm:IPW_Continuous}
    Consider $Q=(S,L,W)$, and fix $s$ in the support of $S$. Define
    \begin{align*}
        g_1(Q)&=I(S\leq s)H_LH_{W_D},\\
        g_2(Q)&=I(S\leq s)H_SH_{W_S},
    \end{align*}
    where the $H$ functions are defined as in Theorems \ref{thm:IPW_S} and \ref{thm:IPW_D}. If Assumption \ref{ass:RegularityIPW} holds for $g_1$ then 
    \begin{align*}
        &\mathbb{E}[I(S\leq s)H_LH_{W_D}\mid V=v_S,D=d]\\=&\mathbb{E}[I(S\leq s)H_LH_{W_D}\pi(Q,v_S,d)\mid V_d=(v_S)_d,D=d],
    \end{align*}
    and if it holds for $g_2$ then 
    \begin{align*}
        &\mathbb{E}[I(S\leq s)H_SH_{W_S}\mid V=v,D=d]\\=&\mathbb{E}[I(S\leq s)H_SH_{W_S}\pi(Q,v,d)\mid V_d=v_d,D=d],
    \end{align*}
    where in both cases $\pi(Q,v,d)$ is taken as in Equation \eqref{eq:DensityRatio}.
\end{theorem}
\begin{proof}
    It is a straightforward application of Lemma \ref{lemma:IPW_Continuous} in both cases. 
\end{proof}
Theorem \ref{thm:IPW_Continuous} motivates two estimators of the counterfactual c.d.f.\ of interest. We first posit (parametric) models for the weights $H$ and for $\pi$, fit them to the observed data, and then compute estimators as
\begin{align*}
    \widehat{\nu}_{{\rm weighted},1}(s)&= \widehat{\mathbb{E}}_n\left[\dfrac{I(S\leq s,V_d=(v_S)_d,D=d)}{\widehat{\mathbb{P}}_n(V_d=(v_S)_d,D=d)}\Tilde{H}_L\Tilde{H}_{W_D}\Tilde{\pi}(Q,v_S,d)\right],\\
    \widehat{\nu}_{{\rm weighted},2}(s)&= \widehat{\mathbb{E}}_n\left[\dfrac{I(S\leq s,V_d=v_d,D=d)}{\widehat{\mathbb{P}}_n(V_d=v_d,D=d)}\Tilde{H}_S\Tilde{H}_{W_S}\Tilde{\pi}(Q,v,d)\right],
\end{align*}
where $Q=(S,L,W)$ as in Theorem \ref{thm:IPW_Continuous}. Fix $s$ in the support of $S$ and $i$ in $\{1,2\}$. Under the identification conditions of Theorem \ref{thm:Identification}, if Assumption \ref{ass:RegularityIPW} holds for $g=g_i$ with $g_i$ as defined in Theorem \ref{thm:IPW_Continuous}, and the posited models are correctly specified and consistently estimated, then $\widehat{\nu}_{{\rm weighted},i}(s)$ is a consistent estimator of the counterfactual c.d.f. of interest.

\subsection{Uniform consistency of the proposed estimators} \label{sec:AppFDP_UniformConsistency}

The estimators we have introduced in Appendices \ref{sec:FDP_Estimation} and \ref{sec:AppIPW_Continuous}, are point-wise consistent, meaning that for each $s$ in the support of $S$ we can construct consistent estimators of $\mathbb{P}(S^{v_S}\leq s\mid V=v,D=d)$, provided the required assumptions hold. However, the strong notion of demographic insensitivity (Definition \ref{def:FDP}) refers to a distribution as a whole, not to a specific probability. Therefore, when trying to estimate $\mathbb{P}(S^{v_S}\leq \cdot \mid V=v,D=d)$ for a fixed $(v,v_S,d)$, we would like to achieve uniform (in $L_\infty$) consistency $$\sup_{s\in{\rm supp}(S)}|\widehat{\nu}(s)-\mathbb{P}(S^{v_S}\leq s\mid V=v,D=d)|\to 0,$$ in probability (or almost surely) w.r.t.\ the distribution of the observed data, for a chosen point-wise consistent estimator $\widehat{\nu}(\cdot)$.

Such uniform consistency is guaranteed if the score $S$ has a finite discrete support, as is the case for almost all cognitive screening tests. In such a case the supremum over the support of $S$ reduces to a maximum over a finite number of differences, all of which converge to zero in probability thanks to point-wise consistency, yielding uniform convergence in probability. 

However, when $S$ has an infinite support, e.g.\ an absolutely continuous distribution w.r.t.\ the Lebesgue measure, uniform consistency does not follow from point-wise consistency.  We will slightly alter the estimators before introduced based on a data-splitting scheme \cite{cox1975note}, and show that under standard conditions these estimators achieve uniform (in $L_\infty$) consistency.

\subsubsection{The simple estimator}\label{sec:AppUniformConsistency_Simple} The simple estimator $\widehat{\nu}_{\rm simple}(s)=\nu(\Tilde{\mathbb{P}}_n,s)$ only uses the available data to fit the models, but does not relate to any empirical averages over the observed data, which is what concerns most of the empirical process theory \cite{gyorfi2002distribution,vandervaart1996weak}. As a statement about uniform convergence of the simple estimator relates only to the posited models $\Tilde{\mathbb{P}}$ and how they are fitted, we just state here the required uniform consistency $$\sup_{s\in{\rm supp}(S)}|\widehat{\nu}_{\rm simple}(s)-\nu(\mathbb{P},s)|\to 0,$$ in probability w.r.t.\ the distribution of the observed data.

\subsubsection{Weighted estimators}\label{sec:AppUniformConsistency_Weighted}
In contrast with the simple estimator, all the weighted estimators we have introduced have the same form: data is used to fit certain functions, and then an empirical average is taken of the fitted functions \textit{using the same data}. This double-fold use of the observed data  complicates the analysis at hand. As a practical (and easily implementable) solution, we suggest data splitting \cite{cox1975note}. For a dataset of size $n$ we chose $I_1$ randomly from  the uniform distribution of all subsets of $\{1,\ldots,n\}$ of size in $\{\lceil n/2\rceil,\lfloor n/2 \rfloor\}$, and define $I_2\defeq\{1,\ldots,n\}\backslash I_1$. The posited models are fitted over $I_1$, and then the empirical average is taken over $I_2$. Under this framework, all weighted estimators that we consider are of the form $$\widehat{\nu}_{\rm weighted}(s)=\widehat{\mathbb{P}}_{I_2}\widehat{g}_s,$$
where $g_s$ is a function of a realization of the process satisfying that $\mathbb{P}g_s=\mathbb{E}[g_s]=\nu(\mathbb{P},s)$, $\widehat{g}_s$ is an estimate of $\widehat{g}_s$ using the data indexed by $I_1$, and $\widehat{\mathbb{P}}_{I_2}f$ denotes an empirical average of $f$ over the data indexed by $I_2$. Here we have borrowed common notation from the empirical process theory literature. See \cite{gyorfi2002distribution} for details.

The independence of the two splits of the original data under the i.i.d.\ assumption will allow us to achieve  the desired uniform consistency under the following assumptions. 
\begin{assumption}\label{ass:ModelWellEstimated}
    The posited models and their estimation from the data in fold $I_1$ are such that $$\sup_{s\in{\rm supp}(S)}\mathbb{P}|\widehat{g}_s-g_s|\to 0,$$ in probability w.r.t.\ the distribution of the observed data in fold $I_1$. 
\end{assumption}
Statements of the form $\mathbb{P}\widehat{g}_s$ mean that the expectation is taken over the distribution of a new independent realization of $\mathbb{P}$, i.e. conditional on all other randomness. Assumption \ref{ass:ModelWellEstimated} states that in expectation $\widehat{g}_s$ approximates $g_s$ uniformly in probability. It is therefore a statement about the posited models and how they are fitted, which strengthens the point-wise consistency we discussed in Appendix \ref{sec:FDP_Estimation}. We will also need to make an assumption concerning the outer empirical average.
\begin{assumption}\label{ass:GlivenkoCantelli}
    There exists a function class $\mathcal{F}$ such that $\{\widehat{g}_s\colon s\in{\rm supp}(S)\}\subseteq \mathcal{F}$ almost surely w.r.t.\ the joint distribution of the data in fold $I_1$. Furthermore, $\mathcal{F}$ is $\mathbb{P}$-Glivenko-Cantelli \cite[Chap.\ 2.1]{vandervaart1996weak}.
\end{assumption}
This essentially states that the class of functions $\{\widehat{g}_s\colon s\in{\rm supp}(S)\}$ almost surely belongs to a class which is ``regular w.r.t.\ $\mathbb{P}$''. There are many characterizations of Glivenko-Cantelli classes, involving bracketing numbers and VC-dimensions. See \cite[Chap.\ 2.1]{vandervaart1996weak} or \cite[Chap.\ 9]{gyorfi2002distribution} for details. Nevertheless, as stated in \cite{galvao2015uniformly} and references therein, Assumption \ref{ass:GlivenkoCantelli} is standard, and thus we present it as such. 

Under these assumptions, we are able to conclude a uniform consistency result.
\begin{theorem}\label{thm:GlivenkoCantelli}
    Under Assumptions \ref{ass:ModelWellEstimated} and \ref{ass:GlivenkoCantelli} it holds that $$\sup_{s\in{\rm supp}(S)}\left|\widehat{\mathbb{P}}_{I_2}\widehat{g}_s-\mathbb{P}g_s\right|\to 0,$$ in probability w.r.t.\ the joint distribution of the observed data.
\end{theorem}
\begin{proof}
    By the triangle inequality $$\sup_{s\in{\rm supp}(S)}\left|\widehat{\mathbb{P}}_{I_2}\widehat{g}_s-\mathbb{P}g_s\right|\leq \underbrace{\sup_{s\in{\rm supp}(S)}\left|\widehat{\mathbb{P}}_{I_2}\widehat{g}_s-\mathbb{P}\widehat{g}_s\right|}_{(I)} + \underbrace{\sup_{s\in{\rm supp}(S)}\left|\mathbb{P}(\widehat{g}_s-g_s)\right|}_{(II)}. $$ By Assumption \ref{ass:ModelWellEstimated} we have that $$(II)\leq \sup_{s\in{\rm supp}(S)}\mathbb{P}|\widehat{g}_s-g_s|\to 0,$$ in probability. Now for $(I)$, conditional on the data in fold $I_1$,  Assumption \ref{ass:GlivenkoCantelli} yields $$(I)\leq \sup_{f\in\mathcal{F}}\left|(\widehat{\mathbb{P}}_{I_2}-\mathbb{P})f\right|\to 0$$ in probability w.r.t.\ the joint distribution of the data in fold $I_2$. This means that for all $\varepsilon>0$ $$\mathbb{P}\left(\sup_{s\in{\rm supp}(S)}\left|\widehat{\mathbb{P}}_{I_2}\widehat{g}_s-\mathbb{P}\widehat{g}_s\right|>\varepsilon\mid \mathcal{T}_1\right)\to 0,$$ where $\mathcal{T}_1$ is the $\sigma$-algebra induced by the data in fold $I_1$. Then by the law of total expectations and the dominated convergence theorem $$\mathbb{P}\left(\sup_{s\in{\rm supp}(S)}\left|\widehat{\mathbb{P}}_{I_2}\widehat{g}_s-\mathbb{P}\widehat{g}_s\right|>\varepsilon\right)=\mathbb{E}\left( \mathbb{P}\left(\sup_{s\in{\rm supp}(S)}\left|\widehat{\mathbb{P}}_{I_2}\widehat{g}_s-\mathbb{P}\widehat{g}_s\right|>\varepsilon   \mid \mathcal{T}_1\right)\right)\to 0$$ as $n\to\infty$. As a result, $(I)$ converges to zero in probability, which concludes the proof.
\end{proof}

Thus we see how under the corresponding assumptions for point-wise consistency, together with Assumptions \ref{ass:ModelWellEstimated} and \ref{ass:GlivenkoCantelli} for the specific form of the $g_s$ functions given by each estimator, we achieve uniform consistency of our weighted estimators after performing the data-splitting procedure before mentioned.

\subsubsection{The one-step estimator introduced in Appendix \ref{sec:FDP_Estimation}}
Lastly, we discuss the one-step estimator, constructed in the case when $\mathcal{V}$ is discrete. Recall that this estimator was of the form $$\widehat{\mathbb{P}}_n\nu^1(\Tilde{\mathbb{P}}_n,s)+\nu(\Tilde{\mathbb{P}}_n,s),$$ where the models $\Tilde{\mathbb{P}}_n$ consistently estimated posited models $\Tilde{\mathbb{P}}$, which could be miss-specified as per Theorem \ref{thm:DR}. As we also saw in the proof of this theorem, $\nu(\Tilde{\mathbb{P}}_n,s)$ is a point-wise consistent estimator of $\nu(\Tilde{\mathbb{P}},s)$ (it is the simple estimator with potentially miss-specified models) and $\widehat{\mathbb{P}}_n\nu^1(\Tilde{\mathbb{P}}_n,s)$ was point-wise consistent for $\nu({\mathbb{P}},s)-\nu(\Tilde{\mathbb{P}},s)$. 

For the simple estimator addend, as we did in Appendix \ref{sec:AppUniformConsistency_Simple} we just assume that the (potentially miss-specified) posited models and their estimation methods satisfy the Glivenko-Cantelli property $$\sup_{s\in{\rm supp}(S)}\left|\nu(\Tilde{\mathbb{P}}_n,s)-\nu(\Tilde{\mathbb{P}},s)\right|\to 0$$ in probability. For the correction provided by the influence function, we run into the same issue as in Appendix \ref{sec:AppUniformConsistency_Weighted}: we use the full data to fit the models $\Tilde{\mathbb{P}}_n$ than to take the outer empirical average. Again, to overcome this issue we propose using again  the data-splitting method introduced before. Thus,  we define $g_s=\nu^1(\Tilde{\mathbb{P}},s)$ and $\widehat{g}_s=\nu^1(\Tilde{\mathbb{P}}_{I_1},s)$. Under Assumptions \ref{ass:ModelWellEstimated} and \ref{ass:GlivenkoCantelli}, Theorem \ref{thm:GlivenkoCantelli} implies that $$\sup_{s\in{\rm supp}(S)}\left|\widehat{\mathbb{P}}_{I_2}\nu^1(\Tilde{\mathbb{P}}_{I_1},s)-(\nu({\mathbb{P}},s)-\nu(\Tilde{\mathbb{P}},s))\right|\to 0,$$ in probability. Thus, we see how data-splitting in the influence function correction allows us to achieve, under assumptions on the posited models and their estimation, uniform consistency of the one-step estimator. 

The double-robustness property studied in \ref{thm:DR} extends as well to the uniform case, as we have shown that $$\widehat{\mathbb{P}}_{I_2}\nu^1(\Tilde{\mathbb{P}}_{I_1},s)+\nu(\Tilde{\mathbb{P}}_n,s)$$ is uniformly consistent for $\nu(\mathbb{P},s)$ despite the models being miss-specified as allowed by said theorem.

 \subsection{Testing for demographic insensitivity}\label{sec:FDP_Testing}

We have so far seen that under appropriate conditions we can identify and consistently estimate the counterfactual distribution required to assess whether a score is insensitive to certain demographics. Yet, Definitions \ref{def:FDP}-\ref{def:FDP_weak} require the equality (point-wise or in distribution) across the support of $(V,D)$. Therefore, if we consider the individual null hypothesis in the strong case $$\mathcal{H}_{0,(v,d)}^{(w,w')}\colon S^{w}\mid \{V=v,D=d\}\stackrel{d}{=} S^{w'}\mid\{V=v,D=d\}, $$ or in the weak case for a fixed threshold $s$ 
\begin{equation}\label{eq:FDP_IndividualNullWeak}
    \mathcal{H}_{0,(v,d)}^{(w,w')}\colon \mathbb{P}(S^{w}\leq s\mid V=v,D=d)\stackrel{}{=} \mathbb{P}(S^{w'}\leq s\mid V=v,D=d), 
\end{equation}
the global null hypothesis of demographic insensitivity can be expressed as 
\begin{align}\label{eq:GlobalIntersectionNull}
    \mathcal{H}_{0}=\bigcap_{\substack{v\in\mathcal{V}\\d\in\{0,1\}}} \mathcal{H}_{0,(v,d)}=\bigcap_{\substack{v\in\mathcal{V}\\d\in\{0,1\}}}  \underbrace{\bigcap_{\substack{(w,w')\in\mathcal{V}^2\\w\neq w'}} \mathcal{H}_{0,(v,d)}^{(w,w')}}_{\mathcal{H}_{0,(v,d)}},
\end{align}
which is an intersection hypothesis. Assume that $V$ has a discrete finite support, and that we are able to compute valid p-values for each of the individual null hypotheses $\mathcal{H}_{0,(v,d)}^{(w,w')}$. In the strong case, these could come from a sup-norm or Cram\'er von Mises type test based on the uniform consistency discussed in Appendix \ref{sec:AppFDP_UniformConsistency}. The main goal is then to combine these p-values to obtain a valid p-value for the global intersection null $\mathcal{H}_0$. Without making any assumptions on the dependence structure of the individual p-values, a valid approach could be a (weighted) Bonferroni procedure \cite[Sec.\ 11.1.1]{wassmer2016group}. If we are also interesting in valid testing for the individual nulls $\mathcal{H}_{0,(v,d)}^{(w,w')}$ or the intermediates $\mathcal{H}_{0,(v,d)}$, the closed testing procedure \cite{marcus1976closed} can be used  to maintain Type-I error control. 

When $\mathcal{V}$ is infinite, the former procedures cannot be applied, as we have an infinite collection of individual hypotheses. In this case  (rarely found in practice) the procedure described in \cite{blanchard2014testing} can be applied  to test the global intersection null when $\mathcal{V}$ is not finite while ensuring Type-I error control.

\section{Classification accuracy and demographic insensitivity in a parametric example}\label{sec:AppParametricModel}
Consider a particular case of the model introduced in Remark \ref{remark:ReciprocateNotTrue}, where $V$ is a univariate random variable with an absolutely continuous distribution w.r.t.\ the Lebesgue measure, $V\sim\mathbb{P}_V$. Furthermore, assume that $\mathbb{P}_V$ supported on an open connected interval, and consider the following hierarchical model
\begin{equation*}
    \begin{cases}
        D\mid V\sim Ber(\expit(\beta_0+\beta_1 V))\\
        X\mid D,V\sim\mathcal{N}(\mu_0+\mu_1V+\mu_2D,\sigma^2)
    \end{cases},
\end{equation*}
where $Ber(p)$ denotes a Bernoulli distribution with success probability $p\in(0,1)$, $\sigma>0$ and $\expit(x)=(1+\exp(-x))^{-1}$. One could then represent $X\stackrel{d}{=}\mu_0+\mu_1V+\mu_2D+\varepsilon$, with $\varepsilon\sim\mathcal{N}(0,\sigma^2)$ independent of $D$ and $V$. Then, under the correction in Equation \eqref{eq:AgeEduCorrection}, we have $$Z\stackrel{d}{=}\frac{\mu_2D+\varepsilon}{\sigma},$$ meaning that Assumption \ref{ass:WellDesignedCorrection} is satisfied by design: $Z\independent V\mid D=0$.
\subsection{Superiority of the raw score}\label{sec:AppParametricModel_X}
We now consider Assumption \ref{ass:TestWellDesigned}, for which  $$LR_{X\mid V}(x\mid v)=\frac{\phi_\sigma(x-\mu_0-\mu_2-\mu_1v)}{\phi_\sigma(x-\mu_0-\mu_1v)}=\exp(-\sigma^{-2}(\mu_2^2/2-\mu_2(x-\mu_0-\mu_1v))),$$ where $\phi_\sigma(\cdot)=\phi(\cdot/\sigma)/\sigma$, and $\phi(\cdot)$ is the density of a standard normal random variable. Therefore $$\partial_x\log(LR_{X\mid V}(x\mid v))=\mu_2/\sigma^2,$$ meaning Assumption \ref{ass:TestWellDesigned} is satisfied in our setting if and only if $\mu_2\leq 0$. 

Lastly we study Assumption \ref{ass:Comonotonicity}. Denote $\pi(v)=\expit(\beta_0+\beta_1 v)$ and $\pi(x,v)=\mathbb{P}(D=1\mid X=x,V=v)$. One can easily check that
\begin{align*}
    \logit(\pi(x,v))&=\logit(\pi(v))+\log(LR_{X\mid V}(x\mid v))\\&=\beta_0+\beta_1v-\frac{\mu_2^2/2-\mu_2(x-\mu_0-\mu_1v)}{\sigma^2},
\end{align*}
where $\logit(p)=\log(p/(1-p))$ for $p\in(0,1)$. This means that  $$\partial_v\logit(\pi(x,v))=\beta_1-\frac{\mu_1\mu_2}{\sigma^2}.$$ Furthermore, the other quantity of interest, as $X$ has an absolutely continuous distribution,  is $$\mathbb{P}(X\leq x\mid V=v,D=0)=\Phi\left(\frac{x-\mu_0-\mu_1v}{\sigma}\right),$$ where $\Phi$ is the  standard normal c.d.f. This leads to $$\partial_v\mathbb{P}(X\leq x\mid V=v,D=0)=-\mu_1\phi_\sigma(x-\mu_0-\mu_1v).$$ The co-monotonicity required by Assumption \ref{ass:Comonotonicity} translates in our case to the two derivatives we have computed having constant and equal sign. Hence,  Assumption \ref{ass:Comonotonicity} boils down to 
\begin{equation}\label{eq:ParametricConditionC}
    -\mu_1\left(\beta_1-\frac{\mu_1\mu_2}{\sigma^2}\right)\geq 0.
\end{equation}
As we mentioned in the main text, we can see how condition \eqref{eq:ParametricConditionC} holds if and only if Assumption \ref{ass:Comonotonicity} holds \textit{for all} $\alpha\in(0,1)$.
\subsection{Superiority of the corrected score}\label{sec:AppParametricModel_Z}
We now check the conditions of Theorem \ref{thm:Superiority_Z}. Firstly $$LR_{Z|V}(z|v)=\dfrac{\phi(z-\mu_2/\sigma)}{\phi(z)}=\exp\left(\frac{z\mu_2}{\sigma}-\frac{\mu_2^2}{2\sigma^2}\right),$$ and therefore Assumption \ref{ass:TestWellDesigned} is satisfied for $Z$ iff $\mu_2\leq0$, which was also the condition for this assumption to be satisfied for $X$, as we already mentioned in the main text. 

We now turn our attention to Assumption \ref{ass:Countermonotonicity}. Analogously computations to those done in Appendix \ref{sec:AppParametricModel_X} yield that $$\partial_v\logit(\mathbb{P}(D= 1\mid Z=z,V=v))=\beta_1,$$ meaning Assumption \ref{ass:Countermonotonicity} is satisfied if and only if 
\begin{equation}\label{eq:ParametricCondition_Z}
    \beta_1\mu_1\geq 0.
\end{equation}
Under the assumption that $\mu_2<0$, Equations \eqref{eq:ParametricConditionC} and \eqref{eq:ParametricCondition_Z} are simultaneously satisfied if and only if $\mu_1=0$. In this case $X\independent V|D$ and trivially the ROC curves for $X$ and $Z$ are equal, in accordance with Corollary \ref{coro:Equality}.
\subsection{Demographic insensitivity}\label{sec:AppParametricModel_FDP}
We now turn to discussing whether raw and corrected scores are insensitive to  demographics in our simplified parametric setting. When comparing this to the formalism introduced in Section \ref{sec:FDP} and Appendix \ref{sec:App_FDP}, we see that we are in a setting where $L=W=\emptyset$, also discussed in Section \ref{sec:FDP_Main}. In this case, a score $S$ satisfies demographic insensitivity if and only if $S\independent V\mid D$. We see how in this example  the corrected z-score $Z$ insensitive to demographics.

The raw score $X$ would be insensitive to demographics if and only if $\mu_1=0$. This corresponds to the case  of the absence of the edge $V\to X$ in the system's causal DAG. This implies $X\independent V\mid D$, and thus under any Type I correction raw and corrected scores have the same marginal classification accuracy (see Lemma \ref{lemma:IndependenceImplication} and Corollary \ref{coro:Equality}). In our example  $Z\stackrel{d}{=}(X-\mu_0)/\sigma$, and we see how the previous result holds as $Z$ is a strictly increasing transformation of $X$.

To summarize, we see how in our parametric example the corrected z-score is always demographic insensitive, even in the cases where the raw score has better classification accuracy. Furthermore, if the raw score satisfies demographic insensitivity then the correction is irrelevant towards classification accuracy, as both raw and corrected scores are demographic insensitive.

\section{Analysis of the OASIS-3 data}\label{sec:AppOASIS}
\subsection{The classification accuracy comparison in Section \ref{sec:OASIS3_Superiority}}\label{sec:AppOASIS_SuperiorityCLT}
In Section \ref{sec:OASIS3_Superiority} we estimated  $$\tau(t)\defeq \mathbb{E}\left[LR_{(X,V)}(t,V)(\alpha_V^X(t)-\alpha(t))\mid D=0\right],$$ where recall that  $\alpha_v^X=\alpha^X_v(t)=F_{X,0}(t|v)$ and $\alpha=\alpha(t)=F_{X,0}(t)$. As $V$ has a discrete support we can expand $\tau(t)=\theta_1-\theta_2$, where
\begin{subequations}\label{eq:SuperiorityFunctions}
    \begin{align}
    \theta_1&=g_1(\boldsymbol{p})=\frac{1}{\pi_1}\sum_v\frac{p(t,v,1)}{p(t,v,0)}\left(\sum_{x\leq t}p(x,v,0)\right),\label{eq:SuperiorityFunction_1}\\
    \theta_2&=g_2(\boldsymbol{p})=\frac{T}{\pi_1(1-\pi_1)}\sum_v\frac{p(t,v,1)}{p(t,v,0)}\left(\sum_{x}p(x,v,0)\right),\label{eq:SuperiorityFunction_2}
\end{align}
\end{subequations}

where we have defined $p(x,v,d)=\mathbb{P}(X=x,V=v,D=d)$, $\pi_1=\mathbb{P}(D=1)=\sum_{x,v}p(x,v,1)$, and $T=\mathbb{P}(X\leq t, D=0)=\sum_{v,x\leq t}p(x,v,0)$, and by $\boldsymbol{p}$ we denote the vector of all $p(x,v,d)$ cell probabilities. The estimator introduced in Section \ref{sec:OASIS3_Superiority} is a plug-in estimator where each of the cell probabilities is estimated by a sample mean $$\widehat{p}(x,v,d)=\widehat{\mathbb{P}}_n(X=x,V=v,D=d).$$ The next theorem provides a central limit theorem with an expression of the asymptotic variance for the resulting plug-in estimator $\widehat{\tau}(t)$. 
\begin{theorem}\label{thm:CLT_Superiority}
    The plug-in estimator $\widehat{\tau}(t)$ built using sample averages estimators for the cell probabilities satisfies $$\sqrt{n}(\widehat{\tau}(t)-\tau(t))\stackrel{\mathcal{D}}{\to}\mathcal{N}(0,\sigma^2),$$ where $\sigma^2=\mathbb{E}[\xi_1^2] + \mathbb{E}[\xi_2^2]-2\mathbb{E}[\xi_1\xi_2]$ and we have defined
    \begin{align*}
        b_v&=\frac{p(t,v,1)}{p(t,v,0)},\\
        c_v&=\sum_{x\leq t}\frac{p(x,v,0)}{p(t,v,0)},\\
        \xi_1&=\frac{I(D=1)}{\pi_1}(I(X=t)c_V-\theta_1)+\frac{I(D=0)}{\pi_1}b_V(I(X\leq t)-I(X=t)c_V),\\
        r_v&=p(t,v,1),\\
        a_v&=p(t,v,0),\\
        k_v&=\sum_xp(x,v,0),\\
        U&=\sum_v k_v\frac{r_v}{a_v},\\
        B&=\frac{T}{\pi_1(1-\pi_1)},\\
        A&=T\frac{2\pi_1-1}{(\pi_1(1-\pi_1))^2}U,\\
        \xi_2&=I(D=1)\left(A+I(X=t)B\frac{k_V}{a_V}\right)\\&\quad+I(D=0)\left(I(X\leq t)\frac{U}{\pi_1(1-\pi_1)}+B\frac{r_V}{a_V}\left(1-I(X=t)\frac{k_V}{a_V}\right)\right)\\&\quad-(A\pi_1+2BU).
    \end{align*}
    Furthermore, a consistent estimator of the asymptotic variance can be built  in a plug-in fashion as 
    \begin{equation}\label{eq:CLT_Superiority_VarianceEstimator}
        \widehat{\sigma^2}_n=\widehat{\mathbb{E}}_n\left[\widehat{\xi}_1^2\right] + \widehat{\mathbb{E}}_n\left[\widehat{\xi}_2^2\right]-2\widehat{\mathbb{E}}_n\left[\widehat{\xi}_1\widehat{\xi}_2\right],
    \end{equation}
    where $\widehat{\xi}_1$ and $\widehat{\xi}_2$ are constructed by plugging-in the estimated cell probabilities.
\end{theorem}
\begin{proof}
    Denote by $Q=(X,V,D)$ and $\widehat{\tau}(t)=g(\widehat{\boldsymbol{p}})$. As our estimator is a plug-in estimator $\widehat{\theta}=g(\widehat{\boldsymbol{p}})$ the (functional) $\delta$-method \cite[Thms. 3.1 \& 20.8]{Vaart1998asymptotic} tells us that 
    \begin{align*}
        \sqrt{n}(\widehat{\tau}(t)-{\tau}(t))&=\sqrt{n}\nabla_{\boldsymbol{p}}g(\boldsymbol{p})'(\widehat{\boldsymbol{p}}-\boldsymbol{p})+o_{\mathbb{P}}(1)\\
        &=\sqrt{n}\left[\frac{1}{n}\sum_{i=1}^n \underbrace{\nabla_{\boldsymbol{p}}g(\boldsymbol{p})'(\boldsymbol{e}(Q_i)-\boldsymbol{p})}_{=\colon \varphi(Q_i)}  \right] +o_{\mathbb{P}}(1)\\
        &=\sqrt{n}\widehat{\mathbb{E}}_n[\varphi(Q)]+o_{\mathbb{P}}(1),
    \end{align*}
    where $\boldsymbol{e}(Q)$ is the ``one-hot'' vector with components $\boldsymbol{e}_{(x,v,d)}(Q)=I(Q=(x,v,d))$. It holds that $\mathbb{E}[\varphi(Q)]=0$ because $\mathbb{E}[\boldsymbol{e}(Q)]=\boldsymbol{p}$, and $\varphi(Q)$ has finite variance as $\nabla_{\boldsymbol{p}}g(\boldsymbol{p})$ exists and has finite components (as we will see). Therefore, thanks to the i.i.d. nature of the observed data and the standard central limit theorem, we have that $$\sqrt{n}(\widehat{\tau}(t)-{\tau}(t))\stackrel{\mathcal{D}}{\to}\mathcal{N}({0},\mathbb{E}[\varphi(Q)^2]).$$ Therefore, it only remains to compute $\varphi(Q)$. With $g=g_1-g_2$ as defined in Equation \eqref{eq:SuperiorityFunctions} we have $$\varphi(Q)=\underbrace{\nabla_{\boldsymbol{p}}g_1(\boldsymbol{p})'(\boldsymbol{e}(Q)-\boldsymbol{p})}_{\xi_1}-\underbrace{\nabla_{\boldsymbol{p}}g_2(\boldsymbol{p})'(\boldsymbol{e}(Q)-\boldsymbol{p})}_{\xi_2}=\xi_1-\xi_2,$$ and thus $\sigma^2=\mathbb{E}[\varphi(Q)^2]=\mathbb{E}[\xi_1^2] + \mathbb{E}[\xi_2^2]-2\mathbb{E}[\xi_1\xi_2]$. It only remains to show that the expressions of the $\xi$'s correspond to the ones given in this theorem. We present here the calculations for $\xi_1$.  Attending to the definition of $g_1$ in Equation \eqref{eq:SuperiorityFunction_1}, we apply the chain rule for a $d=1$ cell  and conclude that
    \begin{align*}
        \partial_{p(x,v,1)}g_1(\boldsymbol{p})&=\frac{I(x=t)}{\pi_1}\frac{\sum_{x'\leq t}p(x',v,0)}{p(t,v,0)}-\frac{\theta_1}{\pi_1}.
    \end{align*}
    For a $d=0$ cell we have
    \begin{align*}
        \partial_{p(x,v,0)}g_1(\boldsymbol{p})&=\begin{cases}
            \frac{1}{\pi_1}\frac{p(t,v,1)}{p(t,v,0)}&x<t\\
            0&x>t\\
            \frac{1}{\pi_1}\frac{p(t,v,1)}{p(t,v,0)}\left(1-\sum_{x'\leq t}\frac{p(x',v,0)}{p(t,v,0)}\right)&x=t
        \end{cases}\\
        &=\frac{1}{\pi_1}\frac{p(t,v,1)}{p(t,v,0)}\left(I(x\leq t)-I(x=t)\sum_{x'\leq t}\frac{p(x',v,0)}{p(t,v,0)}\right).
    \end{align*}
    We have $$\xi_1=\sum_{(x,v,d)}\partial_{p(x,v,d)}g_1(\boldsymbol{p})\boldsymbol{e}_{(x,v,d)}(Q)-\underbrace{\sum_{(x,v,d)}\partial_{p(x,v,d)}g_1(\boldsymbol{p})p(x,v,d)}_{(K)},$$ but this second term $(K)$ adds up to zero, as 
    $$\sum_{(x,v)}\partial_{p(x,v,1)}g_1(\boldsymbol{p})p(x,v,1)=\frac{1}{\pi_1}\sum_vp(t,v,1)c_v-\frac{\theta_1}{\pi_1}\sum_{(x,v)}p(x,v,1)=\theta_1-\theta_1=0,$$ and 
    \begin{align*}
        \sum_{(x,v)}\partial_{p(x,v,0)}g_1(\boldsymbol{p})p(x,v,0)&=\frac{1}{\pi_1}\sum_vb_v\left(\sum_{x\leq t}p(x,v,0)\right)-\frac{1}{\pi_1}\sum_vb_vc_vp(t,v,0)\\&=\theta_1-\frac{1}{\pi_1}\sum_vb_v\left(\sum_{x\leq t}p(x,v,0)\right)=\theta_1-\theta_1=0.
    \end{align*}
    As a result, we conclude that $$\xi_1=\sum_{(x,v,d)}\partial_{p(x,v,d)}g_1(\boldsymbol{p})\boldsymbol{e}_{(x,v,d)}(Q)=\partial_{p(Q)}g_1(\boldsymbol{p}),$$ which is precisely the expression we wanted. The computation of $\xi_2$, while a bit more arduous, follows the same steps, with the only difference that the term analogous to $(K)$ will not be zero, but rather $A\pi_1+2BU$.

    Lastly, the fact that the estimator proposed in Equation \eqref{eq:CLT_Superiority_VarianceEstimator} is consistent for $\mathbb{E}[\varphi(Q)^2]$ follows from the law of large numbers and the continuous mapping theorem.
\end{proof}
Theorem \ref{thm:CLT_Superiority} can be used to construct an asymptotically exact one-sided confidence interval at level $\alpha^*$ for $\tau(t)$, with lower bound $$\widehat{\tau}(t)-\Phi^{-1}(1-\alpha^*)\sqrt{\widehat{\sigma^2}_n/n}.$$ These are the values reported in Table \ref{tab:OASIS_Superiority_NP} in Section \ref{sec:OASIS3_Superiority} used to test the superiority of the raw MMSE score. 

\subsection{The demographic insensitivity study of Section \ref{sec:OASIS3_FDP}}\label{sec:App_OASIS_FDP}
\subsubsection{General testing set-up}\label{sec:App_OASIS_FDP_CLT}
As explained in Section \ref{sec:OASIS3_FDP}, as all variables in the system have discrete support, we estimate the counterfactual probabilities of interest  using the simple estimator, where all factual probabilities are estimated by sample averages. We can expand the target probabilities, similarly to what we did in Section \ref{sec:AppOASIS_SuperiorityCLT}, to clarify the dependence with the cell probabilities as
\begin{align*}
    &\mathbb{P}(Y^{v_S}=1\mid V=v,D=d)=g_{(v_S,v,d)}(\boldsymbol{p})\\=&\sum_{l,w}p(1,l,w,v_S,d)\cdot\frac{\sum_{y=0}^1p(y,l,w,v,d)}{\sum_{y=0}^1p(y,l,w,v_S,d)}\cdot\frac{1}{\sum_{l',w'}\sum_{y'=0}^1p(y',l',w',v,d)},
\end{align*}
where we have defined $Y=I(S\leq t)$ for a score $S$ and a fixed threshold $t$,   the cell probabilities $p(y,l,w,v,d)=\mathbb{P}(Y=y,L=l,W=w,V=v,D=d)$ and $\boldsymbol{p}$ denotes the vector with all these cell probabilities.

As explained in Section \ref{sec:OASIS3_FDP}, and analogously to Section \ref{sec:AppOASIS_SuperiorityCLT}, the simple estimator is a plug-in one based on empirical averages of the cell probabilities in $\boldsymbol{p}$. Now denote by $\boldsymbol{\theta}=\boldsymbol{g}(\boldsymbol{p})$ the vector of all counterfactual probabilities of interest, and by $\boldsymbol{\eta}=\logit(\boldsymbol{\theta})$ the vector of these logit-probabilities, where we understand the logit map applied component-wise. Our analysis is based on the following central limit theorem. 
\begin{theorem}\label{thm:CLT_FPD}
    The plug-in estimator $\widehat{\boldsymbol{\eta}}=\logit(\boldsymbol{g}(\widehat{\boldsymbol{p}}))$ of the logit-probabilities of interest for the demographic insensitivity comparison satisfies
    $$\sqrt{n}(\widehat{\boldsymbol{\eta}}-\boldsymbol{\eta})\stackrel{\mathcal{D}}{\to}\mathcal{N}(\boldsymbol{0},\Sigma_{\boldsymbol{\eta}}),$$
    where $\Sigma_{\boldsymbol{\eta}}=HG(\diag(\boldsymbol{p})-\boldsymbol{p}\boldsymbol{p}')G'H$. The matrix $G$ is the Jacobian matrix of the map $\boldsymbol{g}$ evaluated at $\boldsymbol{p}$, with entries given by 
    \begin{align*}
        &\partial_{p(y',l',w',v',d')}g_{(v_S,v,d)}(\boldsymbol{p})\\=&\begin{cases}
            \frac{1}{b(v,d)}(\frac{p(1,l',w',v_S,d)}{a(l',w',v_S,d)}-g_{(v_S,v,d)}(\boldsymbol{p})) & d=d', v_S\neq v'=v\\
            \frac{a(l',w',v,d)}{b(v,d)a(l',w',v_S,d)}\left(I(y'=1)-\frac{p(1,l',w',v_S,d)}{a(l',w',v_S,d)}\right)& d=d', v_S= v'\neq v\\
            \frac{1}{b(v,d)}(I(y'=1)-g_{(v,v,d)}(\boldsymbol{p})) & d=d', v_S= v'= v\\
            0 & \text{otherwise}
        \end{cases},
    \end{align*}
    where we have further defined
    \begin{align*}
        a(l,w,v,d)&=\mathbb{P}(L=l,V=v,D=d,W=w)= \sum_{y=0}^1p(y,l,w,v,d),\\
        b(v,d)&=\mathbb{P}(V=v,D=d)=\sum_{l',w'}\sum_{y'=0}^1p(y',l',w',v,d).
    \end{align*}
    Finally, the $H$ matrix is the (diagonal) Jacobian matrix of the component-wise $\logit$ map evaluated in $\boldsymbol{\theta}$. 
\end{theorem}
\begin{proof}
We begin by noting that, as the cell probabilities are estimated by sample means, the multinomial CLT \cite[Ex. 2.18]{Vaart1998asymptotic} tells us that $$\sqrt{n}(\widehat{\boldsymbol{p}}-\boldsymbol{p})\stackrel{\mathcal{D}}{\to}\mathcal{N}(\boldsymbol{0},\Sigma_{\boldsymbol{p}}), $$ where $\Sigma_{\boldsymbol{p}}=\diag(\boldsymbol{p})-\boldsymbol{p}\boldsymbol{p}'$. By the (functional) $\delta$-method \cite[Thms. 3.1 \& 20.8]{Vaart1998asymptotic}, we have that $$\sqrt{n}(\widehat{\boldsymbol{\theta}}-\boldsymbol{\theta})\stackrel{\mathcal{D}}{\to}\mathcal{N}(\boldsymbol{0},G\Sigma_{\boldsymbol{p}}G'), $$ where $G$ is $\boldsymbol{g}$'s Jacobian matrix $G=D_{\boldsymbol{p}}\boldsymbol{g}(\boldsymbol{p})$. Lastly, applying the component-wise logit map and again the $\delta$-method we conclude that $$\sqrt{n}(\widehat{\boldsymbol{\eta}}-\boldsymbol{\eta})\stackrel{\mathcal{D}}{\to}\mathcal{N}(\boldsymbol{0},HG\Sigma_{\boldsymbol{p}}G'H), $$ where $H=\diag(h(\boldsymbol{\theta}))$,  $h(x)=1/(x(1-x))$, and we understand $h(\boldsymbol{\theta})$ as applied component-wise. 

It only remains to compute $G$. To that extent, we begin by writing each component of $\boldsymbol{g}$ in terms of the $a$ and $b$ functions as $$\boldsymbol{g}_{(v_S,v,d)}(\boldsymbol{p})=\frac{1}{b(v,d)}\sum_{l,w}p(1,l,w,v_S,d)\frac{a(l,w,v,d)}{a(l,w,v_S,d)}.$$ Fix a $p(y',l',w',v',d')$ cell. Then  $\boldsymbol{g}_{(v_S,v,d)}(\boldsymbol{p})$ only involves $p(y',l',w',v,d)$ and $p(y',l',w',v_S,d)$  cells, so we can restrict to the case $d'=d$ and $v'\in\{v,v_S\}$. We then consider three cases. First, if $v_S\neq v'=v$ we have that $p(y',l',w',v,d)$ enters $\boldsymbol{g}_{(v_S,v,d)}(\boldsymbol{p})$ in $b(v,d)$ and $a(l,w,v,d)$, yielding by the chain rule $$\partial_{p(y',l',w',v',d)}\boldsymbol{g}_{(v_S,v,d)}(\boldsymbol{p})=-\frac{\boldsymbol{g}_{(v_S,v,d)}(\boldsymbol{p})}{b(v,d)}+\frac{p(1,l',w',v_S,d)}{b(v,d)a(l',w',v_S,d)}. $$ Second, if $v_S= v'\neq v$, $p(y',l',w',v_S,d)$ appears in $a(l,w,v_S,d)$ and potentially in $p(1,l,w,v_S,d)$ if $y'=1$, yielding 
\begin{align*}
    \partial_{p(y',l',w',v',d)}\boldsymbol{g}_{(v_S,v,d)}(\boldsymbol{p})=&\frac{I(y'=1)}{b(v,d)}\frac{a(l',w',v,d)}{a(l',w',v_S,d)}\\&-\frac{1}{b(v,d)}p(1,l',w',v_S,d)\frac{a(l',w',v,d)}{a(l',w',v_S,d)^2}\\=&\frac{1}{b(v,d)}\frac{a(l',w',v,d)}{a(l',w',v_S,d)}\left[I(y'=1)-\frac{p(1,l',w',v_S,d)}{a(l',w',v_S,d)}\right]
\end{align*}
and lastly, when  $v_S= v'= v$, we can first rewrite $$\boldsymbol{g}_{(v,v,d)}(\boldsymbol{p})=\frac{1}{b(v,d)}\sum_{l,w}p(1,l,w,v,d),$$ what leads to 
\begin{align*}
    \partial_{p(y',l',w',v,d)}\boldsymbol{g}_{(v,v,d)}(\boldsymbol{p})=&-\frac{1}{b(v,d)^2}\sum_{l,w}p(1,l,w,v,d)+\frac{I(y'=1)}{b(v,d)}\\=&\frac{1}{b(v,d)}\left[I(y'=1)-\boldsymbol{g}_{(v,v,d)}(\boldsymbol{p})\right].
\end{align*}
This concludes the proof.
\end{proof}
Each entry of $\Sigma_\eta$ as given in Theorem \ref{thm:CLT_FPD} can be consistently estimated using the plug-in principle and  the sample average estimate $\widehat{\boldsymbol{p}}$ of the vector of cell probabilities $\boldsymbol{p}$. As each individual null hypothesis $\mathcal{H}_{0,(v,d)}^{(w,w')}$ is equivalent to $\boldsymbol{\eta}_{(w,v,d)}-\boldsymbol{\eta}_{(w',v,d)}=0$, an asymptotically exact p-value for this individual null is $$2\left(1-\Phi\left(\sqrt{\frac{n}{\boldsymbol{\delta}_{(w,w',v,d)}'  \widehat{\Sigma_\eta}\boldsymbol{\delta}_{(w,w',v,d)}}}|\widehat{\boldsymbol{\eta}}'\boldsymbol{\delta}_{(w,w',v,d)}|\right)\right),$$ where $\boldsymbol{\delta}_{(w,w',v,d)}=\boldsymbol{e}(w,v,d)-\boldsymbol{e}(w',v,d)$ and $\boldsymbol{e}(v_S,v,d)$ is the ``one-hot'' vector indicator of component $(v_S,v,d)$. This p-value is symmetric in $(w,w')$, as so is the individual null $\mathcal{H}_{0,(v,d)}^{(w,w')}$. 

\subsubsection{Raw MMSE scores}\label{sec:App_Oasis_FDP_raw}
As explained in Section \ref{sec:OASIS3_FDP}, we will target the weak demographic insensitivity for the raw MMSE score at a threshold of $t=26$. Therefore, our goal  is to jointly estimate the counterfactual probabilities $\{\mathbb{P}(X^{w}\leq 26\mid V=v,D=d)\colon (v,w)\in\mathcal{V}^2,d=0,1\}$. As the socio-economic status is unlikely to be affected by $V_S$, we assume that the dismissible component conditions (Assumption \ref{ass:DCC}) hold under a partition such that $W=W_D$ and $W_S=\emptyset$. Under this assumption all effects of  $V_S$ on the raw MMSE score $X$ in the extended causal DAG are  not mediated by the cognitive impairment status $D$.  We can then estimate the counterfactual probabilities of interest using the simple estimator introduced in Appendix \ref{sec:FDP_Estimation}, with all the probabilities in the right-hand side of Equation \eqref{eq:Identification} estimated by sample averages. This is a consistent estimator as all the variables in the system have a discrete support.  By virtue of this fact, the weighted estimator $\widehat{\nu}_{weighted,X}(\cdot)$ numerically coincides with the simple one. Moreover, as no parametric models have been posited, there is no need to deploy the one-step estimator. We use Theorem \ref{thm:CLT_FPD} to compute asymptotically-valid p-values for the individual null hypotheses $\mathcal{H}_{0,(v,d)}^{(w,w')}$ introduced in Equation \eqref{eq:FDP_IndividualNullWeak}. These p-values can be seen in Table \ref{tab:pValues_FDP_X}. Applying a Bonferroni correction, we reject the global intersection null hypothesis of weak demographic insensitivity at a 5\% level. This means that we have significant evidence towards the fact that the raw MMSE score is not (weakly at $t=26$ or strongly) insensitive to  the age of individuals entering the OASIS-3 study. To provide a more comprehensive picture, we apply a Bonferroni-Holm correction \cite{holm1979simple} (which is equivalent to the closed testing principle \cite{marcus1976closed} with local unweighted Bonferroni adjustments) to test the intermediate intersection nulls $\mathcal{H}_{0,(v,d)}$ as defined in Equation \eqref{eq:GlobalIntersectionNull}. Again at a 5\% level we reject the intermediate intersection nulls $$\{\mathcal{H}_{0,(v,d=1)}\colon v\in\{\text{65-70, above 75}\}\}\cup\{\mathcal{H}_{0,(v,d=0)}\colon v\in\mathcal{V}\}.$$ This means that we have significant evidence that an individual could have had a different probability of being classified as having cognitive impairment, given their age group and  cognitive impairment status, had their sensitive attributes been different at the time of their MMSE assessment.

\begin{table}
\centering
\begin{tabular}{|cccc|}
\hline
\multicolumn{4}{|c|}{$d=0$, $v=$ below 65} \\ \hline
\multicolumn{1}{|c|}{}    & $w'=$ 65-70   & $w'=$ 70-75   & \multicolumn{1}{c|}{$w'=$ above 75} \\ \hline
\multicolumn{1}{|c|}{$w=$ below 65} & 0.126 & 0.039 & \multicolumn{1}{c|}{0.000} \\
\multicolumn{1}{|c|}{$w=$ 65-70}    &       & 0.582 & \multicolumn{1}{c|}{0.014} \\
\multicolumn{1}{|c|}{$w=$ 70-75}    &       &       & \multicolumn{1}{c|}{0.043} \\ \hline\hline
\multicolumn{4}{|c|}{$d=0$, $v=$ 65-70} \\ \hline
\multicolumn{1}{|c|}{}    & $w'=$ 65-70   & $w'=$ 70-75   & \multicolumn{1}{c|}{$w'=$ above 75} \\ \hline
\multicolumn{1}{|c|}{$w=$ below 65} & 0.108 & 0.015 & \multicolumn{1}{c|}{0.000} \\
\multicolumn{1}{|c|}{$w=$ 65-70}    &       & 0.358 & \multicolumn{1}{c|}{0.017} \\
\multicolumn{1}{|c|}{$w=$ 70-75}    &       &       & \multicolumn{1}{c|}{0.120} \\ \hline\hline
\multicolumn{4}{|c|}{$d=0$, $v=$ 70-75} \\ \hline
\multicolumn{1}{|c|}{}    & $w'=$ 65-70   & $w'=$ 70-75   & \multicolumn{1}{c|}{$w'=$ above 75} \\ \hline
\multicolumn{1}{|c|}{$w=$ below 65} & 0.163 & 0.021 & \multicolumn{1}{c|}{0.000} \\
\multicolumn{1}{|c|}{$w=$ 65-70}    &       & 0.299 & \multicolumn{1}{c|}{0.013} \\
\multicolumn{1}{|c|}{$w=$ 70-75}    &       &       & \multicolumn{1}{c|}{0.118} \\ \hline\hline
\multicolumn{4}{|c|}{$d=0$, $v=$ above 75} \\ \hline
\multicolumn{1}{|c|}{}    & $w'=$ 65-70   & $w'=$ 70-75   & \multicolumn{1}{c|}{$w'=$ above 75} \\ \hline
\multicolumn{1}{|c|}{$w=$ below 65} & 0.102 & 0.010 & \multicolumn{1}{c|}{0.000} \\
\multicolumn{1}{|c|}{$w=$ 65-70}    &       & 0.293 & \multicolumn{1}{c|}{0.022} \\
\multicolumn{1}{|c|}{$w=$ 70-75}    &       &       & \multicolumn{1}{c|}{0.189} \\ \hline\hline
\multicolumn{4}{|c|}{$d=1$, $v=$ below 65} \\ \hline
\multicolumn{1}{|c|}{}    & $w'=$ 65-70   & $w'=$ 70-75   & \multicolumn{1}{c|}{$w'=$ above 75} \\ \hline
\multicolumn{1}{|c|}{$w=$ below 65} & 0.840 & 0.759 & \multicolumn{1}{c|}{0.073} \\
\multicolumn{1}{|c|}{$w=$ 65-70}    &       & 0.892 & \multicolumn{1}{c|}{0.011} \\
\multicolumn{1}{|c|}{$w=$ 70-75}    &       &       & \multicolumn{1}{c|}{0.004} \\ \hline\hline
\multicolumn{4}{|c|}{$d=1$, $v=$ 65-70} \\ \hline
\multicolumn{1}{|c|}{}    & $w'=$ 65-70   & $w'=$ 70-75   & \multicolumn{1}{c|}{$w'=$ above 75} \\ \hline
\multicolumn{1}{|c|}{$w=$ below 65} & 0.990 & 0.865 & \multicolumn{1}{c|}{0.013} \\
\multicolumn{1}{|c|}{$w=$ 65-70}    &       & 0.847 & \multicolumn{1}{c|}{0.002} \\
\multicolumn{1}{|c|}{$w=$ 70-75}    &       &       & \multicolumn{1}{c|}{0.001} \\ \hline\hline
\multicolumn{4}{|c|}{$d=1$, $v=$ 70-75} \\ \hline
\multicolumn{1}{|c|}{}    & $w'=$ 65-70   & $w'=$ 70-75   & \multicolumn{1}{c|}{$w'=$ above 75} \\ \hline
\multicolumn{1}{|c|}{$w=$ below 65} & 0.914 & 0.870 & \multicolumn{1}{c|}{0.056} \\
\multicolumn{1}{|c|}{$w=$ 65-70}    &       & 0.949 & \multicolumn{1}{c|}{0.012} \\
\multicolumn{1}{|c|}{$w=$ 70-75}    &       &       & \multicolumn{1}{c|}{0.005} \\ \hline\hline
\multicolumn{4}{|c|}{$d=1$, $v=$ above 75} \\ \hline
\multicolumn{1}{|c|}{}    & $w'=$ 65-70   & $w'=$ 70-75   & \multicolumn{1}{c|}{$w'=$ above 75} \\ \hline
\multicolumn{1}{|c|}{$w=$ below 65} & 0.955 & 0.892 & \multicolumn{1}{c|}{0.024} \\
\multicolumn{1}{|c|}{$w=$ 65-70}    &       & 0.926 & \multicolumn{1}{c|}{0.004} \\
\multicolumn{1}{|c|}{$w=$ 70-75}    &       &       & \multicolumn{1}{c|}{0.001} \\ \hline
\end{tabular}
\caption{Unadjusted asymptotically-valid p-values for the individual weak  null hypotheses $\mathcal{H}_{0,(v,d)}^{(w,w')}$ of whether the raw score satisfies demographic insensitivity in the weak sense for age at the threshold $t=26$.}
\label{tab:pValues_FDP_X}
\end{table}
\subsubsection{Age-corrected MMSE scores} \label{sec:App_Oasis_FDP_corrected}
We now turn to age-corrected scores.  We consider the z-score correction \eqref{eq:AgeEduCorrection}, with means and standard deviations estimated  from the OASIS-3 dataset. To ensure independence of corrected scores across individuals so that our inference remains valid, we performed a 30\%-70\% random data split \cite{cox1975note}. We then used the smaller fold to compute the means and standard deviations among the non-impaired, and then used these to correct the raw scores of the individuals in the larger fold. It was this latter data that we used for our analysis. We target the probabilities describing whether the age-corrected score satisfies demographic insensitivity in the weak sense for age at a threshold of $t=-1.5$, commonly used in practice \cite[Tab. 1]{Bradfield2020-ae}. Analogous reasoning as  for the study of the raw score (Appendix \ref{sec:App_Oasis_FDP_raw}) leads to the individual unadjusted p-values found in Table \ref{tab:pValues_FDP_Z}. In this case, we do not reject the global null hypothesis of demographic insensitivity at a 5\% level, meaning we find no significant evidence indicating that the corrected score is sensitive to the age-quartiles; i.e., we find no evidence against the fact that the age-correction achieved demographic insensitivity.

\begin{table}
\centering
\begin{tabular}{|cccc|}
\hline
\multicolumn{4}{|c|}{$d=0$, $v=$ below 65} \\ \hline
\multicolumn{1}{|c|}{}    & $w'=$ 65-70   & $w'=$ 70-75   & \multicolumn{1}{c|}{$w'=$ above 75} \\ \hline
\multicolumn{1}{|c|}{$w=$ below 65} & 0.496 & 0.010 & \multicolumn{1}{c|}{0.003} \\
\multicolumn{1}{|c|}{$w=$ 65-70}    &       & 0.082 & \multicolumn{1}{c|}{0.031} \\
\multicolumn{1}{|c|}{$w=$ 70-75}    &       &       & \multicolumn{1}{c|}{0.556} \\ \hline\hline
\multicolumn{4}{|c|}{$d=0$, $v=$ 65-70} \\ \hline
\multicolumn{1}{|c|}{}    & $w'=$ 65-70   & $w'=$ 70-75   & \multicolumn{1}{c|}{$w'=$ above 75} \\ \hline
\multicolumn{1}{|c|}{$w=$ below 65} & 0.742 & 0.026 & \multicolumn{1}{c|}{0.008} \\
\multicolumn{1}{|c|}{$w=$ 65-70}    &       & 0.060 & \multicolumn{1}{c|}{0.022} \\
\multicolumn{1}{|c|}{$w=$ 70-75}    &       &       & \multicolumn{1}{c|}{0.535} \\ \hline\hline
\multicolumn{4}{|c|}{$d=0$, $v=$ 70-75} \\ \hline
\multicolumn{1}{|c|}{}    & $w'=$ 65-70   & $w'=$ 70-75   & \multicolumn{1}{c|}{$w'=$ above 75} \\ \hline
\multicolumn{1}{|c|}{$w=$ below 65} & 0.790 & 0.034 & \multicolumn{1}{c|}{0.025} \\
\multicolumn{1}{|c|}{$w=$ 65-70}    &       & 0.060 & \multicolumn{1}{c|}{0.047} \\
\multicolumn{1}{|c|}{$w=$ 70-75}    &       &       & \multicolumn{1}{c|}{0.779} \\ \hline\hline
\multicolumn{4}{|c|}{$d=0$, $v=$ above 75} \\ \hline
\multicolumn{1}{|c|}{}    & $w'=$ 65-70   & $w'=$ 70-75   & \multicolumn{1}{c|}{$w'=$ above 75} \\ \hline
\multicolumn{1}{|c|}{$w=$ below 65} & 0.772 & 0.031 & \multicolumn{1}{c|}{0.024} \\
\multicolumn{1}{|c|}{$w=$ 65-70}    &       & 0.056 & \multicolumn{1}{c|}{0.045} \\
\multicolumn{1}{|c|}{$w=$ 70-75}    &       &       & \multicolumn{1}{c|}{0.783} \\ \hline\hline
\multicolumn{4}{|c|}{$d=1$, $v=$ below 65} \\ \hline
\multicolumn{1}{|c|}{}    & $w'=$ 65-70   & $w'=$ 70-75   & \multicolumn{1}{c|}{$w'=$ above 75} \\ \hline
\multicolumn{1}{|c|}{$w=$ below 65} & 0.516 & 0.074 & \multicolumn{1}{c|}{0.072} \\
\multicolumn{1}{|c|}{$w=$ 65-70}    &       & 0.187 & \multicolumn{1}{c|}{0.175} \\
\multicolumn{1}{|c|}{$w=$ 70-75}    &       &       & \multicolumn{1}{c|}{0.803} \\ \hline\hline
\multicolumn{4}{|c|}{$d=1$, $v=$ 65-70} \\ \hline
\multicolumn{1}{|c|}{}    & $w'=$ 65-70   & $w'=$ 70-75   & \multicolumn{1}{c|}{$w'=$ above 75} \\ \hline
\multicolumn{1}{|c|}{$w=$ below 65} & 0.382 & 0.147 & \multicolumn{1}{c|}{0.062} \\
\multicolumn{1}{|c|}{$w=$ 65-70}    &       & 0.461 & \multicolumn{1}{c|}{0.217} \\
\multicolumn{1}{|c|}{$w=$ 70-75}    &       &       & \multicolumn{1}{c|}{0.709} \\ \hline\hline
\multicolumn{4}{|c|}{$d=1$, $v=$ 70-75} \\ \hline
\multicolumn{1}{|c|}{}    & $w'=$ 65-70   & $w'=$ 70-75   & \multicolumn{1}{c|}{$w'=$ above 75} \\ \hline
\multicolumn{1}{|c|}{$w=$ below 65} & 0.420 & 0.051 & \multicolumn{1}{c|}{0.057} \\
\multicolumn{1}{|c|}{$w=$ 65-70}    &       & 0.175 & \multicolumn{1}{c|}{0.192} \\
\multicolumn{1}{|c|}{$w=$ 70-75}    &       &       & \multicolumn{1}{c|}{0.769} \\ \hline\hline
\multicolumn{4}{|c|}{$d=1$, $v=$ above 75} \\ \hline
\multicolumn{1}{|c|}{}    & $w'=$ 65-70   & $w'=$ 70-75   & \multicolumn{1}{c|}{$w'=$ above 75} \\ \hline
\multicolumn{1}{|c|}{$w=$ below 65} & 0.397 & 0.086 & \multicolumn{1}{c|}{0.066} \\
\multicolumn{1}{|c|}{$w=$ 65-70}    &       & 0.277 & \multicolumn{1}{c|}{0.219} \\
\multicolumn{1}{|c|}{$w=$ 70-75}    &       &       & \multicolumn{1}{c|}{0.967} \\ \hline
\end{tabular}
\caption{Unadjusted asymptotically-valid p-values for the individual weak  null hypotheses $\mathcal{H}_{0,(v,d)}^{(w,w')}$ of whether the age-corrected score satisfies demographic insensitivity in the weak sense for age at the threshold $t=-1.5$.}
\label{tab:pValues_FDP_Z}
\end{table}

\end{document}